\pdfoutput=1
\documentclass[prd,reprint,onecolumn,eqsecnum,floats,aps,amsmath,amssymb,nofootinbib,superscriptaddress,showpacs,longbibliography]{revtex4-2}

\usepackage[utf8]{inputenc}
\usepackage[T1]{fontenc}
\usepackage[english]{babel}
\usepackage{BOONDOX-cal}

\usepackage{amsmath,amsthm,amssymb,amsfonts}
\usepackage{enumerate} 
\usepackage{wrapfig}
\usepackage{verbatim}
\usepackage{graphicx}
\usepackage{xfrac}
\usepackage{textcomp}
\usepackage{bbold}
\usepackage{lmodern}
\usepackage{microtype}

\newtheorem{thr}{Theorem}

\numberwithin{equation}{section}
\numberwithin{thr}{section}
\numberwithin{chr}{section}
\numberwithin{df}{section}

\newcommand{\ket}[1]{\lvert\, #1\,\rangle}

\newcommand{\su}{\text{su(2)}}
\newcommand{\SU}{\text{SU(2)}}
\newcommand{\SUq}{\text{SU}_q\text{(2)}}
\newcommand{\hatE}{\widehat{E}}
\newcommand{\Eu}{E^{(\text{u})}}
\newcommand{\hatEu}{\widehat{E}^{(\text{u})}}
\newcommand{\Ed}{E^{(\text{d})}}
\newcommand{\hatEd}{\widehat{E}^{(\text{d})}}
\newcommand{\Eud}{E^{\text{u+d}}}
\newcommand{\hatEud}{\widehat{E}^{(\text{u+d})}}
\newcommand{\ju}{j^{\text{u}}}
\newcommand{\jd}{j^{\text{d}}}

\newcommand{\jud}{j^{\text{u}+\text{d}}}

\newcommand{\Hp}{{\mathcal{H}(p)}}
\newcommand{\Hpjujd}{{\mathcal{H}^{(\ju,\jd)}(p)}}
\newcommand{\Hpjj}{{\mathcal{H}^{(j,j)}(p)}}
\newcommand{\Hpju}{{\mathcal{H}^{(j,0)}(p)}}
\newcommand{\Hpjd}{{\mathcal{H}^{(0,j)}(p)}}

\newcommand{\Hej}{\Hpju}
\newcommand{\HEj}{\mathcal{H}^{(j)}_{e'}}
\newcommand{\Heju}{\mathcal{H}^{(\ju,0)}}
\newcommand{\Hejd}{\mathcal{H}^{(0,\jd)}}
\newcommand{\Hejej}{\Hpjujd}
\newcommand{\Heej}{\Hpjj}
\newcommand{\pij}[2]{#2}

\newcommand{\W}{\widehat{\mathcal{W}}}

\renewcommand{\d}{\text{d}}
\renewcommand{\i}{\romannumeral 1}
\newcommand{\ii}{\romannumeral 2}
\newcommand{\floor}[1]{\lfloor #1 \rfloor}
\newcommand{\abs}[1]{\lvert #1 \rvert}

\newcommand{\snsUnchanged}{\parbox[m]{1.2cm}{\includegraphics[width=1.2cm]{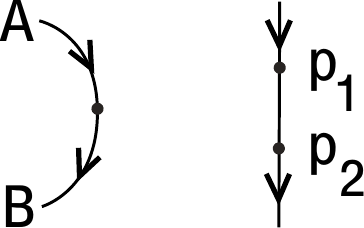}}}
\newcommand{\snsUpper}{\parbox[m]{1.2cm}{\includegraphics[width=1.2cm]{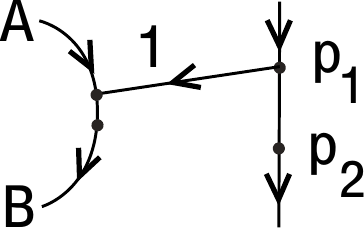}}}
\newcommand{\snsLower}{\parbox[m]{1.2cm}{\includegraphics[width=1.2cm]{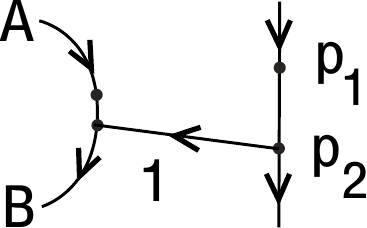}}}
\newcommand{\snsBoth}{\parbox[m]{1.2cm}{\includegraphics[width=1.2cm]{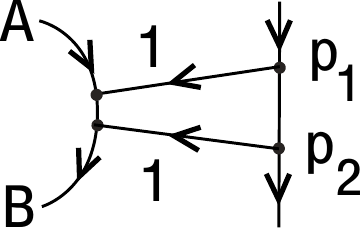}}}
\newcommand{\snsSingleDeformed}{\parbox[m]{1.2cm}{\includegraphics[width=1.2cm]{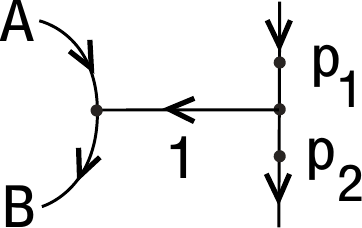}}}
\newcommand{\snsBothDeformed}{\parbox[m]{1.2cm}{\includegraphics[width=1.2cm]{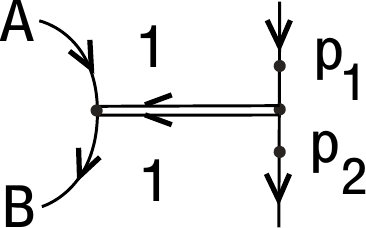}}}

\newcommand{\ABhj}{\parbox[m]{1.0cm}{\includegraphics[width=1.0cm]{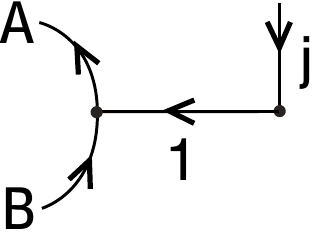}}}
\newcommand{\CBhj}{\parbox[m]{1.0cm}{\includegraphics[width=1.0cm]{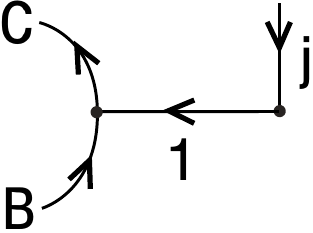}}}
\newcommand{\ABhjj}{\parbox[m]{1.0cm}{\includegraphics[width=1.0cm]{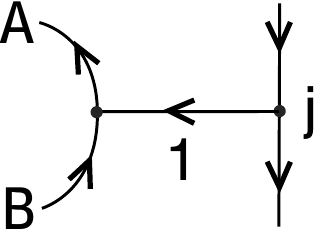}}}
\newcommand{\ABhhj}{\parbox[m]{1.0cm}{\includegraphics[width=1.0cm]{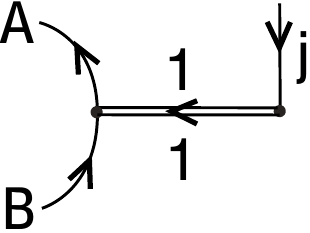}}}
\newcommand{\ABhhjj}{\parbox[m]{1.0cm}{\includegraphics[width=1.0cm]{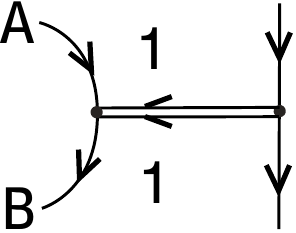}}}
\newcommand{\jstate}{\parbox[t]{0.15cm}{\includegraphics[width=0.15cm]{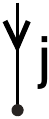}}}
\newcommand{\jjstate}{\parbox[m]{0.15cm}{\includegraphics[width=0.15cm]{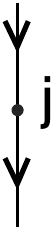}}}

\DeclareMathOperator{\id}{id}

\DeclareMathOperator{\one}{\mathbb{1}}
\DeclareMathOperator{\tr}{tr}
\DeclareMathOperator{\re}{Re}

\begin{document}

\title{Quantum surface holonomies for loop quantum gravity and their application to black hole horizons}

\author{Hanno Sahlmann}
\email{hanno.sahlmann@gravity.fau.de}
\author{Thomas Zilker}
\email{thomas.zilker@gravity.fau.de}
\affiliation{Friedrich-Alexander-Universit\"at  Erlangen-N\"urnberg (FAU)\\ Institute for Quantum Gravity, Staudtstraße 7, 91058 Erlangen, Germany}

\begin{abstract} 
In this work we define a new type of flux operators on the Hilbert space of loop quantum gravity.  We use them to solve an equation of the form $F(A)=c\,\Sigma$ in loop quantum gravity. This equation, which relates the curvature of a connection $A$ with its canonical conjugate $\Sigma=*E$, plays an important role for spherically symmetric isolated horizons, and, more generally, for maximally symmetric geometries and for the Kodama state. If the equation holds, the new flux operators can be interpreted as a quantization of surface holonomies from higher gauge theory. Also, they represent a kind of quantum deformation of SU(2). We investigate their properties and discuss how they can be used to define states that satisfy the isolated horizon boundary condition in the quantum theory. 
\end{abstract} 
      
\maketitle


\tableofcontents

\section{Introduction} 
\label{se_intro}
The classical boundary conditions on a spatial slice $\mathcal{H}$ of a spherically symmetric isolated horizon can be expressed \cite{Ashtekar:1997yu} by the very natural boundary condition \cite{Smolin:1995vq} 
\begin{equation}
\label{eq:ih_bc}
\iota_{\mathcal{H}}^{*}F(A)=C\, \iota_{\mathcal{H}}^{*} \,(*E) \, . 
\end{equation}
Here, $A$ and $E$ are canonically conjugate Ashtekar-Barbero variables \cite{Ashtekar:1986yd, Barbero:1994ap}, an \SU~connection and the corresponding electric field. We will take $E$ to be $\su$-valued, using the Cartan-Killing metric on $\su$.  $*E$ denotes the 2-form 
\begin{equation}
(*E)_{ab}= \epsilon_{abc}E^{c} \, .
\end{equation}
Before the invention of isolated horizons, a boundary condition of the form \eqref{eq:ih_bc}  has already been studied in \cite{Smolin:1995vq}. In that prescient work, Smolin has argued that the imposition of \eqref{eq:ih_bc} in the quantum theory leaves a quantized Chern-Simons theory on the boundary, with defects at the locations where quantized gravitational excitations of the bulk touch the boundary. This picture is the foundation of all later work on the entropy of isolated horizons. In the present work, we will investigate in how far the picture of \cite{Smolin:1995vq} can be derived from an operator version of  \eqref{eq:ih_bc} in the quantum theory. 

In loop quantum gravity (LQG), there exists a well-defined operator for the parallel transport induced by $A$, but $A$ itself, and by extension its curvature $F$, are not well-defined in the quantum theory. 
If one rewrites \eqref{eq:ih_bc} in terms of holonomies of $A$, what objects will one deal with in terms of $E$? And how can one implement  \eqref{eq:ih_bc} in LQG? It is important to answer these questions if one wants to solve the boundary conditions \eqref{eq:ih_bc} from within the formalism of LQG \cite{Sahlmann:2011xu}. 

It is interesting to note that equations of the form
\begin{equation}
\label{eq:sd3}
F(A)=C\, (*E)
\end{equation}
also play a role in different contexts. An equation very similar to \eqref{eq:ih_bc} is part of a condition for spherical symmetry \cite{Beetle:2016brg}. In that case the curvature is that of a related connection -- the spin connection $\Gamma$. 
Also the equation shows up in calculations of quantum gravity amplitudes \cite{Haggard:2014xoa,Haggard:2015yda,Haggard:2015kew}, in an LQG treatment of Chern Simons theory \cite{Sahlmann:2010bd,Sahlmann:2011uh,Sahlmann:2011rv}, and in the context of the Kodama state for LQG \cite{Kodama:1990sc,Ashtekar:1988sw,Smolin:2002sz,Bodendorfer:2016pxg}. In these cases, techniques to implement  \eqref{eq:sd3} might be useful.

One can use a non-Abelian generalization of Stokes' theorem \cite{Arefeva:1980} to obtain a holonomy around the boundary $\partial S$ of a simply connected surface as a function of its curvature $F(A)$:
\begin{equation}
\label{eq:s-exp}
h_{\partial S}= \mathcal{S}\exp\int_S -\mathcal{F} \, .
\end{equation}
This is a surface-ordered exponential integral, a higher dimensional analogue of the path-ordered exponential integral expressing the holonomy  as a function of $A$ on a curve. $\mathcal{F}$ is a suitable parallel transport of $F(A)$.
\eqref{eq:ih_bc} then implies that on a spherically symmetric horizon, the holonomy can similarly be expressed as 
\begin{equation}\label{eq:s-expE}
\mathcal{W}_S:= \mathcal{S}\exp\int_S -C\mathcal{E}. 
\end{equation}
Here and in the following, pullbacks to the horizon are assumed, but not written explicitly.

One can then impose \eqref{eq:ih_bc} in LQG by looking for states $\Psi$ such that 
\begin{equation}
\widehat{\mathcal{W}_S} \Psi = \widehat{h_{\partial S}} \Psi 
\end{equation}
for surfaces $S$ on the horizon. 

We must mention that in the remarkable article \cite{Bodendorfer:2016pxg} Bodendorfer suggests a route to solving \eqref{eq:sd3} that is different from what we propose here. He points out that by modifying the canonical momentum according to $E\mapsto E+*F$, one can regard the Ashtekar-Lewandowski vacuum as a solution of \eqref{eq:sd3}. The advantage of that method is that it is very clean and straightforward. However, functions of $E$ can then not be quantized straightforwardly. Still, \cite{Bodendorfer:2016pxg} contains suggestions for volume and for the Hamiltonian constraint. Our method works with a Hilbert space in which $E$ is still represented straightforwardly. The disadvantage is that it is not straightforward to identify solutions of \eqref{eq:sd3}. We also note that \cite{Bodendorfer:2016pxg} contains an important discussion of the question in how far \eqref{eq:ih_bc} is related to the symmetry of the horizon. We note that \cite{Bodendorfer:2016pxg} makes the argument that \eqref{eq:ih_bc} holds entirely due to symmetry. 

To understand the properties of $\widehat{\mathcal{W}_S}$, it is important to realize that \eqref{eq:ih_bc} and \eqref{eq:s-expE} have a deeper mathematical meaning in the framework of higher gauge theory. This is a formalism which categorically extends the notions of gauge theory. In particular, it defines higher gauge fields and corresponding notions of parallel transport along higher dimensional objects. In this context, \eqref{eq:ih_bc} is just the statement that $A$ and $E$ together define a 2-connection, and \eqref{eq:s-expE} is the parallel transport across a surface $S$. These aspects of the problem are explained in the companion paper \cite{Zilker:2017aey}. They naturally explain the reparametrization independence and other properties of \eqref{eq:s-expE}. 

The quantization of \eqref{eq:s-expE} adds another layer of complexity and is explored in the present work. In LQG, the components of the field $E$ are somewhat singular operators, and they do not commute in the quantum theory. Therefore \eqref{eq:s-expE} presents a host of problems when trying to transfer it to the quantum theory. The non-commutativity is of the type of an $\SU$ current algebra, 
\begin{equation}
[\widehat{E}_i(x),\widehat{E}_j(y)]=\delta_{x,y} \,f_{ij}{}^k \, \widehat{E}_k(x) \, ,
\end{equation}
where $f_{ij}{}^k$ are the structure constants of $\SU$. One can use the fact that it derives from a symplectic structure on $\su^*$ to quantize the surface holonomies $\mathcal{W}_S$ \eqref{eq:s-expE} using the Duflo-Kirillov map \cite{Duflo:1977}. This gives the object special properties \cite{Sahlmann:2011rv, Sahlmann:2015dna}. In  \cite{Sahlmann:2015dna}, the action of  $\widehat{\mathcal{W}_S}$ was determined only on special states. The first result of this work is the extension of the action of this operator to a large class of LQG states. In particular, we are investigating the action on edges carrying arbitrary spin, and we are carefully defining the action at vertices. The latter is important when considering repeated application of surface holonomy operators.  

At the core of the quantization of $\mathcal{W}_S$ is the application of the Duflo-Kirillov map to a function of the form 
\begin{equation}
W= \exp(E^i T_i)
\end{equation}
with $T_I$ a basis of $\su$, and 
\begin{equation}
\{E_i,E_j\}=f_{ij}{}^k E_k. 
\end{equation}
In other words, we are looking for the Duflo-Kirillov quantization of the exponential map. The resulting object, and by extension the quantum surface holonomies $\widehat{\mathcal{W}_S}$ are operator valued matrices with non-commuting entries, 
\begin{equation}
\widehat{W}
= \begin{pmatrix}
\widehat{a}&\widehat{b}\\
-\widehat{b}^\dagger &\widehat{a}^\dagger
\end{pmatrix}.
\end{equation}
We analyze their properties and show that they still retain many properties of $\SU$ group elements. Thus, we are dealing with a kind of quantum deformation of $\SU$. The eigenvalues of traces of $\widehat{\mathcal{W}_S}$ can be expressed in terms of quantum integers, but the commutation relations between the components seem to be of a different kind than the ones described by an R-matrix. This is the second set of results of the present work. 

Coming back to the physics aspects, in the last part of the article we start to analyze what kind of states fulfill the quantum version of the isolated horizon boundary condition \eqref{eq:ih_bc}. We find that a relevant operator seems to be the determinant of $\widehat{\mathcal{W}_S}$ on the horizon. In general it is not equal to $1$, meaning that, according to \eqref{eq:s-exp}, also the holonomies must have quite non-classical properties on the horizon. However, in the holonomy-flux algebra of LQG, the holonomies $h$ all fullfill $\det h=1$. One option is thus to reject states on which $\det \widehat{\mathcal{W}_S}\neq 1$ on the basis that the quantum version of \eqref{eq:ih_bc} cannot be fulfilled. The other option is to \emph{define} the holonomies on the horizon by the $\widehat{\mathcal{W}_S}$. We consider the implications of this identification for very simple states with only two punctures and find that again $\det \widehat{\mathcal{W}_S}$ is relevant for the question whether a state can reasonably be said to solve the IH boundary conditions.

\section{Surface holonomies and the isolated horizon boundary condition}
\label{sec:ClassSurfHol}

In this chapter, we will explain the classical setting and introduce some of our conventions and notation (those related to the quantum theory will be introduced in the next chapter).\\
As already mentioned in the introduction, the basic variables used in LQG are not the Ashtekar-Barbero variables $A$ and $E$ directly, but rather certain smearings of those. For the connection $A$, these smearings are so-called holonomies, which are given explicitly by
\begin{equation}
\begin{split}
h_\alpha[A]&=\mathcal{P}\exp \left( -\int_\alpha A\right)\\
&=\one +\sum_{n=1}^\infty (-1)^n \int_0^1 \text{d}t_1\int_0^{t_1}\text{d}t_2\ldots \int_0^{t_{n-1}}\text{d}t_n  A_{a_1} (\alpha(t_1))\dot{\alpha}^{a_1}(t_1)\ldots A_{a_n} (\alpha(t_n))\dot{\alpha}^{a_n}(t_n) \, .
\end{split}
\end{equation}
Note that $\alpha(t)$ can be any parametrization of the path $\alpha$ and $h_{\alpha}$ will not depend on it. We now want to write down a similar formula for the surface-ordered exponential from equation \eqref{eq:s-expE}. However, in contrast to paths, two-dimensional surface are a priori not equipped with a natural order. In order to have a chance of defining the surface-ordered exponential, we would therefore need to add an ordering of the surface $S$ as an additional structure to the data on which the surface holonomy depends. For example, in \citep{Arefeva:1980} lexicographical ordering is used with respect to some given parametrization of the surface. However, instead of using an ordered surface as label for the surface holonomies, we will be guided by insights from higher gauge theory \citep{Baez:2005qu, Schreiber:2011, Schreiber:2013, Martins:2007uki, Martins:2008, Pfeiffer:2003je, Girelli:2003ev} (see also \cite{Baez:2010ya} for an excellent review). From the perspective of higher gauge theory, the isolated horizon boundary condition just states that, on the horizon surface $\mathcal{H}$, the LQG variables $A$ and $C(*E)$ form a 2-connection \citep{Zilker:2017aey}. The surface holonomies also show up in higher gauge theory, although their definition is rather abstract in this context. However, the main message from higher gauge theory is that surface holonomies are group elements that are actually not associated to surfaces, but to homotopies!\footnote{More precisely, they only depend on equivalence classes of homotopies with respect to thin homotopy. This property is analogous to the parametrization independence of ordinary (path) holonomies.}\\
Let us briefly recall the definition of a homotopy. Consider two paths $\alpha$ and $\beta$ with the same starting and end points. A homotopy $h:\alpha \Rightarrow \beta$ from $\alpha$ to $\beta$ is a continuous map 
\begin{equation}
h:[0,1]\times [0,1] \rightarrow \Sigma
\end{equation}
such that
\begin{align}
h(0,t)&=\alpha(t) & h(s,0)&=\alpha(0)=\beta(0)\\
h(1,t)&=\beta(t) & h(s,1)&=\alpha(1)=\beta(1) \, .
\end{align}
Homotopies can be composed in two distinct ways. Given homotopies  $h_1:\alpha_1 \Rightarrow \beta_1$, $h_2:\alpha_2 \Rightarrow \beta_2$ with $\alpha_2(0)=\alpha_1(1)$  and $\beta_2(0)=\beta_1(1)$, there is a natural composition called horizontal composition $\circ_h$ of 2-morphisms in the path 2-groupoid $\mathcal{P}_2(\Sigma)$ yielding a homotopy from $\alpha_2\circ\alpha_1\Rightarrow\beta_2\circ\beta_1$.  Explicitly, 
\begin{equation}
(h_2\circ_h h_1)(s,t)=\begin{cases} (\id_{h_2(0,0)}\circ h_1(2s,\cdot))(t) & \text{ for } s\in[0,\sfrac{1}{2}]\\
(h_2(2s-1,\cdot)\circ h_1(1,\cdot))(t) & \text{ for } s\in[\sfrac{1}{2},1]
\end{cases}\, .
\end{equation}
The second type of composition in $\mathcal{P}_2(\Sigma)$ called vertical composition, and it is defined for homotopies $h_1:\alpha_1 \Rightarrow \beta_1$ and $h_2:\alpha_2 \Rightarrow \beta_2$ if $\beta_1=\alpha_2$. In this case, vertical composition works just like path composition in the $s$-parameter of homotopies, i.e.
\begin{equation}
(h_2\circ_v h_1)(s,t)=\begin{cases}h_1(2s,t)& \text{ for } s\in[0,\sfrac{1}{2}]\\
h_2(2s-1,t)& \text{ for } s\in[\sfrac{1}{2},1]
\end{cases}\, .
\end{equation}
At this point, we could define abstract classical surface holonomies as 2-functors from the path 2-groupoid to a 2-group as is done in higher gauge theory. On the level of 2-morphisms, these associate group elements to equivalence classes of homotopies. However, we want to give an explicit formula for those surface holonomies and, in order for this formula to be well-defined, we need the homotopies to satisfy certain additional requirements.
For every homotopy $H$, we define a corresponding surface $S_H$ as the interior of the image of $H$. In order for the surface-ordered exponential integral over these $S_H$ to be well-defined, they need to be equipped with an order. If we assume the homotopies $H$ to be one-to-one, they will induce a surface ordering by choosing lexicographical ordering on the parameter space $[0,1]\times [0,1]$. Note that the one-to-one assumption can be violated on measure-zero sets without changing the value of the integral.\footnote{Since we consider only homotopies with fixed end points, the homotopies themselves can actually never be one-to-one maps. However, the only problematic points in this regard are the end points of the paths $H_{s}(t) := H(s,t)$ and they definitely form a subet of measure zero.} We will also require our homotopies to be differentiable because we want to use them as parametrizations for the surfaces $S_H$ in the following.
Now, given a homotopy $H$, we define canonical paths $\alpha_x$ from $x_0 := H(0,1)$ to any point $x = H(s_x,t_x)$ in the surface $S_H$ via
\begin{equation}
\alpha_x (t) =H(s_x,1 - (1-t_x)t) \, ,
\end{equation}
and for every 2-form $B$ we introduce the notation
\begin{equation}
\mathcal{B}(x)= h^{-1}_{\alpha_x} B(x)h_{\alpha_x}^{\vphantom{-1}}
\end{equation}
that has already been used in the introduction. This allows us to write the surface-ordered exponential as
\begin{equation}
\begin{split}
\mathcal{W}_H[A,B]&= \mathcal{S}\exp \left( -\int_{S_H} \mathcal{B} \right)\\
&:=\one + \sum_{n=1}^\infty (-1)^n \underset{\small\begin{array}{c} S_H\times \ldots \times S_H\\ p_1\geq\ldots\geq p_n\end{array}}{\int\ldots \int}  \mathcal{B}(p_1)\ldots\mathcal{B}(p_n) \\
&=\one + \sum_{n=1}^\infty (-1)^n \int_0^1 \text{d}s_1 \int_0^1 \text{d}t_1 \int_0^{s_1}\text{d}s_2  \int_0^1 \text{d}t_2\ldots \\
&\qquad\qquad\qquad\qquad\ldots \int_0^{s_{n-1}}\text{d}s_n  \int_0^1 \text{d}t_n\, (\mathcal{B}_{a_1b_1} H^{a_1}_{,s}H^{b_1}_{,t})(s_1,t_1)\ldots(\mathcal{B}_{a_nb_n} H^{a_n}_{,s}H^{b_n}_{,t})(s_n,t_n) \, .
\end{split}
\label{eqn:surf_hol_explicit_formula}
\end{equation}
In the last line, we have used the homotopy $H$ as parametrization for the surface $S_H$ and we have ignored the ordering in the $t$-parameter since this is only relevant on subsets of measure zero. This surface-ordered integral was first defined in \cite{Arefeva:1980}, where it was used to prove a non-abelian version of Stokes' theorem. In our notation, the non-abelian Stokes' theorem can be written as
\begin{equation}
\mathcal{W}_H[A,F(A)]=h_{H(1,\cdot)}[A] \, ,
\end{equation}
where $H$ is assumed to be a homotopy from the constant path $\id_{x_0}$ to the path given by the boundary $\partial S_H$ which starts and ends at $x_0 \in \partial S_H$. From this point onward, we will always consider homotopies to be of this type. This will ensure that any two homotopies can be horizontally composed, if they have the same distinguished point $x_0$. Furthermore, the resulting homotopy will again be of this form with the same distinguished point.\\

Let us now have a look at the boundary condition for spherically spherically symmetric isolated horizons
\begin{equation}
\iota^{~\,a}_{\mathcal{H},\alpha} \, \iota^{~\,b}_{\mathcal{H},\beta} F(A)_{ab}^i=C \, \iota_{\mathcal{H},\alpha}^{~\,a} \, \iota_{\mathcal{H},\beta}^{~\,b} \, \epsilon_{abc}\,\kappa^{ij} E^c_j \, ,
\label{eqn:IHBC}
\end{equation}
where $\iota_{\mathcal{H}}$ is an embedding of the two-dimensional intersection $\mathcal{H}$ of the isolated horizon and the spatial 3-manifold $\Sigma$ into the latter and
\begin{equation}
C=\frac{4\pi (1-\beta^2)}{a_\mathcal{H}} \, ,
\end{equation}
with $a_\mathcal{H}$ denoting the area of $\mathcal{H}$ \citep{Engle:2009vc}. Equation \eqref{eqn:IHBC} is the same condition that was already stated in the introduction as \eqref{eq:ih_bc}, but here we have explicitly written down all the indices involved. Applying the surface-ordered exponential integral on both sides leads us to
\begin{equation}
\mathcal{W}_H[A,C \, (*E)] = \mathcal{W}_H[A,F(A)]  = h_{H(1,\cdot)}[A] \, .
\label{eqn:expIHBC}
\end{equation}
The trace of this exponentiated and integrated condition has already been studied in \citep{Sahlmann:2011rv, Sahlmann:2011uh}. In a companion paper \citep{Zilker:2017aey}, we actually proof the following theorem:
\begin{thr}
The following are equivalent (using the notation introduced above): 
\begin{enumerate}[(i)]
\item $\iota^{~\,a}_{\mathcal{H},\alpha} \, \iota^{~\,b}_{\mathcal{H},\beta} F(A)_{ab}^i (x) = C\iota^{~\,a}_{\mathcal{H},\alpha} \, \iota^{~\,b}_{\mathcal{H},\beta} \epsilon_{abc}\kappa^{ij} E^c_j(x) \qquad \forall \text{ } x \in \mathcal{H}$\,.
\item $\mathcal{W}_H[A,C \, (*E)] = h_{H(1,\cdot)}[A] \qquad \forall \text{ homotopies } H\text{, s.t. } S_{H} \subset \mathcal{H}$\,.
\end{enumerate}
\end{thr}
We already mentioned in the introduction that there are well-defined quantum operators associated to path holonomies in LQG. The following chapters will thus be devoted to finding a quantization of the surface holonomies appearing in condition \eqref{eqn:expIHBC}, to analyzing the properties of those quantum surface holonomy operators and to solving the quantum version of \eqref{eqn:expIHBC} on the LQG Hilbert space.
\section{Quantization of surface holonomies} 
\label{se_quant}
The aim of this chapter is to define quantum operators for the surface holonomies from the previous chapter on the LQG Hilbert space. In order to do, we will first introduce some further notation from LQG. Let $\Psi_{\gamma}$ denote a spin network state associated to the graph $\gamma$. The action of the $E$-field on such a state can formally be written as
\begin{equation}
\widehat{E}^a_k(x)\Psi_\gamma=8\pi G\hbar \, \beta i\,  \sum_{e\in\gamma} e^a(x) \widehat{E}^{(e)}_k(x)\Psi_\gamma \, .
\label{eqn:E-field_distributional_action}
\end{equation}
Here, the factor $e^a(x)$ makes sure that the action of the operator is concentrated on the graph $\gamma$. It is explicitly given by
\begin{equation}
e^a(x)=\int \dot{e}^a(t) \,\delta^{(3)}(x,e(t))\, \d t \, .
\end{equation}
The $\widehat{E}^{(e)}_k(x)$ obey the commutation relation
\begin{equation}
[\widehat{E}_{i}^{(e)}(p) , \widehat{E}_{j}^{(e')}(p') ]=\delta_{e,e'}\delta_{p,p'} f_{ij}^{~~k}\widehat{E}_k^{(e)}(p) \, ,
\end{equation}
with $f_{ij}^{~~k}$ denoting the structure constants of \su ~in a specific basis $T_i$ satisfying
\begin{equation*}
[T_i,T_j]= f_{ij}{}^k\, T_k \, ,
\end{equation*}
and they act in the representation space associated to the corresponding edge $e$. Note that they behave like genuine $\su$ elements, i.e. without the additional factor $i$ that is typically used in physics when dealing with angular momentum operators.\\
As already indicated above, however, expression \eqref{eqn:E-field_distributional_action} is merely formal in the sense that $\widehat{E}^a_k(x)$ is not an operator but an operator-valued distribution. Therefore, an appropriate smearing is required and in LQG one usually considers the flux operators
\begin{equation}
\begin{split}
\widehat{E}_{S} := \int_S\widehat{E}^a_k(x)\epsilon_{abc}\, \d x^b \d x^c \, \Psi_\gamma &= 8\pi G\hbar \, \beta i\,  \sum_{p} \sum_{e \text{ at } p} \kappa(e,S) \widehat{E}_{k}^{(e)}(p) \, \Psi_\gamma \\
&= 8\pi G\hbar \, \beta i\,  \sum_{p} \left[ \hatEu_{k}(p) - \hatEd_{k}(p) \right]\, \Psi_\gamma \\
&=: 8\pi G\hbar \, \beta i\, \sum_p \widehat{E}_k(p)\, \Psi_\gamma \, ,
\end{split}
\label{eqn:LQG_flux}
\end{equation}
where the sum over $p$ runs over all punctures of the spin network graph $\gamma$ with the surface $S$ and
\begin{equation}
\label{eqn:kappa}
\kappa(e,S)=\begin{cases}
+1 & \text{ if $e$ lies above $S$ }\\
-1 & \text{ if $e$ lies below $S$}\\
0 & \text{ otherwise}
\end{cases}
\end{equation}
encodes the relative orientation of $S$ with respect to each edge $e$ in $\gamma$. In the last line of \eqref{eqn:LQG_flux}, we have defined
\begin{equation}
\widehat{E}_k(p) = \hatEu_{k}(p) - \hatEd_{k}(p)
\end{equation}
in terms of the operators
\begin{equation}
\hatEu_{k}(p)= \sum_{\small\begin{array}{c} e \text{ at } p\\ e \text{ above } S\end{array}} \widehat{E}_{k}^{(e)}(p)
\qquad \qquad \text{and} \qquad\qquad
\hatEd_{k}(p)= \sum_{\small\begin{array}{c} e \text{ at } p\\ e \text{ below } S\end{array}} \widehat{E}_{k}^{(e)}(p) \, ,
\end{equation}
which naturally showed up in the second line. Eventually, let us define
\begin{equation*}
\widehat{E}(p) := \kappa^{ij}\,T_i\widehat{E}_j(p) \, ,
\end{equation*}
where $\kappa^{ij}$ are the components of the inverse of the Cartan-Killing metric
\begin{equation*}
\kappa_{ij} = \tr \left( \operatorname{ad}_{T_i} \operatorname{ad}_{T_j} \right) \, .
\end{equation*}

We can now start evaluating the surface-ordered exponential as defined in \eqref{eqn:surf_hol_explicit_formula}. Consider a surface $S_H$ defined by a homotopy $H$, and  a fixed graph $\gamma$. Denote by $\mathcal{H}_\gamma$ the Hilbert space of cylindrical functions with respect to this graph and let $N$ be the number of punctures of $\gamma$ with $S_H$. The punctures $p_1,\ldots,p_N$ are labeled such that $p_1\leq\ldots\leq p_N$ with respect to the order on $S_H$ induced by $H$. Using
\begin{equation}
c:=-8\pi G\hbar \, \beta i\, C \, ,
\label{eqn:def_c}
\end{equation}
we then obtain
\begin{equation}
\begin{split}
\left.\widehat{\mathcal{W}}_H\right\rvert_{\mathcal{H}_\gamma}
&=\one + \sum_{n=1}^\infty c^{\,n} 
\underset{\small\begin{array}{c} S_H\times \ldots \times S_H\\ x_1\leq\ldots\leq x_n\end{array}}{\int\ldots \int} 
\left. (\ast\widehat{\mathcal{E}})(x_n)\ldots(\ast\widehat{\mathcal{E}})(x_1)\right\rvert_{\mathcal{H}_\gamma}\\
&= \one + \sum_{n=1}^{\infty}c^{\,n} \underset{k_1+\ldots + k_N=n}{\sum_{k_1, \ldots, k_N =  0}}
\frac{1}{k_1!\ldots k_N!}  
\left[h^{-1}_{\alpha_{p_N}} \widehat{E}(p_N) h_{\alpha_{p_N}}\right]^{k_N} \ldots
\left[h^{-1}_{\alpha_{p_1}} \widehat{E}(p_1) h_{\alpha_{p_1}}\right]^{k_1} \\
&= \one + \sum_{n=1}^{\infty}c^{\,n} \underset{k_1+\ldots + k_N=n}{\sum_{k_1, \ldots, k_N =  0}}
\frac{1}{k_1!\ldots k_N!}
\left(h^{-1}_{\alpha_{p_N}} T_{i_{n-k_N+1}}\ldots T_{i_{n}} h_{\alpha_{p_N}}\right) \ldots
\left(h^{-1}_{\alpha_{p_1}}T_{i_1}\ldots T_{i_{k_1}} h_{\alpha_{p_1}}\right)  \times \\
&\qquad\qquad\qquad \qquad\qquad\qquad \times 
\kappa^{i_1j_1}\ldots\kappa^{i_nj_n} 
\left[\widehat{E}_{j_{n-k_N+1}}(p_N)\ldots\widehat{E}_{j_{n}}(p_N) \right]
\ldots
\left[\widehat{E}_{j_1}(p_1)\ldots\widehat{E}_{j_{k_1}}(p_1) \right] \, .
\end{split}
\end{equation}
Obviously, the factors within each of the square brackets do not commute, which implies that there is an ordering ambiguity. Following \citep{Sahlmann:2011rv,Sahlmann:2011uh}, we will use the Duflo-Kirillov map $Q_{DK}$ to resolve this ambiguity. We will make this ordering choice explicit in the notation by writing
\begin{equation}
\label{eqn:surf_hol_Npunct}
\begin{split}
\left.\widehat{\mathcal{W}}_H\right\rvert_{\mathcal{H}_\gamma}
&= \one + \sum_{n=1}^{\infty}c^{\,n} \underset{k_1+\ldots + k_N=n}{\sum_{k_1, \ldots, k_N =  0}^{n}}
\frac{1}{k_1!\ldots k_N!} \left(h^{-1}_{\alpha_{p_N}} T_{i_{n-k_N+1}}\ldots T_{i_{n}} h_{\alpha_{p_N}}\right) \ldots 
\left(h^{-1}_{\alpha_{p_1}}T_{i_1}\ldots T_{i_{k_1}} h_{\alpha_{p_1}}\right)\times\\
&\qquad\qquad\qquad \qquad\qquad\qquad \times 
\kappa^{i_1j_1}\ldots\kappa^{i_nj_n} 
Q_{\text{DK}}\left[{E}_{j_{n-k_N+1}}(p_N)\ldots {E}_{j_{n}}(p_N) \right]
\ldots
Q_{\text{DK}}\left[{E}_{j_1}(p_1)\ldots{E}_{j_{k_1}}(p_1) \right] \, .
\end{split}
\end{equation}
Recall that
\begin{equation}
{E}_{k}(p) = \Eu_{k}(p) - \Ed_{k}(p) = \sum_{\small\begin{array}{c} e \text{ at } p\\ e \text{ above} S\end{array}} {E}_{k}^{(e)}(p) - \sum_{\small\begin{array}{c} e \text{ at } p\\ e \text{ below } S\end{array}} {E}_{k}^{(e)}(p) \, ,
\end{equation}
and while $\hatEu_{k}(p) = Q_{DK} \left( \Eu_{k}(p) \right)$, $\hatEd_{k}(p)$ and $\widehat{E}_{k}^{(e)}(p)$ all behave like \su ~elements, $\hatE_{k}(p)$ does not! Therefore, we will have to decide whether we consider $\Eu_{k}(p)$ and $\Ed_{k}(p)$ as basic quantities and only order these using the Duflo-Kirillov map or whether we apply $Q_{DK}$ to ${E}_{k}^{(e)}(p)$ for all $e$ independently. While the latter approach sounds more fundamental, the first option enables the explicit calculations in the next chapter and we will therefore stick to it throughout this paper.\\
Specializing to the case of a single puncture, equation \eqref{eqn:surf_hol_Npunct} becomes
\begin{equation}
\label{eqn:surfhol_singlepunct}
\begin{split}
\left.\widehat{\mathcal{W}}_H\right\rvert_{\mathcal{H}_\gamma}
&= \one + \sum_{n=1}^{\infty}c^{\,n} \frac{1}{n!} 
\left(h^{-1}_{\alpha_{p}}T_{i_1}\ldots T_{i_n} h_{\alpha_{p}}\right) 
\kappa^{i_1j_1}\ldots\kappa^{i_nj_n}Q_{\text{DK}}\left[{E}_{j_1}(p)\ldots{E}_{j_{n}}(p) \right]\\
&=h^{-1}_{\alpha_{p}}\, Q_{\text{DK}}\left[  \exp \left( c \, T_{i} \kappa^{ij} E_{j}(p) \right)  \right]  h_{\alpha_{p}}\\
&=: h^{-1}_{\alpha_{p}}\, Q_{\text{DK}}\left[  W_{p}  \right]  h_{\alpha_{p}} \, .
\end{split}
\end{equation}
In the last line, the notation $Q_{\text{DK}}\left[  W_{p}  \right]$ indicates that when the resulting operator acts on a spin network state, the result only depends on the edges that start or end at the puncture $p$. However, information about the surface $S_H$ is still present in the splitting $\widehat{E}_i = \hatEu_i - \hatEd_i$, where the co-normal to $S_H$ at $p$ determines which edges contribute to $\hatEu_i$  and $\hatEd_i$, respectively.\\
We can use these explicit formulas for the quantum surface holonomies to prove the following theorem:
\begin{thr}
\label{theorem:factorization}
Consider a graph $\gamma$, a homotopy $H$ and homotopies $H_1,\ldots, H_m$ such that  
\begin{equation}
\label{eqn:horiz_comp}
H=H_m\circ_h\ldots\circ_h H_1
\end{equation}
where $S_{H_i}$ is punctured by $\gamma$ at most once, and $\partial S_{H_i}\cap \gamma=\emptyset$. As mentioned before, we still assume all homotopies starting from the trivial path! Then 
\begin{equation}
\left.\widehat{\mathcal{W}}_H\right\rvert_{\mathcal{H}_\gamma}
=\left.\widehat{\mathcal{W}}_{H_m}\right\rvert_{\mathcal{H}_\gamma} \ldots
\left.\widehat{\mathcal{W}}_{H_1}\right\rvert_{\mathcal{H}_\gamma} \, .
\label{eqn:factorization}
\end{equation}
\end{thr}
\begin{proof}
We can assume without loss of generality that each $H_i$ contains precisely one puncture, because homotopies without puncture contribute just the identity operator, and therefore effectively reduce the number of homotopies in \eqref{eqn:horiz_comp}. With this assumption, every factor on the right hand side just takes the form \eqref{eqn:surfhol_singlepunct}. Multiplying them and sorting with respect to the number of Lie algebra generators, it is straightforward to see that this leads to \eqref{eqn:surf_hol_Npunct}. 
\end{proof}
This theorem allows us to express surface holonomies as products of surface holonomies acting on single punctures, provided we can find a suitable decomposition of the homotopy labeling the surface holonomy. In the following chapters, we will therefore focus our attention on the single puncture case. We will later come back to the case of multiple punctures again.

\newpage
\section{Explicit action of surface holonomy operators on single puncture states} 
\label{sec:quant}
\begin{wrapfigure}{r}{6cm}
\includegraphics[scale=0.55]{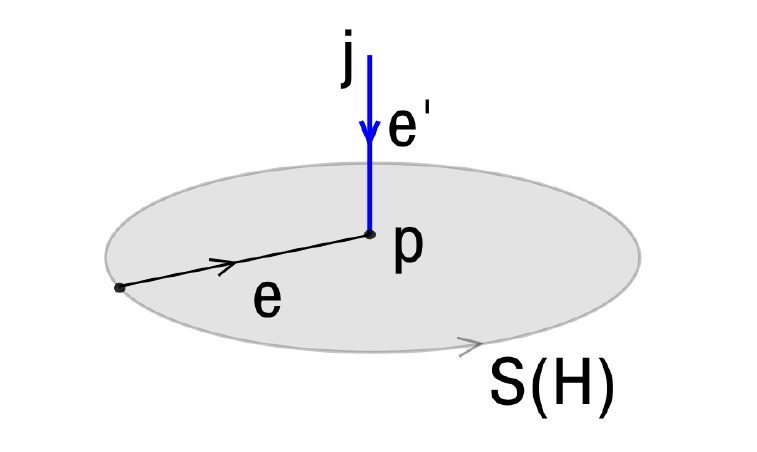}\newline
\includegraphics[scale=0.55]{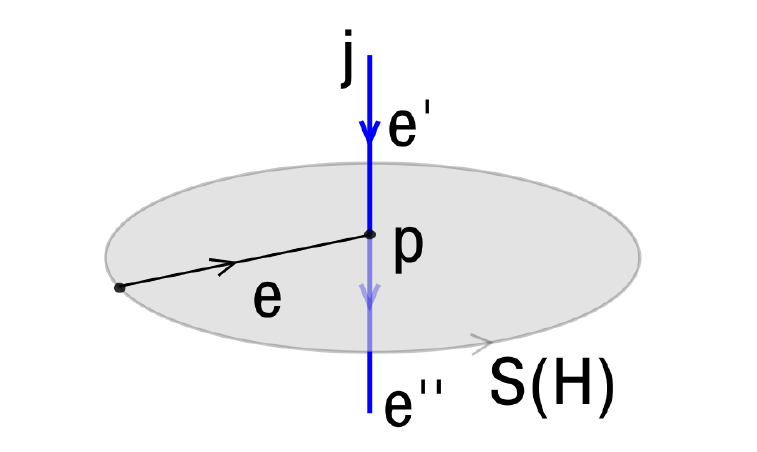}
\caption{The single puncture intersection of a holonomy with a surface}
\label{fig:single_puncture_case}
\end{wrapfigure}
In the following, we will explicitly calculate the action of the previously defined quantum surface holonomy operators on quantum states that are represented by a spin network graph having a single intersection with the surface associated to the homotopy labeling the surface holonomy (see figure~\ref{fig:single_puncture_case} for an illustration). To this end, let us introduce some further notation.\\
We first define the relevant Hilbert spaces. All of these are defined relative to a given homotopy $H$, but to keep things simple, we will not indicate this dependency in the notation. Let
\begin{equation}
\Hp = \text{ span } \{ \text{spin nets with single puncture at $p$} \}.    
\end{equation}
This space decomposes into a direct sum 
\begin{equation}
\Hp = \bigoplus_{\ju,\jd}  \Hpjujd
\end{equation}
under the action of $\widehat{E}(p)$ in the following sense: $\Hpjujd$ is an infinite direct sum of spaces on which 
\begin{equation*}
\widehat{E}(p)= \hatEu(p)- \hatEd(p)
\end{equation*}
acts irreducibly, with $\hatEu$ acting in the $\ju$-irrep of \su, and  $\hatEd$ in the $\jd$-irrep. Due to the additional holonomies in $\widehat{\mathcal{W}_H}$, its components mix these sub-sectors of $\Hpjujd$, but leave  $\Hpjujd$ invariant. 

Given $\ju,\jd$, we call $\Hpjujd$ the state space of a \emph{one-sided puncture} if either $\ju=0$ or $\jd=0$. Otherwise we call it the state space of a \emph{two-sided puncture}. We should remind the reader that the Duflo-quantization for the two-sided puncture in section \ref{ssec:two-edge-punct-calc} was calculated for a state in which $E^{(\jud)}=0$. In the quantum theory, this implies -- among other things -- that $\ju=\jd$ and that states should be in the gauge-invariant subspace of $\Hpjujd$. In the following, for the two-sided puncture we will therefore restrict to the space $\Hpjj$ in which $\ju=\jd$. We will sometimes also display the action on the non-gauge-invariant part of that space. 
\subsection{Action on one-sided puncture state}
In the case of a single puncture, the quantum operator associated to a surface holonomy was given in \eqref{eqn:surfhol_singlepunct}. Two of the three factors in this expression are path holonomies, whose action on the Hilbert space of LQG is well-understood. We will therefore focus on the remaining part, $Q_{\text{DK}}\left[  W_{p}  \right] $. In the following, we will explicitely calculate the action of this operator on a certain class of spin network states $\Psi_\gamma$ in the LQG Hilbert space. Namely, we will assume $\gamma$ to contain only a single edge that intersects $S_H$ at $p$. 
Without loss of generality, we can assume this edge to puncture the surface from above. This effectively leads to 
\begin{equation}
\widehat{E}_k(p)=\hatEu_k(p) \, ,
\end{equation}
and therefore $\widehat{E}_k(p)$ itself satisfies \su ~commutation relations. 
This case was already investigated in earlier work \cite{Sahlmann:2015dna}. However, in this earlier work we used a different convention for the $\kappa$ factor defined in equation \eqref{eqn:kappa}, which made the result appear more general. At the time, we were only able to give an explicit expression for the action of the surface holonomy operator on punctures carrying spin \textonehalf. In the following, we will now generalize this calculation to spin network punctures labeled by arbitrary spin $j$.\\

Recall, from the previous section, the definition 
\begin{equation}
W_{p}=\exp \left( c \, T_{i} \kappa^{ij} E_{j}(p) \right) \, .
\end{equation}
From now on, we will drop the label $p$ indicating the puncture. Throughout this section, the $E$ are understood to be evaluated at the puncture $p$! This actually implies that we will only consider non-trivial representations for the quantum operators corresponding to the $E$, since a puncture with spin label 0 is equivalent to no puncture in the LQG Hilbert space. Therefore, a quantum surface holonomy will always act as the identity operator on a puncture where $j=0$. We will also choose a specific basis 
\begin{equation}
T_i=\tau_i=-\frac{i}{2}\sigma_i
\end{equation}
of \su, where $\sigma_i$ are the Pauli matrices. In this basis, the components of the Cartan-Killing metric become
\begin{equation}
\kappa_{ij}=-2\,\delta_{ij} \, . 
\end{equation}
We can then write
\begin{equation}
\begin{split}
W_{p} &=  \operatorname{cosh} \left( \frac{c}{2\sqrt{2}} ||E|| \right)\; \mathbb{1}_{2} + \frac{\operatorname{sinh} \left( \frac{c}{2\sqrt{2}} ||E|| \right)}{\frac{c}{2\sqrt{2}} ||E||} \, c \, \kappa^{ij} E_{i} \; \tau_{j} \\
&= \sum_{k=0}^{\infty} \frac{1}{(2k)!} \left( \frac{c}{2\sqrt{2}} \right)^{2k}  ||E||^{2k} \; \mathbb{1}_{2} + c \, \sum_{k=0}^{\infty} \frac{1}{(2k+1)!} \left( \frac{c}{2\sqrt{2}} \right)^{2k} \kappa^{ij}  ||E||^{2k} E_{i} \; \tau_{j} \, .
\end{split}
\label{eqn:surf_hol_expansion}
\end{equation}
We already showed in \cite{Sahlmann:2015dna} that
\begin{equation}
\begin{split}
j^{\sfrac{1}{2}} (\partial) \left[ ||E||^{2k} E_{i} \right] &= \sum_{N=0}^{k} \frac{1}{(2N+1)!} \frac{1}{8^{N}} \frac{(2k+1)!}{(2k-2N+1)!} \frac{2k+3}{2k-2N+3} \, ||E||^{2(k-N)} E_{i}\\
&= \sum_{N=0}^{k} \frac{1}{8^{N}} \binom{2k+4}{2N+1} \frac{2k-2N+2}{(2k+2)(2k+4)} \, ||E||^{2(k-N)} E_{i}
\end{split}
\end{equation}
and
\begin{equation}
Q_{S} \left[ ||E||^{2k} E_{i} \right] = \frac{Q_{S} \left[ ||E||^{2(k+1)} \right]}{\Delta_{SU(2)}} \, Q_{S} \left[ E_{i} \right] \, .
\end{equation}
Combining these two expression with the fact that the Laplacian of \SU ~evaluates to
\begin{equation}
\Delta_{\SU} \Biggl\vert_{\Hej} = \frac{j(j+1)}{2} \id_{\Hej}
\end{equation}
on a single edge carrying spin $j$, we obtain
\begin{equation}
\begin{split}
Q_{DK} &\left[ ||E||^{2k} E_{i} \right] \Biggl\vert_{\Hej} = \left[ Q_{S} \circ j^{\sfrac{1}{2}} (\partial) \right] \left( ||E||^{2k} E_{i} \right) \Biggl\vert_{\Hej} \\
&= \sum_{N=0}^{k} \frac{1}{8^{N}} \binom{2k+4}{2N+1} \frac{2k-2N+2}{(2k+2)(2k+4)} Q_{S} \left[ ||E||^{2(k-N)} E_{i} \right] \Biggl\vert_{\Hej} \\
&= - \frac{1}{8^{k+1}} \frac{2}{j(j+1)} \sum_{N=0}^{k} \binom{2k+4}{2N+1} \frac{2k-2N+2}{(2k+2)(2k+4)}~\times \\ &~~~ \sum_{m=0}^{2(k-N+1)} \binom{2(k-N)+3}{m} B_{m} \left( 2^{m} - 2 \right) \left[ 2j + 1 \right]^{2(k-N+1) - m} \pi^{(j)} \left[ \widehat{E}_{i} \right] \, \\
&= - \frac{1}{8^{k+1}} \frac{2}{j(j+1)} \sum_{p=0}^{k} \binom{2k+4}{2p+3} \frac{2p+2}{(2k+2)(2k+4)}~\times \\ &~~\quad \sum_{m=0}^{2(p+1)} \binom{2p+3}{m} B_{m} \left( 2^{m} - 2 \right) \left[ 2j + 1 \right]^{2(p+1) - m} \pi^{(j)} \left[ \widehat{E}_{i} \right] \, .
\end{split}
\label{eqn:Duflo_on_odd_powers_j-rep}
\end{equation}
After simplifying this expression (see appendix \ref{app:one-edge-calc} for details) we end up with
\begin{equation}
\begin{split}
Q_{DK} &\left[ ||E||^{2k} E_{i} \right] \Biggl\vert_{\Hej} \\
&= \frac{8}{2^{k}} \, \frac{2k+3}{2k+2} \, \frac{1}{2j(2j+1)(2j+2)} \, \left[ j \left( \frac{2j+1}{2} \right)^{2k+2} - \sum_{l=1}^{\frac{2j-1}{2}} l^{2k+2} \right] \, \pi^{(j)} \left[ \widehat{E}_{i} \right] \, .
\end{split}
\end{equation}
We can now use this result in combination with equation \eqref{eqn:surf_hol_expansion} to get
\begin{equation}
\begin{split}
Q_{DK} \left[ W_p \right] \Biggl\vert_{\Hej} &= \sum_{n=0}^{\infty} \frac{1}{(2n)!} \left( \frac{c}{2\sqrt{2}} \right)^{2n} Q_{DK} \left[ ||E||^{2n} \right] \Biggl\vert_{\Hej} \otimes \,\mathbb{1}_{2} \\
&\quad + \sum_{n=0}^{\infty} \frac{1}{(2n+1)!} \left( \frac{c}{2\sqrt{2}} \right)^{2n} c \, \kappa^{il} Q_{DK} \left[ ||E||^{2n} E_{i} \right] \Biggl\vert_{\Hej} \otimes \,\tau_{l} \\
&= \sum_{n=0}^{\infty} \frac{1}{(2n)!} \left( \frac{c}{2\sqrt{2}} \right)^{2n} \frac{1}{8^{n}} \left( 2j+1 \right)^{2n} \id_{\Hej} \otimes \,\mathbb{1}_{2} \\ 
&\quad + \sum_{n=0}^{\infty} \frac{1}{(2n+1)!} \left( \frac{c}{2\sqrt{2}} \right)^{2n} c \, \frac{8}{2^{n}} \, \frac{2n+3}{2n+2} \, \frac{1}{2j(2j+1)(2j+2)} \, \times \\
&\qquad \left[ j \left( \frac{2j+1}{2} \right)^{2n+2} - \sum_{l=1}^{\frac{2j-1}{2}} l^{2n+2} \right] \, \kappa^{il} \pi^{(j)} \left[ \widehat{E}_{i} \right] \otimes \,\tau_{l} \\
&= \operatorname{cosh} \left( \frac{(2j+1)c}{8} \right) \, \id_{\Hej} \otimes \,\mathbb{1}_{2} + \frac{128}{c} \, \frac{\kappa^{il} \pi^{(j)} \left[ \widehat{E}_{i} \right] \otimes \,\tau_{l}}{2j(2j+1)(2j+2)} \times \\
& \qquad \sum_{n=0}^{\infty} \frac{2n+3}{(2n+2)!} \left( \frac{c}{4} \right)^{2n+2} \, \left[ j \left( \frac{2j+1}{2} \right)^{2n+2} - \sum_{l=1}^{\frac{2j-1}{2}} l^{2n+2} \right] \, .
\end{split}
\label{eqn:W_S_as_sum}
\end{equation}
Simplifying once more (for details, see again appendix \ref{app:one-edge-calc}) and defining
\begin{equation}
\label{eqn:W_structure_single}
Q_{DK} \left[ W_p \right] \Biggl\vert_{\Hej} =: \xi_c(j) \, \id_{\Hej} \otimes \mathbb{1}_{2} + i\,\xi_s(j) \, \kappa^{im} \pi^{(j)} \left[ \widehat{E_{i}} \right] \otimes \tau_{m}
\end{equation}
we arrive at
\begin{equation}
\xi_c(j) = \operatorname{cosh} \left( \frac{(2j+1)c}{8} \right)
\end{equation}
and
\begin{equation}
\begin{split}
\xi_s(j) &= \frac{-128i}{2j(2j+1)(2j+2)}\\ 
& \quad \times \frac{1}{c} \, \frac{\operatorname{d}}{\operatorname{d}c} \left[  jc \cosh{\left( \frac{(2j+1)c}{8} \right)} - \frac{c}{2} - c \, \frac{\sinh{\frac{(2j-1)c}{16}}}{\sinh{\frac{c}{8}}} \cosh{\left( \frac{(2j+1)c}{16} \right)} \right]
\end{split}
\end{equation}
for the function $\xi_c(j)$ and $\xi_s(j)$. In the expression for $\xi_s(j)$, the derivative with respect to $c$ can still be carried out, leading to
\begin{equation}\label{eq:xis}
\begin{split}
\xi_s(j) &= \frac{-8i}{2j(2j+1)(2j+2)} \times \\
&\qquad \left[ 2j(2j+1) \frac{\cosh{\left( \frac{(2j+1)c}{8} \right)}}{\frac{(2j+1)c}{8}} + 2j(2j+1) \sinh{\left( \frac{(2j+1)c}{8} \right)} \right. \\
& \qquad \left. - \frac{1}{\sinh{\left( \frac{c}{8} \right)}} \left( 2j \cosh{\left( \frac{2jc}{8} \right)} + 2j \frac{\sinh{\left( \frac{2jc}{8} \right)}}{\frac{2jc}{8}} - \sinh{\left( \frac{2jc}{8} \right)} \coth{\left( \frac{c}{8} \right)} \right) \right] \, .
\end{split}
\end{equation}

\subsection{Action on two-sided puncture state}
\label{ssec:two-edge-punct-calc}

In order to perform the same calculation for the case of a two-sided puncture, we start again from the series expansion as given in \eqref{eqn:surf_hol_expansion}:
\begin{equation}
\begin{split}
W_{p} &=  \operatorname{cosh} \left( \frac{c}{2\sqrt{2}} ||E|| \right)\; \mathbb{1}_{2} + \frac{\operatorname{sinh} \left( \frac{c}{2\sqrt{2}} ||E|| \right)}{\frac{c}{2\sqrt{2}} ||E||} \, c \, \kappa^{ij} E_{i} \; \tau_{j} \\
&= \sum_{k=0}^{\infty} \frac{1}{(2k)!} \left( \frac{c}{2\sqrt{2}} \right)^{2k}  ||E||^{2k} \; \mathbb{1}_{2} + c \, \sum_{k=0}^{\infty} \frac{1}{(2k+1)!} \left( \frac{c}{2\sqrt{2}} \right)^{2k} \kappa^{ij}  ||E||^{2k} E_{i} \; \tau_{j} \, .
\end{split}
\end{equation}
When acting on a two-edge puncture state, we now have to distinguish several cases. Assuming that neither of the two edges is tangential to the surface, there are two main scenarios: the two edges can either lie on the same side of the surface $S_H$, or they can lie on different sides. In the first case, however, we can consider the quantity $\Eu_i = E^{(e)}_i + E^{(e')}_i$, which again behaves like an element of \su. This case can thus be treated as in the previous subsection. In the following, we will therefore focus on the case where one edge, $e$, lies above the surface and the other edge, $e'$, lies below $S_H$. In other words, we now have
\begin{equation}
\hatE_i = \hatEu_i - \hatEd_i \, ,
\label{eqn:E_decomp}
\end{equation}
where $\hatEu_i = \hatE^{(e)}_i$ and $\hatEd_i = \hatE^{(e')}_i$. Thus, $\hatEu_{i}$ inserts a generator of $\SU$ into the holonomy associated to the edge $e$ and $\hatEd_{i}$ acts analogously on $e'$. Since the combination \eqref{eqn:E_decomp} does no longer behave as an element of $\su$, we will have to order the quantities $\hatEu_{i}$ and $\hatEd_{i}$ individually. We can write
\begin{equation}
\begin{split}
||E||^{2} &= \kappa^{ij} E_{i} E_{j}  \\
&= \kappa^{ij} \left( \Eu_{i} - \Ed_{i} \right) \left( \Eu_{j} - \Ed_{j} \right) \\
&= ||\Eu||^{2} + ||\Ed||^{2} - 2 \kappa^{ij} \Eu_{i} \Ed_{j} \, .
\end{split}
\end{equation}
We thus see that, if we want to order both the $\Eu_{i}$ and $\Ed_{i}$ separately using the Duflo-Kirillov map, we need to evaluate said map on terms of the form
\begin{equation}
||\Eu||^{2k} \Eu_{i_{1}} \ldots \Eu_{i_{n}}
\end{equation}
and, unfortunately, we don't have a formula for this. In order to circumvent this problem, we will use the relation
\begin{equation}
||\Eud||^{2} = ||\Eu||^{2} + ||\Ed||^{2} + 2 \kappa^{ij} \Eu_{i} \Ed_{j}
\end{equation}
to obtain
\begin{equation}
||E||^{2} = 2 \, ||\Eu||^{2} + 2 \, ||\Ed||^{2} - ||\Eud||^{2} \, ,
\end{equation}
where
\begin{equation}
\Eud_{i} = \Eu_{i} + \Ed_{i} \, .
\end{equation}
Unfortunately, we cannot quantize $\Eu$, $\Ed$ and $\Eud$ independently, since, e.g., $\hatEu_{i}$ does not commute with $||\hatEud||^{2}$. However, if we focus on the sector of the quantum theory invariant under $\SU$ gauge transformations, $\hatEu$ and $\hatEd$ must couple to the trivial representation in the absence of transversal edges. We will therefore assume
\begin{equation}
||E^{(u+d)}||^{2} = 0
\end{equation}
already on the classical side.
The expression for $||E||^{2}$ then simplifies to
\begin{equation}
||E||^{2} = 2 \, ||\Eu||^{2} + 2 \, ||\Ed||^{2}
\end{equation}
and we can write arbitrary powers of this term as
\begin{equation}
\begin{split}
||E||^{2k} &= 2^{k} \, \left[ ||\Eu||^{2} + ||\Ed||^{2} \right]^{k} \\
&= 2^{k} \, \sum_{m=0}^{k} \binom{k}{m} \, ||\Eu||^{2m} \, ||\Ed||^{2(k-m)} \, .
\end{split}
\end{equation}
Inserting this expression into equation \eqref{eqn:surf_hol_expansion}, we then obtain
\begin{equation}
\begin{split}
W_{p} &= \sum_{k=0}^{\infty} \frac{1}{(2k)!} \left( \frac{c}{2} \right)^{2k} \, \sum_{m=0}^{k} \binom{k}{m} \, ||\Eu||^{2m} \, ||\Ed||^{2(k-m)} \; \mathbb{1}_{2} \\ &+ c \, \sum_{k=0}^{\infty} \frac{1}{(2k+1)!} \left( \frac{c}{2} \right)^{2k} \kappa^{ij} \, \sum_{m=0}^{k} \binom{k}{m} \, ||\Eu||^{2m} \, ||\Ed||^{2(k-m)} \, \left[ \Eu_i - \Ed_i \right] \; \tau_{j} 
\end{split}
\end{equation}
and applying the Duflo-Kirillov map leaves us with
\begin{equation}
\begin{split}
Q_{DK} &\left[ W_{p} \right] \Biggl\vert_{\Hejej} \\
&= \sum_{k=0}^{\infty} \frac{1}{(2k)!} \left( \frac{c}{2} \right)^{2k} \, \sum_{m=0}^{k} \binom{k}{m} \, Q_{DK} \left[ ||\Eu||^{2m} \right] \Biggl\vert_{\Heju} \, Q_{DK} \left[ ||\Ed||^{2(k-m)} \right] \Biggl\vert_{\Hejd} \; \otimes \mathbb{1}_{2} \\
&+ c \, \sum_{k=0}^{\infty} \frac{1}{(2k+1)!} \left( \frac{c}{2} \right)^{2k} \kappa^{ij} \, \sum_{m=0}^{k} \binom{k}{m} \, Q_{DK} \left[ ||\Eu||^{2m} \Eu_i \right] \Biggl\vert_{\Heju} \, Q_{DK} \left[ ||\Ed||^{2(k-m)} \right] \Biggl\vert_{\Hejd} \; \otimes \tau_{j} \\ 
&-c \, \sum_{k=0}^{\infty} \frac{1}{(2k+1)!} \left( \frac{c}{2} \right)^{2k} \kappa^{ij} \, \sum_{m=0}^{k} \binom{k}{m} \, Q_{DK} \left[ ||\Eu||^{2m} \right] \Biggl\vert_{\Heju} \, Q_{DK} \left[ ||\Ed||^{2(k-m)} \Ed_i \right] \Biggl\vert_{\Hejd} \; \otimes \tau_{j} \, .
\end{split}
\label{eqn:quantum_surf_hol_expanded_before_Duflo_evaluation}
\end{equation}
Note that we have calculated the action of the Duflo-Kirillov map on both types of terms showing up in this expression already in the previous subsection. If $E$ is associated to an edge labeled by spin $j$, this action is given by
\begin{equation}
\label{eqn:DK-map_gauge-inv_terms}
Q_{DK} \left[ ||E||^{2k} \right] \Biggl\vert_{\Hej} = \left( Q_{DK} \left[ ||E||^{2} \right] \Biggl\vert_{\Hej} \right)^{k} = \left[ \Delta_{\SU} \Biggl\vert_{\Hej} + \frac{1}{8} \, \id_{\Hej} \right]^{k} = \left[ \frac{(2j+1)^{2}}{8} \right]^{k} \, \id_{\Hej}
\end{equation}
and
\begin{equation}
Q_{DK} \left[ ||E||^{2k} E_{i} \right] \Biggl\vert_{\Hej} = \frac{2}{8^{k}} \frac{1}{2j(2j+1)(2j+2)} \frac{2k+3}{2k+2} \, \left[ j \left( 2j+1 \right)^{2k+2} - \sum_{l=1}^{\lfloor j \rfloor} \left( 2l \right)^{2k+2} \right] \, \pi^{(j)}(\hatE_{i}) \, ,
\label{eqn:DK-map_non-gauge-inv_terms}
\end{equation}
respectively, with $\lfloor j \rfloor$ denoting the floor function of $j$. Now, inserting these expressions into equation \eqref{eqn:quantum_surf_hol_expanded_before_Duflo_evaluation} and writing the result as
\begin{equation}
\begin{split}
Q_{DK} \left[ W_{p} \right] \Biggl\vert_{\Hejej} &= \chi_{c}(j^{u},j^{d}) \, \id_{\Hejej} \otimes \,\mathbb{1}_{2}\\
&+ i \chi_{s}(\ju,\jd) \, \kappa^{mn} \pi^{(j^{u})}(\hatEu_{m}) \otimes \id_{\Hejd} \otimes \tau_{n}\\
&- i \chi_{s}(\jd,\ju) \, \kappa^{mn} \id_{\Heju} \otimes \pi^{(j^{d})}(\hatEd_{m}) \otimes \tau_{n}\, ,
\end{split}
\end{equation}
the functions $\chi_{c}(\ju,\jd)$ and $\chi_{s}(\ju,\jd)$ take the forms
\begin{equation}
\chi_{c}(\ju,\jd) = \cosh \left( \frac{c}{2} \sqrt{\frac{(2\ju+1)^{2}}{8} + \frac{(2\jd+1)^{2}}{8}} \right)
\end{equation}
and
\begin{equation}
\begin{split}
\chi_{s}(\ju,\jd) &= -\frac{2i}{\ju + 1} \left[ \frac{\cosh \left( \frac{c}{2} \sqrt{\frac{(2\ju + 1)^2}{8} + \frac{(2\jd + 1)^2}{8}} \right) - \cosh \left( \frac{(2\jd + 1)c}{4\sqrt{2}} \right) }{\frac{(2\ju+1)c}{8}} + \frac{2\ju + 1}{2} \frac{\sinh \left( \frac{c}{2} \sqrt{\frac{(2\ju + 1)^2}{8} + \frac{(2\jd + 1)^2}{8}} \right)}{\sqrt{\frac{(2\ju + 1)^2}{8} + \frac{(2\jd + 1)^2}{8}}} \right]\\
&+ \frac{8i}{\ju  (\ju + 1)(2\ju + 1)} \sum_{k=1}^{\floor{\ju}} \left[ \frac{\cosh \left( \frac{c}{2} \sqrt{\frac{(\jd + 1)^2}{8} + \frac{k^2}{2}} \right) - \cosh \left( \frac{(2\jd + 1)c}{4\sqrt{2}} \right)}{\frac{c}{2}} + \frac{\frac{k^2}{2} \sinh \left( \frac{c}{2} \sqrt{\frac{(\jd + 1)^2}{8} + \frac{k^2}{2}} \right)}{\sqrt{\frac{(\jd + 1)^2}{8} + \frac{k^2}{2}}} \right] \, ,
\end{split}
\end{equation}
respectively. The details of the calculation can be found in appendix \ref{app:two_edge_calc}. Specializing to the gauge-invariant case\footnote{Recall that we have already imposed gauge-invariance partially on the classical side by demanding that $||\Eud||^{2} = 0$. The result for $\chi_{s}$ will probably change without this assumption.} where $j^{u} = j^{d} = j$, we end up with
\begin{equation}
\label{eqn:W_structure_double}
Q_{DK} \left[ W_{p} \right] \Biggl\vert_{\Heej} = \chi_{c}(j) \, \id_{\Heej} \otimes \,\mathbb{1}_{2} + i \chi_{s}(j) \, \kappa^{mn} \left[ \pi^{(j)}(\hatEu_{m}) \otimes \id_{\HEj} - \id_{\Hej} \otimes \pi^{(j)}(\hatEd_{m}) \right] \otimes \tau_{n}\, ,
\end{equation}
where now
\begin{equation}
\chi_{c}(j) = \cosh \left( \frac{(2j+1)c}{4} \right)
\end{equation}
and
\begin{equation}\label{eq:chis}
\begin{split}
\chi_{s}(j) &= - \frac{2i}{j+1} \left[  \frac{\cosh \left( \frac{(2j+1)c}{4} \right) - \cosh \left( \frac{(2j+1)c}{4\sqrt{2}} \right)}{\frac{(2j+1)c}{8}} + \sinh \left( \frac{(2j+1)c}{4} \right) \right]\\
&+ \frac{8i}{j(j+1)(2j+1)} \sum_{k=1}^{\floor{j}} \left[ \frac{\cosh \left( \frac{c}{2} \sqrt{\frac{(2j+1)^2}{8} + \frac{k^2}{2}} \right) - \cosh \left( \frac{(2j+1)c}{4\sqrt{2}} \right)}{\frac{c}{2}} + \frac{\frac{k^2}{2} \sinh \left( \frac{c}{2} \sqrt{\frac{(2j+1)^2}{8} + \frac{k^2}{2}} \right)}{\sqrt{\frac{(2j+1)^2}{8} + \frac{k^2}{2}}} \right] \, .
\end{split}
\end{equation}

\section{Properties of quantum surface holonomy operators} 
\label{se_quantum_group}

In this chapter we focus on the properties of the  holonomy operators just calculated. These properties are important since they determine the existence and the properties of solutions to the quantised isolated horizon boundary condition.

\subsection{Behavior under gauge transformations}
Gauge transformations $g:\Sigma \rightarrow \SU$ act as unitary operators $U_g$ on the LQG Hilbert space. They transform the basic field operators as 
\begin{equation}
U_g\, h_e \,U_g^\dagger = g(t(e)) \,h_e \, g(s(e))^{-1}, \qquad U_g\,E_k^{(e)}(p) \,U_g^\dagger = \pi_1(g(p)^{-1})^j{}_k \,E_j^{(e)}(p). 
\end{equation}
As a consequence, using the equations  \eqref{eqn:surfhol_singlepunct}, \eqref{eqn:W_structure_single}, \eqref{eqn:W_structure_double}, \eqref{eqn:factorization}  that define $\widehat{\mathcal{W}}_H$ in terms of $E$ and holonomies $h$, we find that it transforms as  
\begin{equation}
\label{eqn_gaugetrafo_W}
U_g\, \widehat{\mathcal{W}}_H\, U_g^\dagger= g(x_0)\,\widehat{\mathcal{W}}_H\,g(x_0)^{-1}, 
\end{equation}
where $x_0 \in \partial S_H$ denotes the special point on the boundary of $S_H$. Thus $\widehat{\mathcal{W}}_H$ transforms exactly as a holonomy beginning and ending in $x_0$.

\subsection{Matrix elements}
\label{se:matrix_elements}
The quantum surface holonomy operators $\widehat{\mathcal{W}}$ are operator-valued matrices. In the following, we will consider their components.  In particular, we will take a look at the adjointness and commutation relations between matrix elements of $Q_{DK}[W_p]$ and $\mathcal{W}_H$ and compare them to those from known quantum group deformations of $SU(2)$. We will always assume that the holonomies act on single puncture states. We will distinguish the case of a one-sided and a two-sided puncture. We also assume a relative orientation between the surface $S$ and the intersecting edge as in figure \ref{fig:single_puncture_case}. Changing the orientation of $S$ will change the sign of the second term in \eqref{eq:w_action} and \eqref{eq:w_action2}, and hence some signs in the equations following them. 

Let us first consider the operator $\widehat{W}_p$ on a one-sided puncture. We explictly consider only the action on $\Hpju$. The action on $\Hpjd$ just differs by a factor of $-1$ in $\widehat{E}(p)$. In the previous chapter, we found 
\begin{equation}
\begin{split}
Q_{DK}[W_p] \Biggl\vert_{\Hpju}  &= \xi_{c}(j) \id_{\Hpju} \otimes \,\mathbb{1}_{2} + i \xi_{s}(j) \, \kappa^{mn} \, \pij{j}{\widehat{E_{m}}} \otimes \,\tau_{n} \\
&= \begin{pmatrix}
\xi_{c}(j) \id_{\Hpju}  - \frac{1}{4} \xi_{s}(j) \pij{j}{\widehat{E_{3}}} & -\frac{1}{4} \xi_{s}(j) \left( \pij{j}{\widehat{E_{1}}} - i \pij{j}{\widehat{E_{2}}} \right)\\
-\frac{1}{4} \xi_{s}(j) \left( \pij{j}{\widehat{E_{1}}} + i \pij{j}{\widehat{E_{2}}} \right) & \xi_{c}(j) \id_{\Hpju}  + \frac{1}{4} \xi_{s}(j) \pij{j}{\widehat{E_{3}}}
\end{pmatrix}\\
&=\begin{pmatrix}
\xi_{c}(j) \id_{\Hpju}  - \frac{1}{4} \xi_{s}(j) \pij{j}{\widehat{E_{3}}} & -\frac{1}{4} \xi_{s}(j) \pij{j}{\widehat{E_{-}}}\\
-\frac{1}{4} \xi_{s}(j) \pij{j}{\widehat{E_{+}}}  & \xi_{c}(j) \id_{\Hpju}  + \frac{1}{4} \xi_{s}(j) \pij{j}{\widehat{E_{3}}}
\end{pmatrix}\, ,
\label{eqn:W_S_matrixform}
\end{split}
\end{equation} 
where we have now introduced the notation 
\begin{equation}
\widehat{E_{\pm}} := \widehat{E_{1}} \pm i \widehat{E_{2}}.
\end{equation}
Using the fact that the $\widehat{E}_i$ are skew-adjoint, we can write  
\begin{equation}
Q_{DK}[W_p] \Biggl\vert_{\Hpju} 
= \begin{pmatrix}
\widehat{a}&\widehat{b}\\
-\widehat{b}^\dagger &\widehat{a}^\dagger
\end{pmatrix} \, ,
\label{eqn:adjointness}
\end{equation}
with 
\begin{align}
\widehat{a}&=\xi_{c}(j) \id_{\Hpju}  - \frac{1}{4} \xi_{s}(j) \pij{j}{\widehat{E_{3}}},\\
\widehat{b}&=-\frac{1}{4} \xi_{s}(j) \pij{j}{\widehat{E_{-}}}.
\end{align}
For the double puncture, the structure is similar:
\begin{equation}
\begin{split}
Q_{DK} &\left[ W_{p} \right] \Biggl\vert_{\Hpjj} = \chi_{c}(j) \, \id_{\Hpjj} \otimes \,\mathbb{1}_{2} + i \chi_{s}(j) \, \kappa^{mn} \left[ \pij{j}{\hatEu_{m}}  - \pij{j}{\hatEd_{m}} \right] \otimes \,\tau_{n}\\
&= \begin{pmatrix}
\chi_{c}(j) \, \id_{\Hpjj} - \frac{1}{4} \chi_{s}(j) \, \pij{j}{\widehat{E_{3}^{(u)}} - \widehat{E_{3}^{(d)}}} 
& -\frac{1}{4} \chi_{s}(j) \left( \pij{j}{\widehat{E_{-}^{(u)}} - \widehat{E_{-}^{(d)}}} \right)\\
-\frac{1}{4} \chi_{s}(j) \left( \pij{j}{\widehat{E_{+}^{(u)}} - \widehat{E_{+}^{(d)}}} \right) 
& \chi_{c}(j) \id_{\Hpjj}  + \frac{1}{4} \chi_{s}(j) \left( \pij{j}{\widehat{E_{3}^{(u)}} - \widehat{E_{3}^{(d)}}} \right)
\end{pmatrix}\\
&=: \begin{pmatrix}
\widehat{a}&\widehat{b}\\
-\widehat{b}^\dagger &\widehat{a}^\dagger
\end{pmatrix} \, 
\label{eqn:W_S_matrixform_2}
\end{split}
\end{equation}
with 
\begin{align}
\widehat{a}&=\chi_{c}(j) \, \id_{\Hpjj} - \frac{1}{4} \chi_{s}(j) \, \left(\pij{j}{\widehat{E_{3}^{(u)}} - \widehat{E_{3}^{(d)}}}\right) \\
\widehat{b}&=-\frac{1}{4} \chi_{s}(j) \left( \pij{j}{\widehat{E_{-}^{(u)}} - \widehat{E_{-}^{(d)}}} \right) \, .
\end{align}
We will now turn to the matrix elements of $\widehat{\mathcal{W}}$. Recall from \eqref{eqn:surfhol_singlepunct} that 
\begin{equation}
\label{eqn:WW}
\left.\widehat{\mathcal{W}}_H\right\rvert_{\Hpju}
= h^{-1}_{\alpha_{p}}\, Q_{\text{DK}}\left[  W_{p}  \right]  h_{\alpha_{p}} \, .
\end{equation}
We first observe that the matrix elements of $h_{\alpha_{p}}$ and $h^{-1}_{\alpha_{p}}$ commute with the $\widehat{E}_i$, and hence with $Q_{\text{DK}}\left[  W_{p}  \right]$ because $h_{\alpha_{p}}$ runs tangential to the surface and there is no intertwiner connecting $h_{\alpha_{p}}$ and the holonomy of the puncture. 
Secondly, we also notice that products of matrices with the adjointness structure \eqref{eqn:adjointness} again have the same structure.  The matrices on the right hand side of \eqref{eqn:WW} are operator-valued, but ,as observed, the entries of the holonomies commute with those of $W_p$. We can thus conclude that 
\begin{equation}
\left.\widehat{\mathcal{W}}_H\right\rvert_{\Hpju}=\begin{pmatrix}\widehat{\mathcal{a}}&\widehat{\mathcal{b}}\\-\widehat{\mathcal{b}}^\dagger&\widehat{\mathcal{a}}^\dagger
\end{pmatrix} \, .
\end{equation}
Next, we can determine the matrix entries of $\widehat{\mathcal{W}}_H$. To this end, note the intertwiner properties
\begin{equation}
g\tau_ig^{-1}= \tau_j\pi_1(g)^{j}{}_{i}, \qquad \pi_1(g^{-1})^{n'}{}_n \,\kappa^{nm}=\kappa^{m'n'}\, \pi_1(g)^{m}{}_{m'} 
\end{equation}
of the $\tau_i$ and $\kappa$. As a consequence, we can write 
\begin{equation}
\begin{split}
\widehat{\mathcal{W}}_H \Biggl\vert_{\Hpju}  &= \xi_{c}(j) \id_{\Hpju} \otimes \,\mathbb{1}_{2} + i \xi_{s}(j) \, \kappa^{mn} \, \pij{j}{\widehat{E_{m}}} \otimes \,h^{-1}_{\alpha_{p}}\tau_{n}h_{\alpha_{p}} \\
&=\xi_{c}(j) \id_{\Hpju} \otimes \,\mathbb{1}_{2} + i  \xi_{s}(j)\, \pi_1(h^{-1}_\alpha)^{n'}{}_{n} \kappa^{mn} \, \pij{j}{\widehat{E_{m}}} \otimes \,\tau_{n'}\\
&=: \xi_{c}(j) \id_{\Hpju} \otimes \,\mathbb{1}_{2} + i  \xi_{s}(j)\, \kappa^{mn} \, \widehat{\mathcal{E}}_{m} \otimes \,\tau_{n'} \, ,
\end{split}
\label{eqn:calW_S_matrixform}
\end{equation} 
where we have introduced 
\begin{equation}
\widehat{\mathcal{E}}_{m}=h^{m'}{}_m  \, \pij{j}{\widehat{E_{m'}}}\, .
\end{equation}
Note that the last expression in \eqref{eqn:calW_S_matrixform} is of identical form as that in \eqref{eqn:W_S_matrixform}, except for the replacement of $\pij{j}{\widehat{E_{m}}}$ by $\widehat{\mathcal{E}}_{m}$. Therefore, we have
\begin{align}
\widehat{\mathcal{a}}&=\xi_{c}(j) \id_{\Hpju}  - \frac{1}{4} \xi_{s}(j) \widehat{\mathcal{E}_{3}},\\
\widehat{\mathcal{b}}&=-\frac{1}{4} \xi_{s}(j) \widehat{\mathcal{E}_{-}}\, .
\end{align}
The same reasoning applies to the case of the two-sided puncture, hence 
\begin{equation}
\left.\widehat{\mathcal{W}}_H\right\rvert_{\Hpjj}=\begin{pmatrix}\widehat{\mathcal{a}}&\widehat{\mathcal{b}}\\-\widehat{\mathcal{b}}^\dagger&\widehat{\mathcal{a}}^\dagger
\end{pmatrix}
\end{equation}
with 
\begin{align}
\widehat{\mathcal{a}}&=\chi_{c}(j) \, \id_{\Hpjj} - \frac{1}{4} \chi_{s}(j) \, \left(\widehat{\mathcal{E}_{3}^{(u)}} - \widehat{\mathcal{E}{_{3}^{(d)}}}\right) \\
\widehat{\mathcal{b}}&=-\frac{1}{4} \chi_{s}(j) \left( \widehat{\mathcal{E}_{-}^{(u)}} - \widehat{\mathcal{E}_{-}^{(d)}} \right)\, .
\end{align}
Let us remark that the adjointness structure of $\widehat{W}$ and $\widehat{\mathcal{W}}$ mirrors that of an SU(2) element in the defining representation. The remaining condition on the matrix components of an SU(2) is given by the requirement that the determinant equals unity. We will turn to this requirement in the next subsection. Here, we will demonstrate that we are far from classical SU(2), by calculating the commutators of matrix elements. 

Let us first consider the case of the one-sided puncture. Using the fact that the $\widehat{E}_i$ have su(2) commutators in this case, we find 
\begin{align}
[ \widehat{a} , \widehat{b} ] &= - \frac{i\xi_{s}(j)}{4} \, \widehat{b} & [ \widehat{a} , \widehat{b}^\dagger ] &= \frac{i\xi_{s}(j)}{4} \, \widehat{b}^\dagger \\
[ \widehat{a} , \widehat{a}^\dagger ] &= 0 & [ \widehat{b} , \widehat{b}^\dagger ] &= \frac{i\xi_{s}(j)}{4} \left( \widehat{a} - \widehat{a}^\dagger \right) \\
[ \widehat{b} , \widehat{a}^\dagger ] &= [ \widehat{a} , \widehat{b} ] & [ \widehat{b}^\dagger , \widehat{a}^\dagger ] &= [ \widehat{a} , \widehat{b}^\dagger ] \, .
\end{align}
Using the fact that the holonomies $h_\alpha$ in the surface commute with the $\widehat{E}_i$, and that 
$ \pi_1(h_\alpha)$ is an orthogonal matrix, one can show that also the $\widehat{\mathcal{E}}_{m}$ satisfy 
su(2) commutation relations, and hence in complete analogy
\begin{align}
\label{eq:one_sided_commutator}
[ \widehat{\mathcal{a}} , \widehat{\mathcal{b}} ] &= - \frac{i\xi_{s}(j)}{4} \, \widehat{\mathcal{b}} & [ \widehat{\mathcal{a}} , \widehat{\mathcal{b}}^\dagger ] &= \frac{i\xi_{s}(j)}{4} \, \widehat{\mathcal{b}}^\dagger \\
[ \widehat{\mathcal{a}} , \widehat{\mathcal{a}}^\dagger ] &= 0 & [ \widehat{\mathcal{b}} , \widehat{\mathcal{b}}^\dagger ] &= \frac{i\xi_{s}(j)}{4} \left( \widehat{\mathcal{a}} - \widehat{\mathcal{a}}^\dagger \right) \\
[ \widehat{\mathcal{b}} , \widehat{\mathcal{a}}^\dagger ] &= [ \widehat{\mathcal{a}} , \widehat{\mathcal{b}} ] & [ \widehat{\mathcal{b}}^\dagger , \widehat{\mathcal{a}}^\dagger ] &= [ \widehat{\mathcal{a}} , \widehat{\mathcal{b}}^\dagger ] \, .
\end{align}
For the double sided puncture, the reasoning is again analogous. Note however, that in contrast to the sum of two angular momenta the difference of two angular momenta is not again an angular momentum operator in the sense of commutation relations. This holds in particular for $\widehat{E^{(u)}}-\widehat{E^{(d)}}$ and $\widehat{\mathcal{E}^{(u)}}-\widehat{\mathcal{E}^{(d)}}$. For example
\begin{equation*}
\left[\widehat{E_-^{(u)}}-\widehat{E^{(d)}_-}, \widehat{E^{(u)}_+}-\widehat{E^{(d)}_+}\right]=2i\,\widehat{E^{(u+d)}_3}.
\end{equation*}
This changes the commutation relations of the matrix elements slightly. We will only give the relations for the matrix elements of the full surface holonomy, since the ones for $\widehat{W_p}$ are structurally identical. They are:
\begin{align}
[ \widehat{\mathcal{a}} , \widehat{\mathcal{b}} ] &= \frac{i\chi^2_{s}(j)}{16} \, \widehat{\mathcal{E}_-^{(u+d)}} & [ \widehat{\mathcal{a}} , \widehat{\mathcal{b}}^\dagger ] &=\frac{i\chi^2_{s}(j)}{16} \, \widehat{\mathcal{E}_+^{(u+d)}}\\
[ \widehat{\mathcal{a}} , \widehat{\mathcal{a}}^\dagger ] &= 0 & [ \widehat{\mathcal{b}} , \widehat{\mathcal{b}}^\dagger ] &= -\frac{i\chi^2_{s}(j)}{8} \widehat{\mathcal{E}_3^{(u+d)}}\\
[ \widehat{\mathcal{b}} , \widehat{\mathcal{a}}^\dagger ] &= [ \widehat{\mathcal{a}} , \widehat{\mathcal{b}} ] & [ \widehat{\mathcal{b}}^\dagger , \widehat{\mathcal{a}}^\dagger ] &= [ \widehat{\mathcal{a}} , \widehat{\mathcal{b}}^\dagger ] \, .
\label{eq:two_sided_commutator}
\end{align}
Let us finally compare these commutation relations to those appearing in standard quantum deformations of $\SU$, such as $\SUq$ (see for example \cite{masudaetal}). At least in the standard representations, the latter have a different structure. For example it would hold that $\widehat{\mathcal{a}} \widehat{\mathcal{b}}=q\widehat{\mathcal{b}} \widehat{\mathcal{a}}$ which would correspond to a commutator 
\begin{equation*}
[ \widehat{\mathcal{a}} , \widehat{\mathcal{b}} ]{\underset{\SUq}{}}=(q-1)\widehat{\mathcal{b}} \widehat{\mathcal{a}}=(q+1)\widehat{\mathcal{a}} \widehat{\mathcal{b}} \, .
\end{equation*}
By comparison, our commutators are linear in the matrix elements. Thus, we are very likely dealing with a different mathematical object. 
 
\subsection{Determinant}
In the present section, we will consider the determinant of surface holonomy operators. The determinant is especially relevant if we aim to solve the quantized isolated horizon BC by states in a representation of the standard holonomy-flux algebra: The holonomies of the HF-algebra are $SU(2)$-valued functionals and therefore their determinant is unity. 

We define 
\begin{equation}
\operatorname{det}_{\delta} \widehat{\mathcal{W}}_H 
\equiv \operatorname{det}_{\delta} \begin{pmatrix}\widehat{\mathcal{a}}&\widehat{\mathcal{b}}\\-\widehat{\mathcal{b}}^\dagger&\widehat{\mathcal{a}}^\dagger \end{pmatrix}
:=  \widehat{\mathcal{a}}\widehat{\mathcal{a}}^\dagger + \delta\,\widehat{\mathcal{b}}\widehat{\mathcal{b}}^\dagger + (1-\delta)\, \widehat{\mathcal{b}}^\dagger\widehat{\mathcal{b}}
\end{equation}
where the parameter $\delta$ labels some of the possible operator orderings.  We will first consider the transformation behaviour under gauge transformations. We parametrize a classical SU(2) element as 
\begin{equation*}
g=\begin{pmatrix} \alpha&\beta\\-\overline{\beta}&\overline{\alpha}\end{pmatrix}, \qquad 
\alpha, \beta\in \mathbb{C} \text{ with } |\alpha|^2+|\beta|^2=1. 
\end{equation*}
A tedious but straightforward calculation shows that 
\begin{equation}
\begin{split}
\operatorname{det}_{\delta}\left[ g\,\widehat{\mathcal{W}}_H\, g^{-1}\right]= \widehat{\mathcal{a}}\widehat{\mathcal{a}}^\dagger
&+ \widehat{\mathcal{b}}\widehat{\mathcal{b}}^\dagger\left[|\alpha|^2|\beta|^2 +\delta|\alpha|^4 +(1-\delta)|\beta|^4 \right]\\
&+\widehat{\mathcal{b}}^\dagger \widehat{\mathcal{b}}\left[|\alpha|^2|\beta|^2 +\delta|\beta|^4 +(1-\delta)|\alpha|^4 \right]\\
&+ (\widehat{\mathcal{a}}^\dagger \widehat{\mathcal{b}}^\dagger-  \widehat{\mathcal{b}}^\dagger \widehat{\mathcal{a}}^\dagger) (2\delta-1) \, .
\end{split}
\end{equation}
Thus 
\begin{equation}
\operatorname{det}_{\frac{1}{2}}\left[ g\,\widehat{\mathcal{W}}_H\, g^{-1}\right]=\operatorname{det}_{\frac{1}{2}}\left[\widehat{\mathcal{W}}_H\right]
\end{equation}
and, in view of  \eqref{eqn_gaugetrafo_W}, the symmetrically ordered determinant is gauge-invariant. This also implies that 
\begin{equation}
\operatorname{det}_{\frac{1}{2}}\left[\widehat{\mathcal{W}}_H\right]=\operatorname{det}_{\frac{1}{2}}\left[\widehat{W_p}\right] \, .
\end{equation}
Altogether, the symmetric ordering seems to be preferred, and we will often restrict consideration to this case. We start with the action on the one-sided puncture:
\begin{equation}
\begin{split}
\operatorname{det}_{\delta} \left. \widehat{\mathcal{W}}_H \right \rvert_{\Hpju} 
&= \xi_{c}(j)^{2} \, \id_{\Hpju} - \frac{\xi_{s}(j)^{2}}{16} \left( \pij{j}{\widehat{\mathcal{E}_{3}}} \right)^{2}\\
&\quad- \frac{\xi_{s}(j)^{2}}{16} \left[ \left( \pij{j}{\widehat{\mathcal{E}_{1}}} \right)^{2} + \left( \pij{j}{\widehat{\mathcal{E}_{2}}} \right)^{2} + i \delta \,[\pij{j}{\widehat{\mathcal{E}_{1}}}, \pij{j}{\widehat{\mathcal{E}_{2}}}] +i (1-\delta) [\pij{j}{\widehat{\mathcal{E}_{2}}}, \pij{j}{\widehat{\mathcal{E}_{1}}}]\right]\\
&=\xi_{c}(j)^{2} \, \id_{\Hpju} - \frac{\xi_{s}(j)^{2}}{16}\, \pij{j}{\widehat{\mathcal{E}}^2}+ i (1-2\delta)  \frac{\xi_{s}(j)^{2}}{16} \pij{j}{\widehat{\mathcal{E}_{3}}}\\
& =\left(\xi_{c}(j)^{2} + \frac{\xi_{s}(j)^{2}}{8}\,  \Delta_j \right)\, \id_{\Hpju} +i (1-2\delta) \frac{\xi_{s}(j)^{2}}{16} \pij{j}{\widehat{\mathcal{E}_{3}}}\, .
\end{split}
\end{equation}
For symmetric ordering this reduces to 
\begin{equation}
\operatorname{det}_{\frac{1}{2}} \left. \widehat{\mathcal{W}}_H \right \rvert_{\Hpju} 
= \left(\xi_{c}(j)^{2} + \frac{\xi_{s}(j)^{2}}{8} \, \Delta_j \right)\, \id_{\Hpju}. 
\end{equation}
For the two sided puncture, the determinant acts as 
\begin{equation}
\begin{split}
\operatorname{det}_{\delta} \left. \widehat{\mathcal{W}}_H \right \rvert_{\Hpjj}
&= \chi_{c}(j)^{2}\,\id_{\Hpjj} - \frac{1}{16} \, \chi_{s}(j)^{2} \, \left[ \left( \pij{j^u}{\widehat{\mathcal{E}_{3}^{(u)}}} \right)^{2} + \left( \pij{j^d}{\widehat{\mathcal{E}_{3}^{(d)}}} \right)^{2} - 2 \, \pij{j^u}{\widehat{\mathcal{E}_{3}^{(u)}}} \, \pij{j^d}{\widehat{\mathcal{E}_{3}^{(d)}}}\right] \\
&\qquad \qquad - \frac{1}{16} \, \chi_{s}(j)^{2} \, \Biggl[ \left( \pij{j^u}{\widehat{\mathcal{E}_{1}^{(u)}}} \right)^{2} + \left( \pij{j^u}{\widehat{\mathcal{E}_{2}^{(u)}}} \right)^{2} + \left( \pij{j^d}{\widehat{\mathcal{E}_{1}^{(d)}}} \right)^{2} + \left( \pij{j^d}{\widehat{\mathcal{E}_{2}^{(d)}}} \right)^{2} 
- 2 \, \pij{j^u}{\widehat{\mathcal{E}_{1}^{(u)}}} \, \pij{j^d}{\widehat{\mathcal{E}_{1}^{(d)}}} - 2 \, \pij{j^u}{\widehat{\mathcal{E}_{2}^{(u)}}} \, \pij{j^d}{\widehat{\mathcal{E}_{2}^{(d)}}}\\
&\qquad \qquad+ i(2\delta-1)\, \left[ \pij{j^u}{\widehat{\mathcal{E}_{1}^{(u)}}} , \pij{j^u}{\widehat{\mathcal{E}_{2}^{(u)}}} \right] + i(2\delta-1)\,\left[ \pij{j^d}{\widehat{\mathcal{E}_{1}^{(d)}}} , \pij{j^d}{\widehat{\mathcal{E}_{2}^{(d)}}} \right]  \Biggr]\\
&=\chi_{c}(j)^{2}\,\id_{\Hpjj} - \frac{1}{16} \, \chi_{s}(j)^{2} \, \left[ \left(\pij{j^u}{\widehat{\mathcal{E}^{(u)}}} \right)^{2} + \left( \pij{j^u}{\widehat{\mathcal{E}^{(d)}}} \right)^{2} -2\, \widehat{\mathcal{E}^{(u)}} \cdot \widehat{\mathcal{E}^{(d)}}\right]
- \frac{i}{16}(2\delta-1) \, \chi_{s}(j)^{2} \, \widehat{\mathcal{E}^{(u+d)}_3} \, .
\end{split}
\end{equation}
For the symmetric ordering, this reduces to 
\begin{equation}
\operatorname{det}_{\frac{1}{2}} \left. \widehat{\mathcal{W}}_H \right \rvert_{\Hpjj}
=\chi_{c}(j)^{2}\,\id_{\Hpjj} + \frac{1}{8} \, \chi_{s}(j)^{2} \, \left[ 2\Delta_j \,\id_{\Hpjj} + \widehat{{E}^{(u)}} \cdot \widehat{{E}^{(d)}}\right], 
\end{equation}
and on the gauge-invariant Hilbert space to 
\begin{equation}
\begin{split}
\operatorname{det}_{\frac{1}{2}} \left. \widehat{\mathcal{W}}_H \right \rvert_{\Hpjj}
&=\chi_{c}(j)^{2}\,\id_{\Hpjj} + \frac{1}{8} \, \chi_{s}(j)^{2} \, \left[ 2\Delta_j \,\id_{\Hpjj} - \left(\widehat{{E}^{(u)}}\right)^2\right]\\
&=\left(\chi_{c}(j)^{2} + \frac{1}{2} \, \chi_{s}(j)^{2} \, \Delta_j \right)\,\id_{\Hpjj}  \, .
\end{split}
\end{equation}
We see that, in general, the eigenvalues of the determinant operator differ from 1. However, there is a limit in which they get close. Recall that $\xi_{c}(j)$ and $\xi_{s}(j)$ (and similarly $\chi_{c}(j)$ and $\chi_{s}(j)$)
are both power series in the parameter $c$ introduced in \eqref{eqn:def_c}. For small $c$, \footnote{In the application to black holes, $c$ contains the area of the black hole horizon in the denominator, thus we can assume $c$ to be small in the case of macroscopic black holes, for example.}
we can consider the Taylor expansion of the eigenvalue of the determinant operator to second order. We get
\begin{equation}
\xi_{c}(j) \approx 1 + \frac{\left( 2j + 1 \right)^{2} \, c^{2}}{32} 
\end{equation}
and
\begin{equation}
\xi_{s}(j) \approx \mathcal{O}(c) \, ,
\end{equation}
and a similar result for $\chi_{c}(j)$ and $\chi_{s}(j)$. This shows that for small $c$ we are in a regime in which $\widehat{\mathcal{W}}_H$ is close to a classical $\SU$-element.  

Another regime in which the determinant is close to 1 can be seen from the plots in figure \ref{fig:determinant}. 
\begin{figure}
\includegraphics[width=0.4\textwidth]{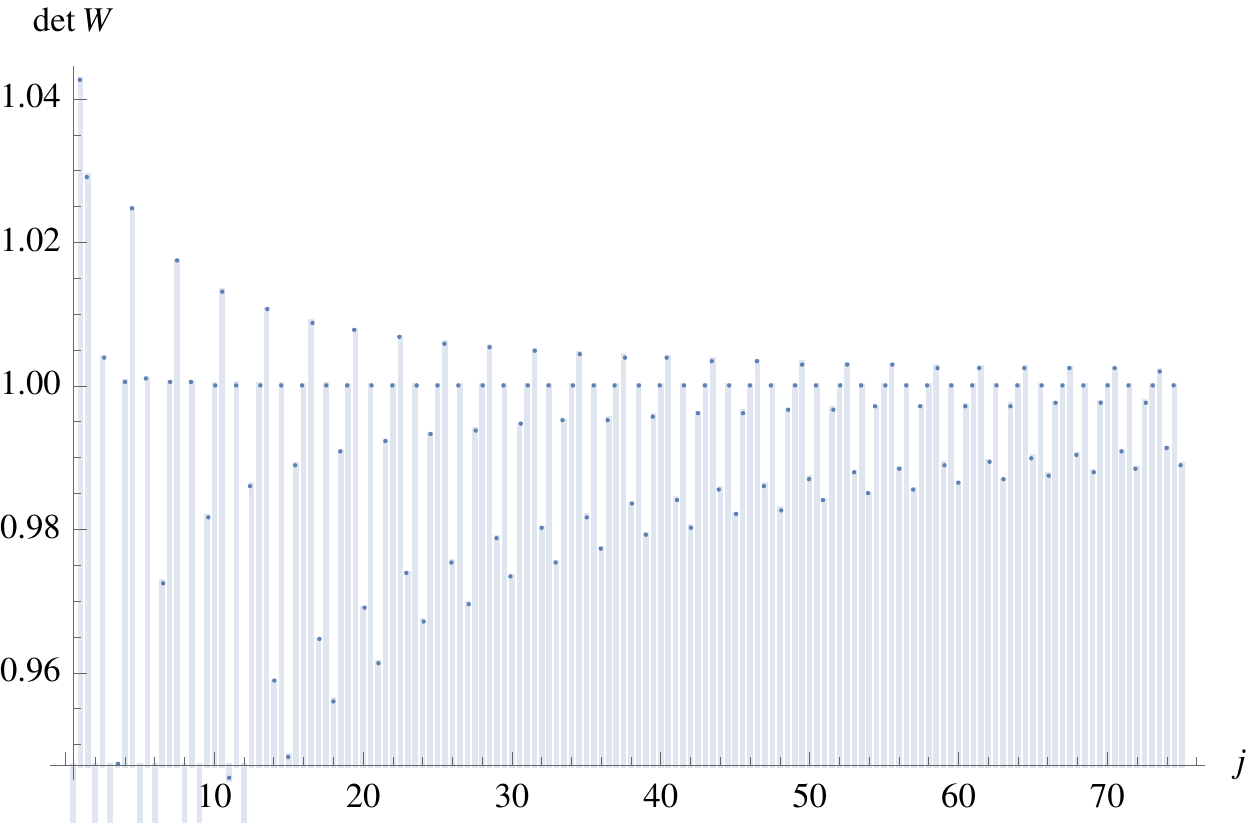}
\hspace{0.05\textwidth}
\includegraphics[width=0.4\textwidth]{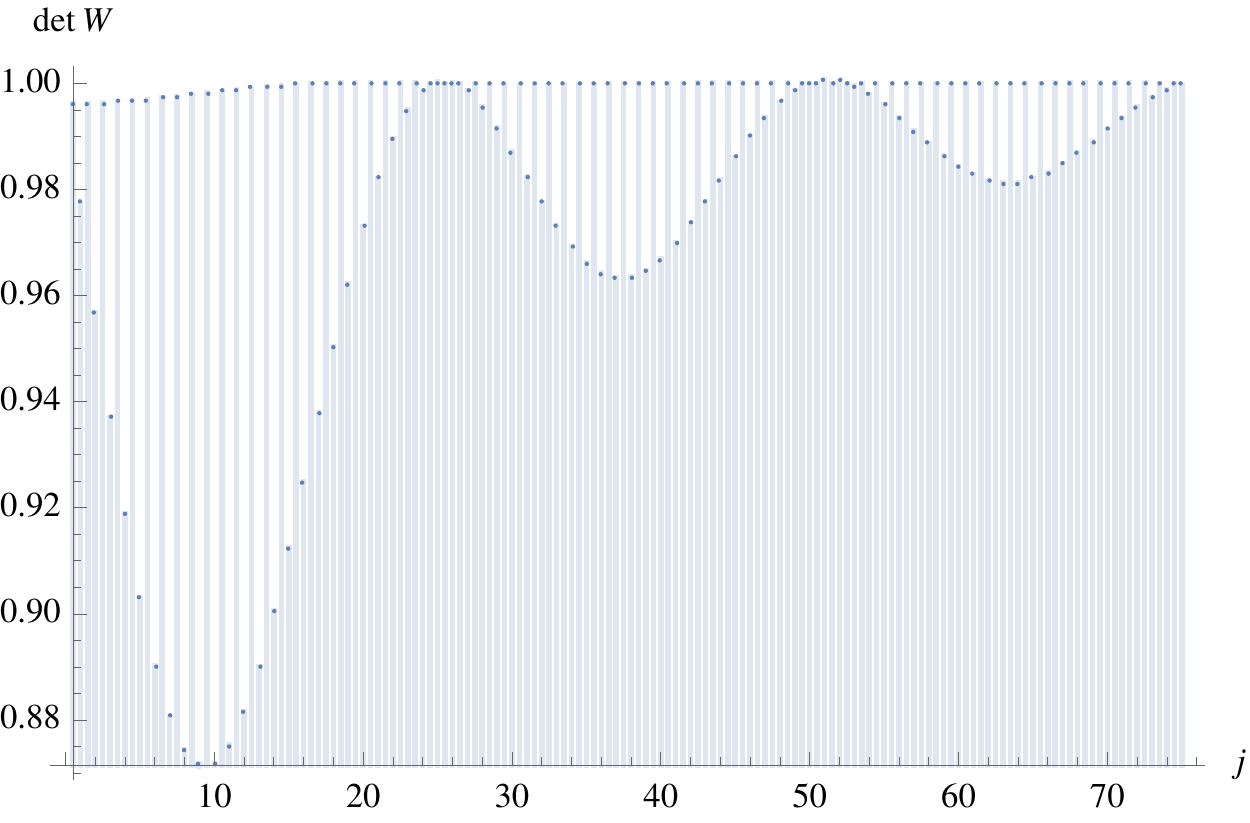}
\caption{The eigenvalues of the determinant for $c=8\pi i/k$ with $k=3$ (left)and $k=101$ (right). }
\label{fig:determinant}
\end{figure}
For fixed $c=8\pi i/k$ with $k\in \mathbb{N}$, the eigenvalues oscillate as a function of $j$ with a period set by $k$, but they tend to 1 quickly as $j$ gets larger.  Additionally, it appears that there are also certain small values of $j$ for which the eigenvalue is very close to 1. For example, in the plot for $k=3$ there is a series $\{5/2, 4, 11/2, 7, 17/2,\ldots \}$ of $j$-values with determinant close to 1. One notices a spacing of $k/2$.  For $k=101$ there is a similar series $\{1/2, 3/2, 5/2, 7/2, \ldots \}$. 

\subsection{Adjoint operator}

As we have seen in chapter \ref{sec:quant}, the quantum operator associated to a surface holonomy $\mathcal{W}_{H}$ takes the form
\begin{equation}
\widehat{\mathcal{W}}_{H} \Biggl\vert_{\Heej} = \chi_{c}(j) \, \id_{\Heej} \otimes \, \mathbb{1}_{2} + i \chi_{s}(j) \, \kappa^{mn} \, \pi^{(j)}[\widehat{E}_{m}]  \otimes \, h^{-1}_{\alpha_{p}} \tau_{n} h^{\vphantom{-1}}_{\alpha_{p}}\, ,
\end{equation}
where $p$ denotes the location of the puncture. Quantum surface holonomies can thus be regarded as two-by-two matrices whose entries are operators acting on spin network states. The adjoint $\widehat{W}_{H}^{\dagger}$ is thus given by transposing the two-by-two matrix and then taking the adjoint of each entry as an operator. As both $\chi_{c}(j)$ and $\chi_{s}(j)$ are real, this leads to
\begin{equation}
\widehat{\mathcal{W}}_{H}^{\dagger} \Biggl\vert_{\Heej} = \chi_{c}(j) \, \id_{\Heej} \otimes \, \mathbb{1}_{2} - i \chi_{s}(j) \, \kappa^{mn} \, \pi^{(j)}[\widehat{E}_{m}] \otimes \, h^{-1}_{\alpha_{p}} \tau_{n} h^{\vphantom{-1}}_{\alpha_{p}}
\label{eqn:adjoint_surf_hol}
\end{equation}
for the action of the adjoint of a quantum surface holonomy on a single puncture state. Now, recall that, classically, surface holonomies are elements of $\SU$ and, as such, their adjoint is equal to the inverse of the surface holonomy. Furthermore, the inverse surface holonomy is equal to the surface holonomy associated to the (horizontally) inverse homotopy, 
\begin{equation}
\mathcal{W}^{-1}_H = \mathcal{W}_{H^{-1}}.
\label{eqn:hom_property_class_surf_hol}
\end{equation}
It is not clear, however, whether the latter property carries over to the quantum theory, since the horizontal inverse of a homotopy is only an inverse on the level of equivalence classes with respect to thin homotopy and the quantum surface holonomy operators are not well-defined on those equivalence classes. Therefore, let us next evaluate $\widehat{W}_{H^{-1}}$ on a single-puncture state and compare the result to \eqref{eqn:adjoint_surf_hol}.

The horizontal inverse of a homotopy $H(s,t)$ is given by $H^{-1}(s,t) = H(s,1-t)$. It is immediate to see from equation \eqref{eqn:surf_hol_explicit_formula} that the inverse homotopy induces the inverse orientation on the surface $S_H$. Other than that, the integral is over the same surface and therefore, comparing with \eqref{eqn:surfhol_singlepunct}, we have
\begin{equation}
\begin{split}
\widehat{\mathcal{W}}_{H^{-1}} \Biggl\vert_{\Heej} &= h^{-1}_{\tilde{\alpha}_{p}}\, \left( Q_{\text{DK}}\left[  \exp \left( - c \, T_{i} \kappa^{ij} E_{j}(p) \right)  \right] \Biggl\vert_{\Heej} \right) h^{\vphantom{-1}}_{\tilde{\alpha}_{p}} \\
&= \chi_{c}(j) \, \id_{\Heej} \otimes \, \mathbb{1}_{2} - i \chi_{s}(j) \, \kappa^{mn} \, \pi^{(j)}[\widehat{E}_{m}] \otimes \, h^{-1}_{\tilde{\alpha}_{p}} \tau_{n} h^{\vphantom{-1}}_{\tilde{\alpha}_{p}} \, .
\end{split}
\label{eqn:inverse_quant_surf_hol}
\end{equation}
This is almost identical to the action of the adjoint operator. Note, however, that the holonomies conjugating the generators of $\su$ in the second term in \eqref{eqn:inverse_quant_surf_hol} are calculated along different paths than in \eqref{eqn:adjoint_surf_hol}. However, in the absence of further punctures, the corresponding quantum states are related by a diffeomorphism. This indicates that the quantum analogue of \eqref{eqn:hom_property_class_surf_hol} might hold on single-puncture states at the diffeomorphism-invariant level. Another class of states on which it might hold  are those where the connection is flat on the part of the surface enclosed by  $\tilde{\alpha}^{-1}_p \circ \alpha_p$. This includes in particular the single-puncture states satisfying the IH boundary condition on $S_H$.
\subsection{Products of quantum surface holonomies}
Let us now consider products of surface holonomy operators. We will again restrict to the one puncture case. For the setup and notation see again figure~\ref{fig:single_puncture_case} and the text surrounding it. We will be working in the standard basis $\{\tau_i\}$ of $\su$ in which 
\begin{equation}
\label{eq:CartanKilling}
\kappa_{ik}=-2\delta_{ik}, \qquad \kappa^{ik}=-\frac{1}{2}\delta^{ik}. 
\end{equation}
We will also use the fact that the basis can be regarded as an intertwiner,  
\begin{equation}
\label{eq:adjoint_action}
g\tau_ig^{-1}= \tau_j\pi_1(g)^{j}{}_{i}. 
\end{equation}
We recall that the action of a single surface holonomy for the case depicted in the upper part of figure \ref{fig:single_puncture_case} (one-sided puncture) is given by
\begin{equation}
\label{eq:w_action}
(\widehat{\mathcal{W}_H})^A{}_B \ket{\jstate}= \xi_{c}(j)\,\delta^A_B\,\ket{\jstate}-i\,\xi_{s}(j)\,\ket{\ABhj} \, ,
\end{equation}
where 
\begin{equation}
\ket{\ABhj}= (h^{-1}_e \tau_i h^{\vphantom{-1}}_e)^A{}_B\; \kappa^{ik}\; \pi_j(\tau_kh_{e'})
=  \tau_j{}^A{}_B\, \pi_1(h_e^{-1})^j{}_i \; \kappa^{ik}\; \pi_j(\tau_kh^{\vphantom{-1}}_{e'}) \, . 
\end{equation}
The negative sign in \eqref{eq:w_action} is due to the fact that the edge $e'$ is assumed as incoming with respect to $S(H)$, and we have used  \eqref{eq:adjoint_action} to rewrite the holonomies connecting the puncture with the source of $H$. The double application of the surface holonomy operator then gives
\begin{equation}
\label{eq:ww_action}
(\widehat{\mathcal{W}_H} \widehat{\mathcal{W}_H}){}^A{}_B\ket{\jstate}
= \xi_{c}(j)^2\,\delta^A_B\,\ket{\jstate}
-2i\,\xi_{c}(j)\,\xi_{s}(j)\,\ket{\ABhj} 
-\xi{s}(j)^2 \,\ket{\ABhhj} 
\end{equation}
with 
\begin{equation*}
\ket{\ABhhj} =  (\tau_{i'}\tau_i)^A{}_B\, \pi_1(h_e^{-1})^i{}_k  \pi_1(h_e^{-1})^{i'}{}_{k'}\; \kappa^{kl}\kappa^{k'l'}\; \pi_j(\tau_l\tau_{l'}h^{\vphantom{-1}}_{e'}). 
\end{equation*}
This state is not linearly independent from $\ket{\jstate}$ and $\ket{\ABhj}$. In fact, it decomposes into a linear combination of them due to the fact that the latter are spin networks. We will use
\begin{equation}
\label{eq:TauProduct}
\tau_i\tau_{i'}=-\frac{1}{4}\delta_{ii'}\one+\frac{1}{2} \epsilon_{ii'k'} \delta^{kk'} \tau_k 
\end{equation}
to decompose the first product of $\tau$s. The orthogonality of  the matrix $\pi_1(h^{\vphantom{-1}}_e)$ simplifies the first resulting term, while for the second we obtain from the intertwiner property of $\epsilon$ 
\begin{equation*}
\epsilon_{i'jk} \pi_1(h)^j{}_{j'}\pi_1(h)^k{}_{k'}= \epsilon_{i''j'k'}\pi_1(h^{-1})^{i''}{}_{i'} \, . 
\end{equation*}
This gives 
\begin{align*}
\ket{\ABhhj} &=  \left( -\frac{1}{4}\delta_{kk'}\delta^A_B-\frac{1}{2}\epsilon_{kk'n} \pi_1(h)^n{}_{m'}\delta^{mm'}(\tau_m)^A{}_B  \right)\; \kappa^{kl}\kappa^{k'l'}\; \pi_j(\tau_l\tau_{l'}h_{e'})\\
&=  \frac{1}{2}\left( \frac{1}{4}\kappa_{kk'}\delta^A_B-\epsilon_{kk'n} \pi_1(h)^n{}_{m'}\delta^{mm'}(\tau_m)^A{}_B  \right)\; \kappa^{kl}\kappa^{k'l'}\; \pi_j(\tau_l\tau_{l'}h_{e'})\\
&=  \frac{1}{2}\left( \frac{1}{4}\kappa^{ll'}\delta^A_B-\epsilon^{ll'}{}_n \pi_1(h)^n{}_{m'}\delta^{mm'}(\tau_m)^A{}_B  \right)\; \pi_j(\tau_l\tau_{l'}h_{e'})\\
&=\frac{1}{8} \delta^A_B \pi_j(\kappa^{ll'}\tau_l\tau_{l'}h_{e'})-\frac{1}{8}  \pi_1(h)^n{}_{m'}\delta^{mm'}(\tau_m)^A{}_B \delta^k_n\;\pi_j(\tau_kh_{e'}), 
\end{align*}
where in the last line we have used 
\begin{equation*}
\epsilon^{ll'}{}_n\tau_l\tau_{l'}=\frac{1}{4} \epsilon_{ll'n}\delta^{lm}\delta^{l'm'} \tau_m\tau_{m'}
=\frac{1}{8}\epsilon_{ll'n}\delta^{lm}\delta^{l'm'} \epsilon_{mm'k'}\delta^{kk'}\tau_k 
= \frac{1}{4}\delta^k_n\tau_k 
\end{equation*}
because of \eqref{eq:CartanKilling} and \eqref{eq:TauProduct}. We can further simplify
\begin{align*}
\ket{\ABhhj} &= \frac{1}{8} \Delta_j \delta^A_B \ket{\jstate}+\frac{1}{4}  \pi_1(h)^{nm}(\tau_m)^A{}_B \delta^k_n\;\pi_j(\tau_kh_{e'})\\
&=\frac{1}{8} \Delta_j \delta^A_B \ket{\jstate}+\frac{1}{4}  \pi_1(h^{-1})^{mn}(\tau_m)^A{}_B \delta^k_n\;\pi_j(\tau_kh_{e'})\\
&=\frac{1}{8} \Delta_j \delta^A_B \ket{\jstate}+\frac{1}{4}  \pi_1(h^{-1})^{m}{}_{n'}(\tau_m)^A{}_B\kappa^{nn'} \;\pi_j(\tau_nh_{e'})\\
&=\frac{1}{8} \Delta_j \delta^A_B \ket{\jstate}+\frac{1}{4}  \ket{\ABhj}, 
\end{align*}
with  
\begin{equation*}
\pi_j(\kappa^{ab}\tau_a\tau_b)= \frac{1}{2}\, j(j+1) \one =: \Delta_j\,\one.
\end{equation*}
Thus we find
\begin{align}
(\widehat{\mathcal{W}_H} \widehat{\mathcal{W}_H}){}^A{}_B\ket{\jstate}
&= \left(\xi_{c}(j)^2-\frac{1}{8} \Delta_j \,\xi_{s}(j)^2 \right) \,\delta^A_B\,\ket{\jstate}
-\left(2i \,\xi_{c}(j)\,\xi_{s}(j) + \frac{1}{4}\xi_{s}(j)^2\right) \,\ket{\ABhj} \\
\tr(\widehat{\mathcal{W}_H} \widehat{\mathcal{W}_H}) \ket{\jstate}
&=  \left(2\, \xi_{c}(j)^2-\frac{1}{4} \Delta_j \,\xi_{s}(j)^2 \right) \ket{\jstate}
\label{eq:tr_onesided}
\end{align}
for the product of two surface holonomies and the trace thereof. We also see from this calculation that the space spanned by the states $\ket{\jstate}$ and $\ket{\ABhj}$ is closed under the action of the surface holonomy operator. This action was already given in \eqref{eq:w_action} for the former, while on the latter state, the action is explicitly given by
\begin{equation}
\begin{split}
(\widehat{\mathcal{W}_H})^A{}_C \ket{\CBhj} &= \xi_{c}(j)\,\ket{\ABhj}-i\,\xi_{s}(j)\,\ket{\ABhhj}\\
&= -\frac{i}{8} \, \Delta_j \, \xi_{s}(j) \, \delta^A_B \, \ket{\jstate} + \left( \xi_{c}(j) - \frac{i}{4} \, \xi_{s}(j) \right) \, \ket{\ABhj} \, .
\end{split}
\end{equation}

For the case that the holonomy runs through the puncture (lower part of figure \ref{fig:single_puncture_case}), there are some changes to the above result. The action of the surface holonomy is now  
\begin{equation}\label{eq:w_action2}
(\widehat{\mathcal{W}_H})^A{}_B \ket{\jjstate}= \chi_{c}(j)\,\delta^A_B\,\ket{\jjstate}-2i \, \chi_{s}(j)\,\ket{\ABhjj} \, ,
\end{equation}
where 
\begin{equation*}
\ket{\ABhjj} =  \tau_j{}^A{}_B\, \pi_1(h_e^{-1})^j{}_i \; \kappa^{ik}\; \pi_j(h_{e''}\tau_kh_{e'}) \, . 
\end{equation*}
Acting a second time, one obtains
\begin{equation*}
(\widehat{\mathcal{W}_H} \widehat{\mathcal{W}_H}){}^A{}_B\ket{\jjstate}
=  \chi_{c}(j)^2\,\delta^A_B\,\ket{\jjstate}
-4i\,\chi_{c}(j)\,\chi_{s}(j) \,\ket{\ABhjj} 
-2 \, \chi_{s}(j)^2 \ket{\ABhhjj} \, ,
\end{equation*}
where now
\begin{align*}
\ket{\ABhhjj} &=  (\tau_{i'}\tau_i)^A{}_B\, \pi_1(h_e^{-1})^i{}_k  \pi_1(h_e^{-1})^{i'}{}_{k'}\; \kappa^{kl}\kappa^{k'l'}\; \pi_j(h_{e''}(\tau_l\tau_{l'}+\tau_{l'}\tau_{l})h_{e'})\\
&=  (\tau_{i'}\tau_i)^A{}_B\, \pi_1(h_e^{-1})^i{}_k  \pi_1(h_e^{-1})^{i'}{}_{k'}\; \kappa^{kl}\kappa^{k'l'} \left(2\pi_j(h_{e''}\tau_l\tau_{l'}h_{e'})+\epsilon_{l'ln'}\delta^{nn'} \pi_j(h_{e''}\tau_nh_{e'})\right)\\
&= \frac{1}{4} \Delta_j \delta^A_B \ket{\jjstate}+\frac{1}{2}  \ket{\ABhjj}+ 
 (\tau_{i'}\tau_i)^A{}_B\, \pi_1(h_e^{-1})^i{}_k  \pi_1(h_e^{-1})^{i'}{}_{k'}\; \kappa^{kl}\kappa^{k'l'} \;\epsilon_{l'ln'}\delta^{nn'} \pi_j(h_{e''}\tau_nh_{e'})\\
 &= \frac{1}{4} \Delta_j \delta^A_B \ket{\jjstate}+\frac{1}{2}  \ket{\ABhjj}-
 (\tau_{i'}\tau_i)^A{}_B\, \pi_1(h_e^{-1})^{ik}  \pi_1(h_e^{-1})^{i'k'}\;\epsilon_{n'kk'}\delta^{nn'} \pi_j(h_{e''}\tau_nh_{e'})\\
 &= \frac{1}{4} \Delta_j \delta^A_B \ket{\jjstate}+\frac{1}{2}  \ket{\ABhjj}-
 (\tau_{i'}\tau_i)^A{}_B\, \pi_1(h_e)^{ki}  \pi_1(h_e)^{k'i'}\;\epsilon_{n'kk'}\delta^{nn'} \pi_j(h_{e''}\tau_nh_{e'})\\
 &= \frac{1}{4} \Delta_j \delta^A_B \ket{\jjstate}+\frac{1}{2}  \ket{\ABhjj}-
 (\tau_{i'}\tau_i)^A{}_B\, \kappa^{ik}\kappa^{i'k'}\epsilon_{lkk'} \pi_1(h_e^{-1})^{l}{}_{n'} \delta^{nn'} \pi_j(h_{e''}\tau_nh_{e'})\\
 &= \frac{1}{4} \Delta_j \delta^A_B \ket{\jjstate}+\frac{1}{2}  \ket{\ABhjj}
 -\frac{1}{2}\tau_l{}^A{}_B\; \pi_1(h_e^{-1})^{l}{}_{n'} \kappa^{nn'} \pi_j(h_{e''}\tau_nh_{e'})\\
 &= \frac{1}{4} \Delta_j \delta^A_B \ket{\jjstate}. 
 \end{align*}
We thus get 
\begin{align}
(\widehat{\mathcal{W}_H} \widehat{\mathcal{W}_H}){}^A{}_B\ket{\jjstate}
&= \left( \chi_{c}(j)^2 - \frac{1}{2} \Delta_j \,\chi_{s}(j)^2 \right) \,\delta^A_B\,\ket{\jjstate}
-4i \, \chi_{c}(j)\, \chi_{s}(j) \, \ket{\ABhjj} \\
\tr (\widehat{\mathcal{W}_H} \widehat{\mathcal{W}_H})\ket{\jjstate}
&= \left(2 \, \chi_{c}(j)^2 - \Delta_j \, \chi_{s}(j)^2 \right) \,\ket{\jjstate}\, . 
\end{align}
We are also interested in products involving the (matrix and operator) adjoint defined in \eqref{eqn:adjoint_surf_hol}. We have 
\begin{align*}
(\mathcal{W}^\dagger_H)^A{}_B \ket{\jstate}&= \overline{\xi_{c}(j)}\,\delta^A_B\,\ket{\jstate} +i\,\overline{\xi_{s}(j)}\,\ket{\ABhj} \, ,\\ 
(\mathcal{W}^\dagger_H)^A{}_B \ket{\jjstate}&= \overline{\chi_{c}(j)}\,\delta^A_B\,\ket{\jjstate}+2i\,\overline{\chi_{s}(j)}\,\ket{\ABhjj} 
\end{align*}
and hence
\begin{align*}
(\widehat{\mathcal{W}_H}\widehat{\mathcal{W}_H}^\dagger)^A{}_B \ket{\jstate}
&= \left(|\xi_{c}(j)|^2+\frac{1}{8} \Delta_j \,|\xi_{s}(j)|^2 \right) \,\delta^A_B\,\ket{\jstate}
-\left(2\re((-i)\xi_{c}(j)\overline{\xi_{s}(j)}) + \frac{1}{4}|\xi_{s}(j)|^2\right) \, \ket{\ABhj}\, ,\\
(\mathcal{W}_H \mathcal{W}^\dagger_H){}^A{}_B\ket{\jjstate}
&= \left(|\chi_{c}(j)|^2+ \frac{1}{2} \Delta_j \,|\chi_{s}(j)|^2 \right) \,\delta^A_B \,\ket{\jjstate}
-4\re((-i)\chi_{c}(j)\overline{\chi_{s}(j)}) \ket{\ABhjj}.
\end{align*}
Since we have 
\begin{equation*}
\overline{\xi_{c}(j)}=\xi_{c}(j)\qquad \overline{\xi_{s}(j)}=\xi_{s}(j)
\end{equation*}
and
\begin{equation*}
\overline{\chi_{c}(j)}=\chi_{c}(j)\qquad \overline{\chi_{s}(j)}=\chi_{s}(j)\, ,
\end{equation*}
this simplifies to 
\begin{align}
(\widehat{\mathcal{W}_H}\widehat{\mathcal{W}_H})^\dagger{}^A{}_B \ket{\jstate}
&= \left(\xi_{c}(j)^2+\frac{1}{8} \Delta_j \,\xi_{s}(j)^2 \right) \,\delta^A_B\,\ket{\jstate}
 - \frac{1}{4}\xi_{s}(j)^2 \;\ket{\ABhj},\\
(\widehat{\mathcal{W}_H} \widehat{\mathcal{W}_H}^\dagger){}^A{}_B\ket{\jjstate}
&= \left(\chi_{c}(j)^2 + \frac{1}{2} \Delta_j \,\chi_{s}(j)^2 \right) \,\delta^A_B\,\ket{\jjstate}\\
&= \det{}_\delta (\widehat{\mathcal{W}_H}) \;\delta^A_B\;\ket{\jjstate} \, .
\end{align}
One can also ask about products of surface holonomies that are not contracted, and in particular, about their commutators. These questions can be answered using the results of section \ref{se:matrix_elements}. 
In particular, \eqref{eq:one_sided_commutator}-\eqref{eq:two_sided_commutator} give the commutators. We would like to point out that these commutators vanish on two-sided gauge-invariant punctures.
\begin{equation}
\left[\widehat{(\mathcal{W}_H})^A{}_B, \widehat{(\mathcal{W}_H})^C{}_D  \right] \ket{\jjstate}= 0. 
\end{equation}
This is interesting, since it shows that on these states, the surface holonomies have the same adjointness and commutation relations as normal holonomy operator in the holonomy-flux algebra of loop quantum gravity. 

\subsection{Traces, relations, other irreducible representations} 
We have already considered traces of products of surfaces holonomies. We will now discuss traces a bit more systematically. Consider the trace of a single surface holonomy,  
\begin{equation}
\tr(\widehat{\mathcal{W}_H})= \mathcal{a}+\mathcal{a}^\dagger. 
\end{equation}
There is obviously no ordering ambiguity, and the traces are automatically gauge-invariant:
\begin{equation}
\tr(g\,\widehat{\mathcal{W}_H}\,g^{-1})= \tr(\widehat{\mathcal{W}_H}) \qquad \text{ for } g\in \SU\,. 
\end{equation}
On single punctures this implies $\tr(\widehat{\mathcal{W}_H})=\tr(\widehat{W_p}) $, and hence
\begin{equation}
\tr \left. \widehat{\mathcal{W}_H} \right\rvert_\Hpju= 2\xi_{c}(j) \id_{\Hpju}, \qquad \tr \left. \widehat{\mathcal{W}_H} \right\rvert_\Hpjj= 2\chi_{c}(j) \id_{\Hpjj}\,.
\end{equation}
  
We note that there are classical relations between the objects we have considered so far. For example, the relation
\begin{equation}
\det(W)=\frac{1}{2}\left((\tr W)^2 -\tr (W^2) \right)
\end{equation}
holds for any 2x2 matrix $W$. This is a relation which is intact in the quantum theory. For example, we can show that
\begin{align*}
\frac{1}{2}\left((\tr \widehat{\mathcal{W}_H})^2 -\tr (\widehat{\mathcal{W}_H}^2) \right)\,\ket{\jjstate}
&= \frac{1}{2}\left(4\, \chi_{c}(j)^2 - 2\, \chi_{c}(j)^2 + \Delta_j \, \chi_{s}(j)^2 \right)\,\ket{\jjstate}\\
&= \frac{1}{2}\left(2 \, \chi_{c}(j)^2 + \frac{j(j+1)}{2} \, \chi_{s}(j)^2\right)\,\ket{\jjstate}\\
&=\left( \chi_{c}(j)^2 + \frac{j(j+1)}{4} \, \chi_{s}(j)^2\right)\,\ket{\jjstate}\\
&=\det (\widehat{\mathcal{W}_H}) \;\ket{\jjstate} \, . 
\end{align*}
For the two-sided puncture, we similarly have
\begin{align*}
\frac{1}{2}\left((\tr \widehat{\mathcal{W}_H})^2 -\tr (\widehat{\mathcal{W}_H}^2) \right)\,\ket{\jstate}
&= \frac{1}{2}\left(4\, \xi_{c}(j)^2 - 2\, \chi_{c}(j)^2 + \frac{1}{4}\Delta_j \, \chi_{s}(j)^2 \right)\,\ket{\jstate}\\
&=\left( \chi_{c}(j)^2 + \frac{1}{8}\Delta_j \, \chi_{s}(j)^2\right)\,\ket{\jstate}\\
&=\det{}_\frac{1}{2} (\widehat{\mathcal{W}_H}) \;\ket{\jstate}. 
\end{align*}
We can also use the traces above to find expressions for the traces of surface holonomies in different representations of $\SU$. For example, 
\begin{equation}
 \tr(\pi_1(g)):=\frac{1}{2}\left[\tr(g^2)+\tr(g)^2\right]\qquad \text{ for } g\in \SU \, .
\end{equation}
We can thus define 
\begin{equation*}
 \tr \pi_1(\widehat{\mathcal{W}_H)} =\frac{1}{2}\left[\tr(\widehat{\mathcal{W}_H}^2)+\tr(\widehat{\mathcal{W}_H})^2\right] \, ,
\end{equation*}
and we find 
\begin{equation}
\begin{split}
\tr \pi_1(\widehat{\mathcal{W}_H)}\,\ket{\jstate}
&= \frac{1}{2}\left[6\,\xi_c^2(j)-\frac{1}{4}\Delta_j \, \xi_s^2(j)\right]\,\ket{\jstate}
=\frac{1}{2}\left[8\,\xi_c^2(j)-2\det{}_\frac{1}{2}\widehat{\mathcal{W}_H}  \right]\,\ket{\jstate}\\
&=\left[4\,\xi_c^2(j)-\det{}_\frac{1}{2}\widehat{\mathcal{W}_H}  \right]\,\ket{\jstate} \, . 
\end{split}
\end{equation}
Similarly, 
\begin{equation}
\tr \pi_1(\widehat{\mathcal{W}_H)}\,\ket{\jjstate} 
=\left[4\,\chi_c^2(j)-\det{}_\delta\widehat{\mathcal{W}_H}  \right]\,\ket{\jjstate}. 
\end{equation}
Assuming $\det{}_\delta\widehat{\mathcal{W}_H}=1$, we note that both eigenvalues are of the form
\begin{equation}
\lambda_j=4\cos^2 \left((2j+1)\theta\right) -1 \qquad \text { with } \theta = -\frac{ic}{8},  -\frac{ic}{4} \, . 
\end{equation}
This is interesting because it can be rewritten as
\begin{equation}
\begin{split}
\lambda_j&=3- 4\sin^2 \left((2j+1)\theta\right)=\frac{3\sin^2 \left((2j+1)\theta\right)- 4\sin^3 \left((2j+1)\theta\right)}{\sin^2 \left((2j+1)\theta\right)}\\
&= \frac{\sin \left(3 (2j+1)\theta\right)}{\sin \left((2j+1)\theta\right)}\\
&= \frac{q^{3(2j+1)}-q^{-3(2j+1)}}{q-q^{-1}}\\
&= \frac{[3(2j+1)]_q}{[2j+1]_q} \qquad \text{ with } q=e^{i\theta}.  
\end{split}
\end{equation}
Here we have used the quantum integers 
\begin{equation}
[n]_q = \frac{q^n-q^{-n}}{q-q^{-1}}. 
\end{equation}
The eigenvalues are quotients of the Chern-Simons expectation value for holonomies around the Hopf-link and the unlink, respectively, \cite{Engle:2011vf} and thus arguably the expectation value 
of a surface holonomy around a Chern-Simons puncture.

\section{Application to black holes} 
\label{se_black_hole}

In this section, we want to come back to our original motivation for investigating quantum operators associated to surface holonomies. Namely, we want to use our quantum surface holonomies to quantize the isolated horizon boundary condition (IHBC)
\begin{equation}
\iota_{\mathcal{H}}^{*}F=C\, \iota_{\mathcal{H}}^{*} \,(*E) \, .
\end{equation}
As already stated in section \ref{sec:ClassSurfHol}, this condition is equivalent to
\begin{equation}
\mathcal{W}_{H}[A,C\, \iota_{\mathcal{H}}^{*} \,(*E)] = h_{H(1,.)}[A] \, ,
\label{eqn:QIHBC}
\end{equation}
which has to be satisfied for all homotopies $H$, s.t. the surface $S_H$ lies entirely within $\mathcal{H}$. Recall our assumption that all homotopies $H$ start from the trivial path at the point at which our surface holonomies transform. Also, note that we can evaluate these conditions for any two-dimensional surface $\mathcal{H}$. In this section, we will take $\mathcal{H}$ to be homeomorphic to a 2-sphere, but we will not assume that it is the spatial section of an isolated horizon! Ideally, we would find states in the quantum theory on which the quantum version of the IHBC is exactly satisfied. However, from our results in the previous section we can conclude that one-puncture states and two-puncture states cannot be solutions to this quantum operator equation. The reason is that holonomies are quantized as multiplication operators in LQG. Therefore, they act by multiplying the state with an element of $\SU$, which necessarily has unit determinant. The determinant of our quantum surface holonomies, however, does not equal unity for any choice of spin on such states. Nevertheless, we have seen that for some spins the determinant is very close to unity, which indicates that their behavior may be similar to real $\SU$ elements on some states. Therefore, instead of trying to implement the quantum isolated horizon boundary condition (QIHBC) exactly, we will take it as a definition for some kind of \emph{quantum holonomies} replacing the standard holonomy operators on $\mathcal{H}$.\footnote{This approach is supported by the fact that, in the original works on quantum isolated horizons in LQG, the holonomies on the horizon also did not take values in $\SU$ but rather in a quantum group deformation thereof.} We will regard states on which these quantum holonomies behave closely to classical holonomies as solutions to the QIHBC.\\ 

Since we are interested in surfaces of spherical topology here, we are faced with the topological property of such surfaces that a circle on a 2-sphere $S^2$ forms the boundary of two distinct surfaces. In terms of homotopies, this translates to the existence of two distinct equivalence classes of homotopies between any two paths sharing their endpoints. Let us denote representatives of these equivalence classes by $H_1$ and $H_2$, respectively. This now introduces an ambiguity in the definition of holonomies via the QIHBC. Consider any circular path $\gamma$ on $S^2$. Without loss of generality, we will assume $\gamma$ to coincide with the equator of $S^2$. We can then define the holonomy $h_\gamma$ using equation \eqref{eqn:QIHBC} in two different ways: either by
\begin{equation}
h_\gamma = \mathcal{W}_{H_1}
\end{equation}
or by
\begin{equation}
h_\gamma = \mathcal{W}_{H_2} \, ,
\end{equation}
where $H_1$ and $H_2$ now denote the homotopies from the constant path at the starting and end point $p$ of $\gamma$ to $\gamma$ by passing over the northern and southern hemisphere of $S^2$, respectively (see also figures \ref{fig:single_puncture_BH} and \ref{fig:double_puncture_BH}). Therefore, only states on which we have
\begin{equation}
\mathcal{W}_{H_1} = \mathcal{W}_{H_2}
\label{eqn:surf-hol-consistency-condition-on-sphere}
\end{equation}
qualify as candidates for implementing the QIHBC. Furthermore, classical holonomies satisfy the relation
\begin{equation}
h_{\gamma^{-1}} = h_{\gamma}^{-1} \, .
\end{equation}
If we want this property to hold also for our holonomies that are defined in terms of surface holonomies, then we need to restrict ourselves to states on which the relation
\begin{equation}
\mathcal{W}_{H_1} = h_{\gamma} = h_{\gamma^{-1}}^{-1} = \mathcal{W}_{H_2^{-1}}^{-1}
\end{equation}
holds. Here, $H_2^{-1}$ denotes the horizontal inverse of $H_2$ in the path 2-groupoid, i.e. it is a homotopy from the trivial path at $x_0$ to $\gamma^{-1}$. We can equivalently write this relation between surface holonomies on a 2-sphere as
\begin{equation}
\mathcal{W}_{H_1} \mathcal{W}_{H_2^{-1}} = \mathbb{1}_{2} = \mathcal{W}_{H_2^{-1}} \mathcal{W}_{H_1} \, ,
\label{eqn:two-puncture-consistency-condition}
\end{equation}
where we have now avoided the use of inverse surface holonomies. In the following, we will use the quantum version of \eqref{eqn:two-puncture-consistency-condition} as ameasures for how closely the holonomies defined via the quantized IHBC, i.e., in terms of quantum surface holonomies, resemble classical path holonomies. In the remainder of this section, we will analyze quantized versions of \eqref{eqn:surf-hol-consistency-condition-on-sphere} and \eqref{eqn:two-puncture-consistency-condition}. Throughout this section, we will only consider two-edge punctures. We deem this reasonable because, as we have seen in the previous section, the behavior of the quantum surface holonomy operators resembles that of their classical counterparts more closely when evaluated on this type of punctures.\\


\subsection{Single two-edge puncture}
Let us first consider the simple case of a single puncture as depicted in figure~\ref{fig:single_puncture_BH}. 
\begin{wrapfigure}{l}{0.3\textwidth}
\includegraphics[scale=0.55, trim={0 1.5cm 0 0},clip]{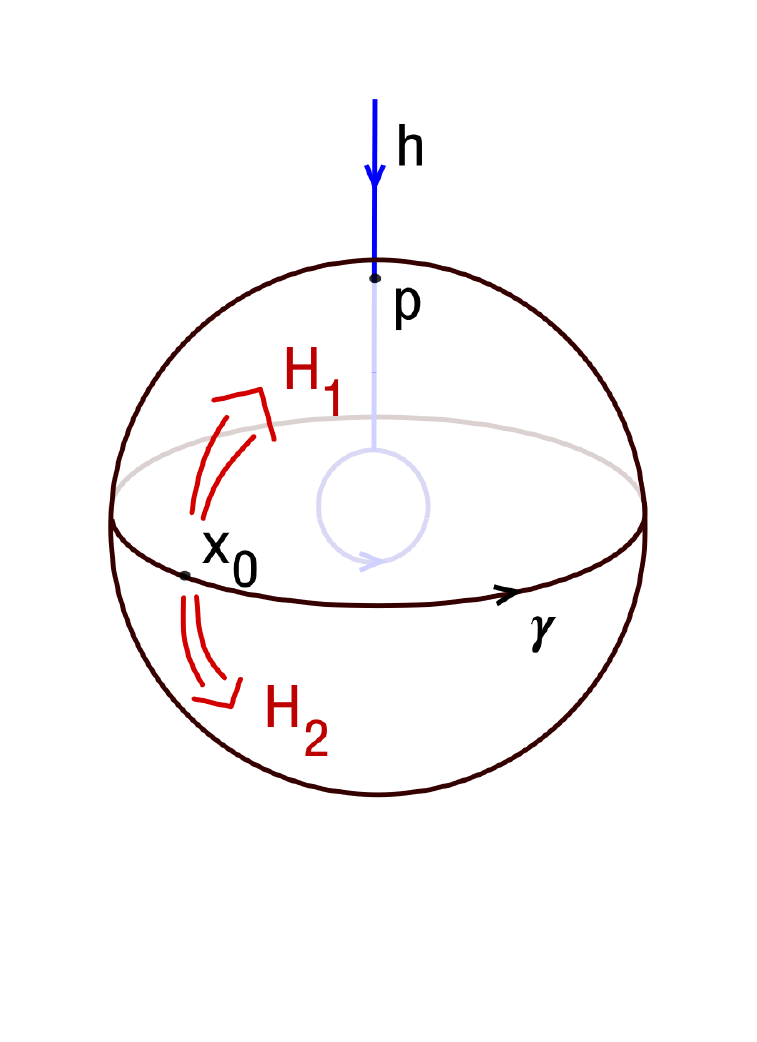}
\caption{Single puncture state}
\label{fig:single_puncture_BH}
\end{wrapfigure}
We will refer to this one-puncture state as $\Psi_{1P}$. Since there is no puncture on the southern hemisphere, we have
\begin{equation}
\widehat{\mathcal{W}}_{H_2} \, \ket{\Psi_{1P}} = \mathbb{1}_{2} \, \ket{\Psi_{1P}} \, .
\end{equation}
On the other hand, we have
\begin{equation}
\widehat{\mathcal{W}}_{H_1} \, \ket{\Psi_{1P}} = \chi_{c}(j) \, \mathbb{1}_{2} \ket{\Psi_{1P}} - 2i \, \chi_{s}(j) \, \kappa^{mn} \, \tau_{n} \, \pi^{(j)} \left[ \widehat{\mathcal{E}}_{m} \right] \ket{\Psi_{1P}} \, .
\end{equation}
Therefore, the consistency condition \eqref{eqn:surf-hol-consistency-condition-on-sphere} is only satisfied for spins $j$, such that
\begin{equation}
\chi_{c}(j) = 1 \qquad \text{and} \qquad \chi_{s}(j) = 0 \, .
\end{equation}
As there are no integer or half-integer spins satisfying this condition, we can already conclude that there are no single-puncture solutions to the quantized IHBC. Note that for $j=0$, the condition reduces to
\begin{equation}
\chi_{c}(0) = 1 \, ,
\end{equation}
since $\pi^{(j)}(\widehat{E}_{i})$ equals zero in the $j=0$ representation. This is of course trivially satisfied, as quantum surface holonomies act as the identity on punctures with $j=0$, and it implies that spin network graphs that do not puncture the sphere under consideration are solutions to the quantized IHBC. However, the area eigenvalue of the sphere vanishes on such states and thus the sphere would be unobservable.

\subsection{Double two-edge puncture}

\begin{wrapfigure}{r}{6cm}
\includegraphics[scale=0.55]{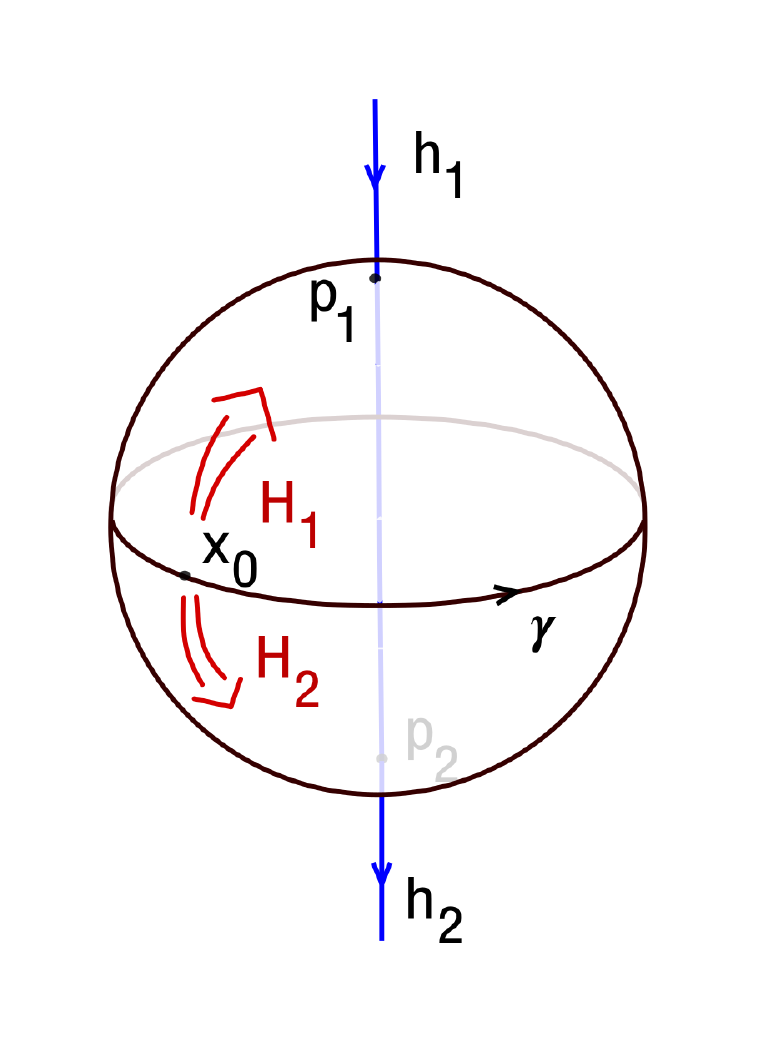}
\caption{Double puncture case}
\label{fig:double_puncture_BH}
\end{wrapfigure}
Let us go one step further and take a look at a situation with two punctures. Consider a spin network consisting of a single edge that punctures the sphere $H$ in two points labeled $p_1$ and $p_2$. Without loss of generality, we assume that $p_1$ and $p_2$ coincide with the north and south poles of the sphere (as illustrated in figure~\ref{fig:double_puncture_BH}). This state will be referred to as $\Psi_{2P}$. Let us start again with checking the consistency condition \eqref{eqn:surf-hol-consistency-condition-on-sphere}. This time we have
\begin{equation}
\widehat{\mathcal{W}}_{H_1} \ket{\Psi_{2P}} = \chi_{c}(j) \, \mathbb{1}_{2} \ket{\Psi_{2P}} - 2i \, \chi_{s}(j) \kappa^{mn} \tau_{n} \, \pi^{(j)} \left[ \widehat{\mathcal{E}}_{m} \right] \ket{\Psi_{2P}} \, .
\end{equation}
Recall that the minus sign before the second term is due to the different orientations of the spin network edge puncturing the surface at $p_1$ and the normal on the surface whose orientation is induced by using the homotopy $H_1$ as a parametrization. Noting that on the southern hemisphere the normal induced by $H_2$ again points upwards, we can immediately conclude that we get the same expression
\begin{equation}
\widehat{\mathcal{W}}_{H_2} \ket{\Psi_{2P}} = \chi_{c}(j) \, \mathbb{1}_{2} \ket{\Psi_{2P}} - 2i \, \chi_{s}(j) \kappa^{mn} \tau_{n} \, \pi^{(j)} \left[ \widehat{\mathcal{E}}_{m} \right] \ket{\Psi_{2P}}
\end{equation}
\begin{wrapfigure}{r}{6cm}
\includegraphics[scale=0.55]{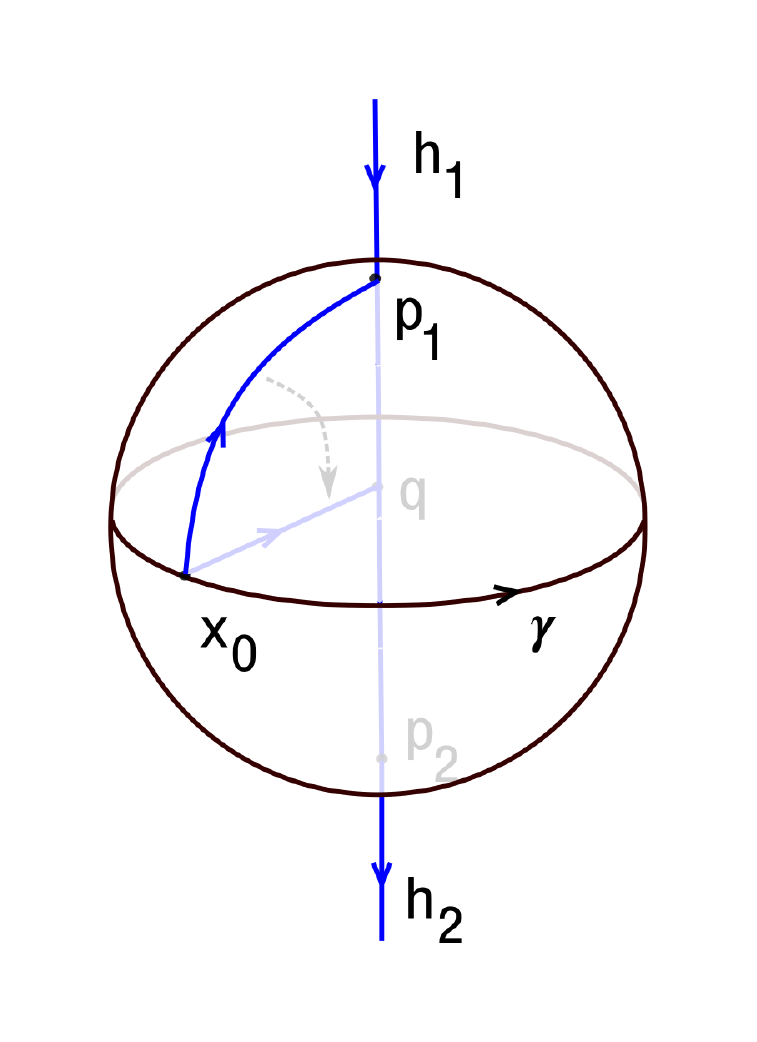}
\caption{moving of holonomies}
\label{fig:hol_move}
\end{wrapfigure}
for the surface holonomy associated to $H_2$. However, they are not exactly the same. We have hidden the conjugation with the holonomies from and to the distinguished point $x_0$ of the surface by writing $\mathcal{E}$ instead of $E$. The paths along which these holonomies are calculated are different for $H_1$ and $H_2$! One may hope that this difference disappears at the diffeomorphism invariant level. For the moment, we will assume that we can move the holonomy from $x_0$ to $p_1$ as illustrated in figure~\ref{fig:hol_move} (and accordingly for $p_2$).\\

Let us turn our attention to the other consistency condition \eqref{eqn:two-puncture-consistency-condition}. As discussed above, this condition encodes the property
\begin{equation}
h_{\gamma^{-1}} = h_{\gamma}^{-1}
\end{equation}
of classical path holonomies. In terms of surface holonomies, we now need to consider the two conditions
\begin{equation}
\widehat{\mathcal{W}}_{H_1} \widehat{\mathcal{W}}_{H_2^{-1}} \Psi_{2P} = \mathbb{1}_{2} \Psi_{2P}
\label{eqn:cond2a}
\end{equation}
and
\begin{equation}
\widehat{\mathcal{W}}_{H_2^{-1}} \widehat{\mathcal{W}}_{H_1} \Psi_{2P} = \mathbb{1}_{2} \Psi_{2P} \, .
\label{eqn:cond2b}
\end{equation}
One might be tempted now to conclude from our previous results that these two conditions are satisfied on states where the surface holonomies have unit determinant. However, this is not the case. While we have found in the previous chapter that
\begin{equation}
(\widehat{\mathcal{W}_H} \widehat{\mathcal{W}_H}^\dagger){}^A{}_B\ket{\jjstate}
= \det{}_\delta (\widehat{\mathcal{W}_H}) \;\delta^A_B\;\ket{\jjstate}
\end{equation}
and
\begin{equation}
\widehat{\mathcal{W}_H}^{-1} \ket{\jjstate} \approx \widehat{\mathcal{W}_{H^{-1}}} \ket{\jjstate} \, ,
\label{eqn:inv_homot_rel}
\end{equation}
which, when combined with our statements in the previous paragraph, would seem to imply that
\begin{equation}
\widehat{\mathcal{W}}_{H_1}^{\dagger} \ket{\jjstate} = \widehat{\mathcal{W}}_{H_2}^{\dagger} \ket{\jjstate} \quad \Longrightarrow \quad \widehat{\mathcal{W}}_{H_1} \widehat{\mathcal{W}}_{H_1}^{\dagger} \ket{\jjstate} = \widehat{\mathcal{W}}_{H_1} \widehat{\mathcal{W}}_{H_2}^{\dagger} \ket{\jjstate}
\quad \Longrightarrow \quad  \det{}_\delta (\widehat{\mathcal{W}}_{H_1}) \, \ket{\jjstate} = \widehat{\mathcal{W}}_{H_1} \widehat{\mathcal{W}}_{H_2^{-1}} \, \ket{\jjstate} \, ,
\label{eqn:false_argument}
\end{equation}
it is important to remember that the first equation in this deduction only holds if a certain diffeomorphism is applied. Therefore, it would probably be more precise to write it as
\begin{equation}
D_{2 \longleftarrow 1}\,\widehat{\mathcal{W}_{H_1}} \ket{\jjstate} = \widehat{\mathcal{W}_{H_2}} \ket{\jjstate} \, ,
\end{equation}
where the diffeomorphism $D_{2 \longleftarrow 1}$ moves the attachment point from $p_1$ to $p_2$. With this notation it is obvious that the argument \eqref{eqn:false_argument} already fails in the first step! We will therefore have to evaluate conditions \eqref{eqn:cond2a} and \eqref{eqn:cond2b} independently. Let us start with the latter.\\

In order to evaluate the left hand side, let us recall the action of the surface holonomy operator again, which was calculated in chapter~\ref{sec:quant} to be
\begin{equation}
\left( \widehat{\mathcal{W}}_{H_1} \right)^{A}_{~B} \ket{\jjstate} = \chi_{c}(j) \, \ket{\snsUnchanged} - 2i \,\chi_{s}(j) \, \ket{\snsUpper} \, .
\end{equation}
Here, we have introduced the graphical notation
\begin{align}
\ket{\snsUnchanged} &= \delta^{B}_{~A} \, h_{2} h_{p_1 \rightarrow p_2} h_{1} \, ,\\
\ket{\snsUpper} &= \left( \tilde{h}^{-1} \tau_{i} \tilde{h} \right)^{B}_{~A} \, \kappa^{il} \, h_{2} h_{p_1 \rightarrow p_2} T_{l} h_{1} \, ,
\end{align}
where $h_{p_1 \rightarrow p_2}$ denotes the holonomy along the segment of the spin network edge between $p_1$ and $p_2$, $T_{l}$ is a generator of the Lie algebra $\su$ in the spin~j representation and $\tilde{h}$ denotes a holonomy from the starting point $x_{0}$ of the surface holonomy to the point $p_1$. Analogously, we will also use
\begin{align}
\ket{\snsLower} &= \left( \tilde{\tilde{h}}^{-1} \tau_{i} \tilde{\tilde{h}} \right)^{B}_{~A} \, \kappa^{il} \, h_{2} T_{l} h_{p->q} h_{1}\, ,\\
\ket{\snsBoth} &= \left( \tilde{\tilde{h}}^{-1} \tau_{m} \tilde{\tilde{h}} \tilde{h}^{-1} \tau_{i} \tilde{h} \right)^{B}_{~A} \, \kappa^{il} \kappa^{mn} \, h_{2} T_{n} h_{p->q} T_{l} h_{1}\, .
\end{align}
We can now calculate how the operator appearing on the left hand side of condition~\eqref{eqn:cond2b} acts on the spin network state under consideration and we obtain
\begin{equation}
\begin{split}
\left( \W_{H_{2}^{-1}} \right)^{A}_{~C} \, \left( \W_{H_{1}} \right)^{C}_{~B} \, \ket{\jjstate} &= \left( \W_{H_{2}^{-1}} \right)^{A}_{~C} \left[ \chi_{c}(j) \, \ket{\snsUnchanged} - 2i \, \chi_{s}(j) \, \ket{\snsUpper} \right]\\
 &= \chi_{c}(j) \, \chi_{c}(j) \ket{\snsUnchanged} + 2i \, \chi_{c}(j) \, \chi_{s}(j) \ket{\snsLower}\\
 &\,- 2i \, \chi_{s}(j) \, c_{c}(j) \ket{\snsUpper} + 4\, \chi_{s}(j) \, \chi_{s}(j) \ket{\snsBoth}\, .
\end{split}
\end{equation}
The spin network states appearing in this result are not independent. We will assume
\begin{equation}
\alpha^{-1} \, \ket{\snsUpper} = \ket{\snsSingleDeformed} = \beta^{-1} \, \ket{\snsLower} \, ,
\end{equation}
where
\begin{equation}
\ket{\snsSingleDeformed} = \left( h^{-1} \tau_{i} h \right)^{B}_{~A} \, \kappa^{il} \, h_{2} h_{x->q} T_{l} h_{p->x} h_{1} \, ,
\end{equation}
and, consequently, we get
\begin{equation}
\ket{\snsBoth} = \alpha \beta \, \ket{\snsBothDeformed} \, .
\end{equation}
This last state can be expressed as a linear combination of the other two via
\begin{equation}
\begin{split}
\ket{\snsBothDeformed} &= \left( \tau_{b} \tau_{a} \right)^{B}_{~A} \, \pi_{1}(h^{-1})^{a}_{~i} \, \pi_{1}(h^{-1})^{b}_{~j} \, \kappa^{im} \kappa^{jn} \, h_{2} h_{x->q} T_{n} T_{m} h_{p->x} h_{1}\\
&= \left( - \frac{1}{4} \delta_{ab} \, \delta^{B}_{~A} + \frac{1}{2} \epsilon_{ba}^{~~c} \, \tau_{c~\,A}^{~B} \right) \, \pi_{1}(h^{-1})^{a}_{~i} \, \pi_{1}(h^{-1})^{b}_{~j} \, \kappa^{im} \kappa^{jn} \, h_{2} h_{x->q} T_{n} T_{m} h_{p->x} h_{1}\\
&= - \frac{1}{4} \, \delta^{B}_{~A} \, \delta_{ij} \, \kappa^{im} \kappa^{jn} \, h_{2} h_{x->q} T_{n} T_{m} h_{p->x} h_{1}\\
&\,+ \frac{1}{2} \epsilon_{ba}^{~~c} \, \tau_{c~\,A}^{~B} \, \pi_{1}(h^{-1})^{a}_{~i} \, \pi_{1}(h^{-1})^{b}_{~j} \, \kappa^{im} \kappa^{jn} \, h_{2} h_{x->q} T_{[n} T_{m]} h_{p->x} h_{1}\\
&= \frac{1}{8} \, \delta^{B}_{~A} \, \kappa^{mn} \, h_{2} h_{x->q} T_{n} T_{m} h_{p->x} h_{1}\\
&\,+ \frac{1}{4} \epsilon_{ba}^{~~c} \, \tau_{c~\,A}^{~B} \, \pi_{1}(h^{-1})^{a}_{~i} \, \pi_{1}(h^{-1})^{b}_{~j} \, \kappa^{im} \kappa^{jn} \, \epsilon_{nm}^{~~k} \, h_{2} h_{x->q} T_{k} h_{p->x} h_{1}\\
&= \frac{\Delta(j)}{8} \, \delta^{B}_{~A} \, h_{2} h_{p->q} h_{1}\\
&\,+ \frac{1}{16} \tau_{c~\,A}^{~B} \, \delta^{cd} \, \epsilon_{bad} \epsilon^{jik} \, \pi_{1}(h^{-1})^{a}_{~i} \, \pi_{1}(h^{-1})^{b}_{~j} \, h_{2} h_{x->q} T_{k} h_{p->x} h_{1}\\
&= \frac{j(j+1)}{16} \, \delta^{B}_{~A} \, h_{2} h_{p->q} h_{1} + \frac{1}{8} \tau_{c~\,A}^{~B} \, \delta^{cd} \, \pi_{1}(h)^{k}_{~d} \, h_{2} h_{x->q} T_{k} h_{p->x} h_{1}\\
&= \frac{j(j+1)}{16} \, \delta^{B}_{~A} \, h_{2} h_{p->q} h_{1} - \frac{1}{4} \tau_{c~\,A}^{~B} \, \pi_{1}(h^{-1})^{c}_{~l} \, \kappa^{kl} \, h_{2} h_{x->q} T_{k} h_{p->x} h_{1}\\
&= \frac{j(j+1)}{16} \, \ket{\snsUnchanged} - \frac{1}{4} \ket{\snsSingleDeformed} \, ,
\end{split}
\end{equation}
where we used that $\Delta(j) = \frac{j(j+1)}{2}$. Putting everything together, we end up with
\begin{equation}
\begin{split}
\left( \W_{{H_{2}^{-1}}} \right)^{A}_{~C} \, \left( \W_{H_1}^{\vphantom{-1}} \right)^{C}_{~B} \, \ket{\jjstate} &= \left[ \chi_{c}(j)^{2} + 4 \alpha \beta \frac{j(j+1)}{16} \, \chi_{s}(j)^{2} \right] \, \ket{\snsUnchanged}\\
&\,+ \left[ 2i \, \left( \beta - \alpha \right) \, \chi_{c}(j) \, \chi_{s}(j) - \alpha\beta \, \chi_{s}(j)^{2} \right] \, \ket{\snsSingleDeformed} \, .
\end{split}
\label{eqn:WWpsi}
\end{equation}
If we now want the state $\ket{\jjstate}$ to satisfy the quantized isolated horizon boundary condition, the right hand side of equation~\eqref{eqn:WWpsi} has to be equal to $\ket{\snsUnchanged} = \delta^{A}_{~B} \, \ket{\jjstate}$. We can therefore read off the equations
\begin{align}
\chi_{c}(j)^{2} + \alpha \beta \frac{j(j+1)}{4} \, \chi_{s}(j)^{2} &= 1
\label{eqn:QIHBCSolSystemEq1}\\
2i \, \left( \beta - \alpha \right) \, \chi_{c}(j) \, \chi_{s}(j) - \alpha \beta \, \chi_{s}(j)^{2} &= 0
\label{eqn:QIHBCSolSystemEq2}
\end{align}
that need to be fulfilled by $\alpha$, $\beta$ and $j$. Note that, in principle, $\alpha$ and $\beta$ are allowed to depend on $j$! We recognize that \eqref{eqn:QIHBCSolSystemEq1} will reduce to the condition of the surface holonomies having unit determinant if $\beta = \alpha^{-1}$. Since this condition has already shown up quite often during the analysis of the properties of surface holonomies, this seems like a natural condition and we will assume that $\alpha$ and $\beta$ satisfy the relation above. However, we get an additional condition from \eqref{eqn:QIHBCSolSystemEq2}. This can be solved by choosing $\alpha$ as a function of $j$ satisfying
\begin{equation}
\frac{2i}{\alpha(j)} \, \left( 1 - \alpha(j)^{2} \right) \chi_{c}(j) = \chi_{s}(j) \, ,
\end{equation}
which implies that $\alpha(j)$ has to be a solution to the quadratic equation
\begin{equation}
\alpha(j)^{2} + \frac{\chi_{s}(j)}{2i \, \chi_{c}(j)} \, \alpha(j) - 1 = 0 \, .
\end{equation}
We therefore get
\begin{equation}
\alpha(j) = \frac{\chi_{s}(j)}{4 \, \chi_{c}(j)}\, i \pm \sqrt{1 - \left[ \frac{\chi_{s}(j)}{4\, \chi_{c}(j)} \right]^{2}}
\end{equation}
and we immediately see that $\alpha(j)$ is purely imaginary if $\abs{\frac{\chi_{s}(j)}{4\, \chi_{c}(j)}} > 1$. On the other hand, if $\abs{\frac{\chi_{s}(j)}{4\, \chi_{c}(j)}} \leq 1$, we have $\abs{\alpha(j)} = 1$ and $\alpha(j)$ will therefore just be a phase. Actually, the latter is the case for most values of $j$. This can for example be seen from figure~\ref{fig:discriminant_plots}, where we have plotted the full discriminant
\begin{equation}
D := 1 - \left[ \frac{\chi_{s}(j)}{4\, \chi_{c}(j)} \right]^{2} \, .
\end{equation}
Remember that $\chi_{s}$ depends on the constant
\begin{equation}
c = -8\pi G \hbar \beta i \, \frac{4\pi (1-\beta^{2})}{a_{\mathcal{H}}} =: -\frac{8\pi i}{k} \, ,
\end{equation}
where
\begin{equation}
k = \frac{a_{\mathcal{H}}}{4\pi l_{P}^{2} \beta (1-\beta^{2})} \, ,
\end{equation}
$l_{P} = \sqrt{\hbar G}$ denotes the Planck length and $a_{\mathcal{H}}$ is the classical area of the horizon in the IHBC. We can now either keep $a_{\mathcal{H}}$ as a free classical parameter, or we can replace it with the eigenvalue of the area operator in the state under consideration. In figure~\ref{fig:discriminant_plots}, we show plots for both options\footnote{Note that in both cases we still have the Barbero-Immirzi parameter $\beta$ as a free parameter and the numerical values of the solutions will depend on its value. For all plots in this chapter we have used $\beta = 0.274$. This is the value determined from the entropy calculation for type I isolated horizons with gauge group \SU ~\citep{Ghosh:2006ph,Engle:2011vf,Agullo:2010zz}.}

\begin{figure}[!htb]
	\includegraphics[width=0.3\textwidth]{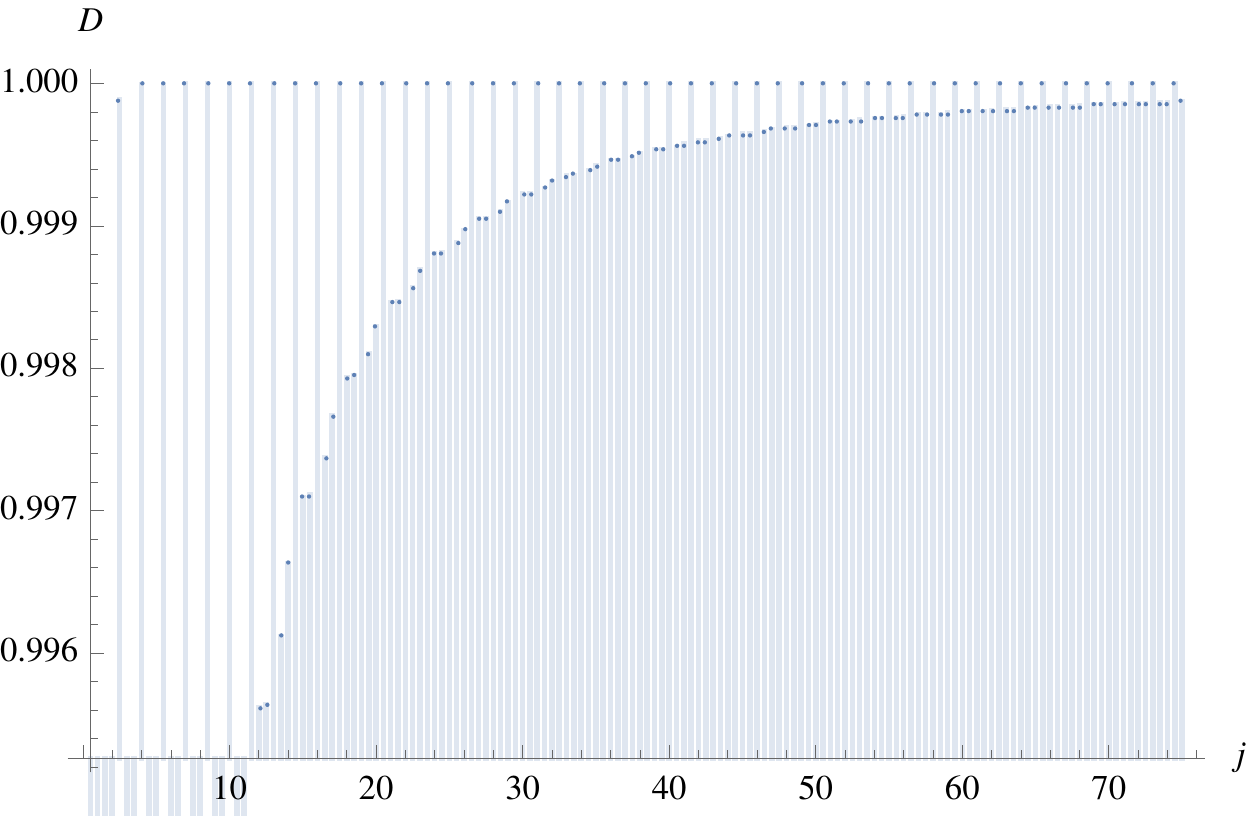}
	\includegraphics[width=0.3\textwidth]{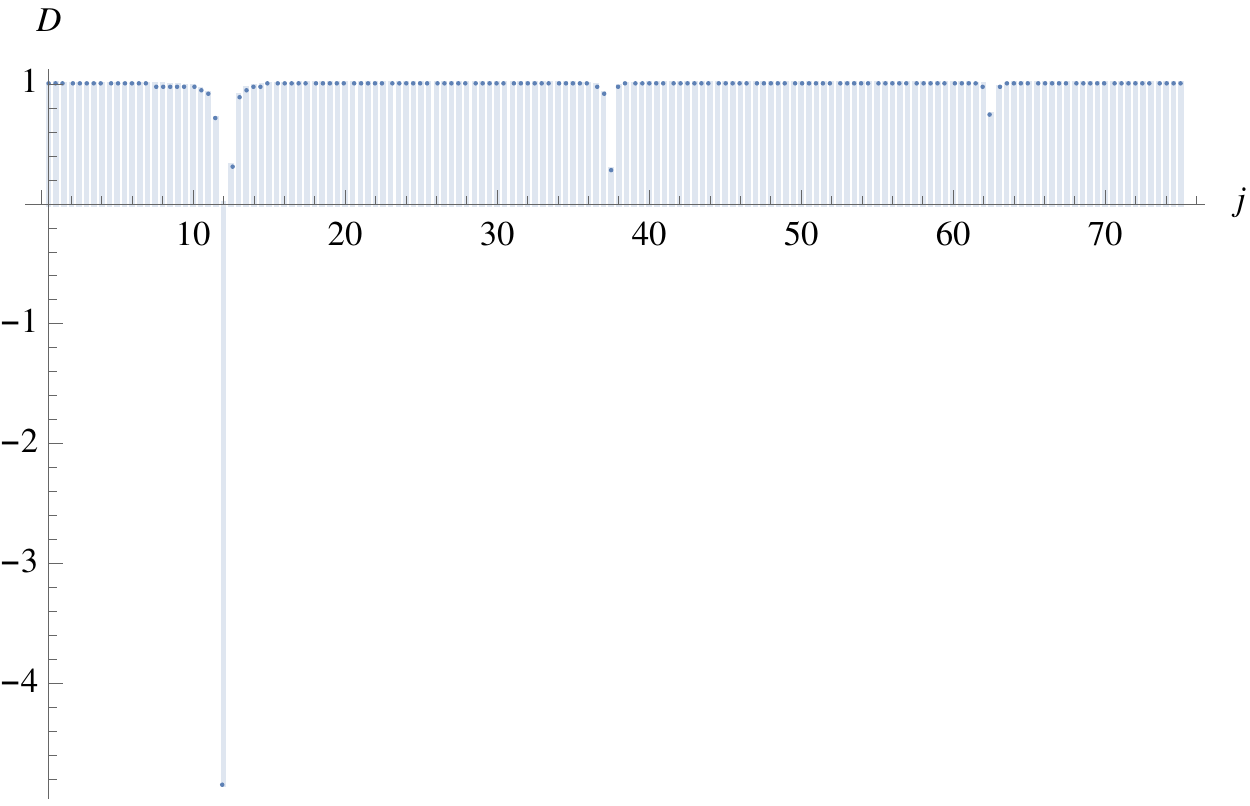}
	\includegraphics[width=0.3\textwidth]{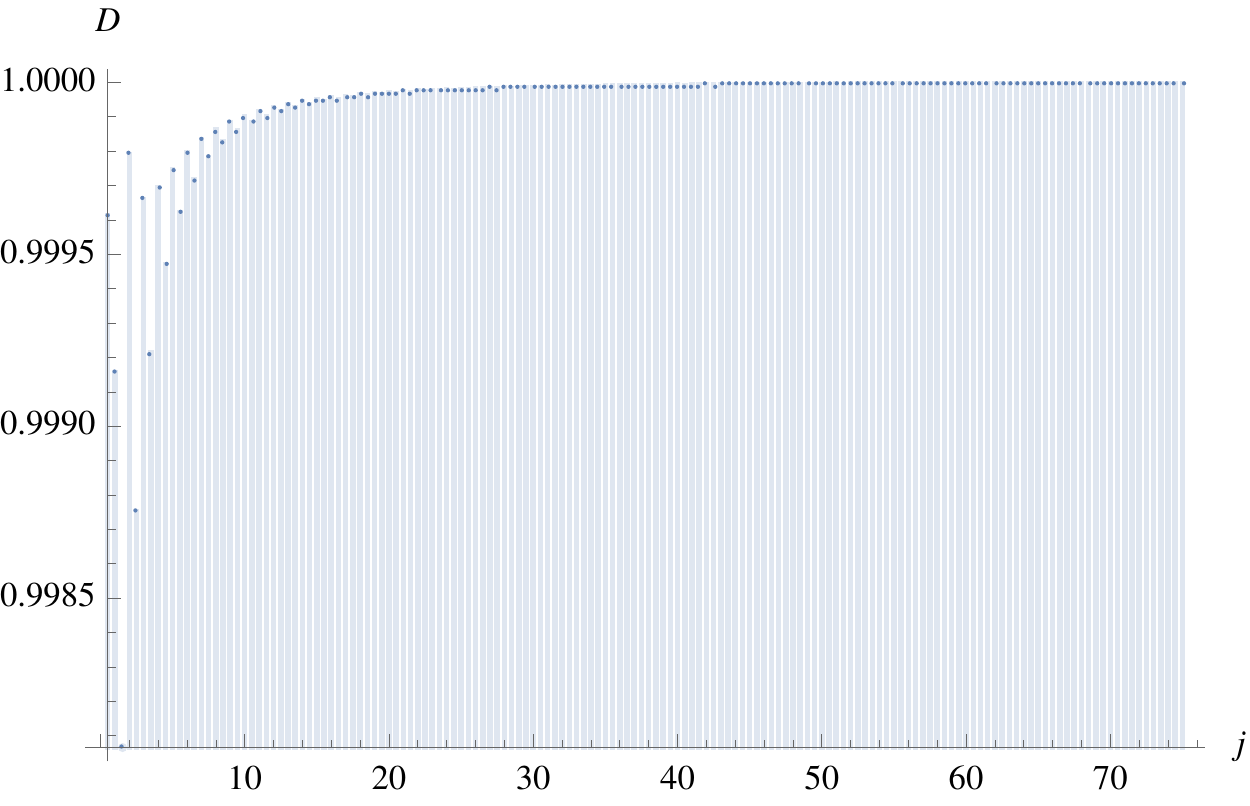}
	\caption{This figure shows the discriminant $D$ as a function of $j$. From left to right, the plots are for $k=3$, $k=101$ and $k=k(j)$, respectively.}
	\label{fig:discriminant_plots}
\end{figure}

The plots to the left and in the middle are for fixed values of $k$ ($k=3$ and $k=101$, respectively), while for the plot to the right we used
\begin{equation}
k(j) = \frac{4}{1 - \beta^2} \, \sqrt{j(j+1)}
\label{eqn:k(j)}
\end{equation}
following directly from inserting the area eigenvalue
\begin{equation}
a_{\mathcal{H}} = 16\pi\beta l_P^2 \, \sqrt{j(j+1)}
\end{equation}
into the definition of $k$. All three plots show that the discriminant $D$ tends to $1$ as $j$ increases. However, the details differ between the two choices for $k$. The plots for fixed $k$ show some periodic behavior (with period approximately $\sfrac{k}{4}$ in the plot for $k=101$). Also, the convergence of $D$ to unity seems slower in this case. The fast convergence rate in the case where $k=k(j)$ can also be seen more clearly from figure~\ref{fig:plot_1-discriminant_k=k(j)}, where we have plotted the difference $1-D$. This deviation from $1$ is less than $10^{-7}$ for all spins greater than $20$, and it seems to decrease by another order of magnitude before reaching spin $30$! This implies that if $j$ becomes large enough, the solutions for $\alpha$ will approximately become $\pm 1$. However, this way of solving \eqref{eqn:QIHBCSolSystemEq2} has a serious drawback. Recall that in order to solve the QIHBC, a state needs to satisfy not only \eqref{eqn:cond2b} but also \eqref{eqn:cond2a}. The latter condition leads to almost the same set of equations, but with the opposite sign in the second term of equation~\eqref{eqn:QIHBCSolSystemEq2}. We thus get a different solution for $\alpha$, implying that conditions~\eqref{eqn:cond2a} and~\eqref{eqn:cond2b} cannot be solved simultaneously using this approach.\\
\begin{figure}[!htb]
	\includegraphics[scale=0.4]{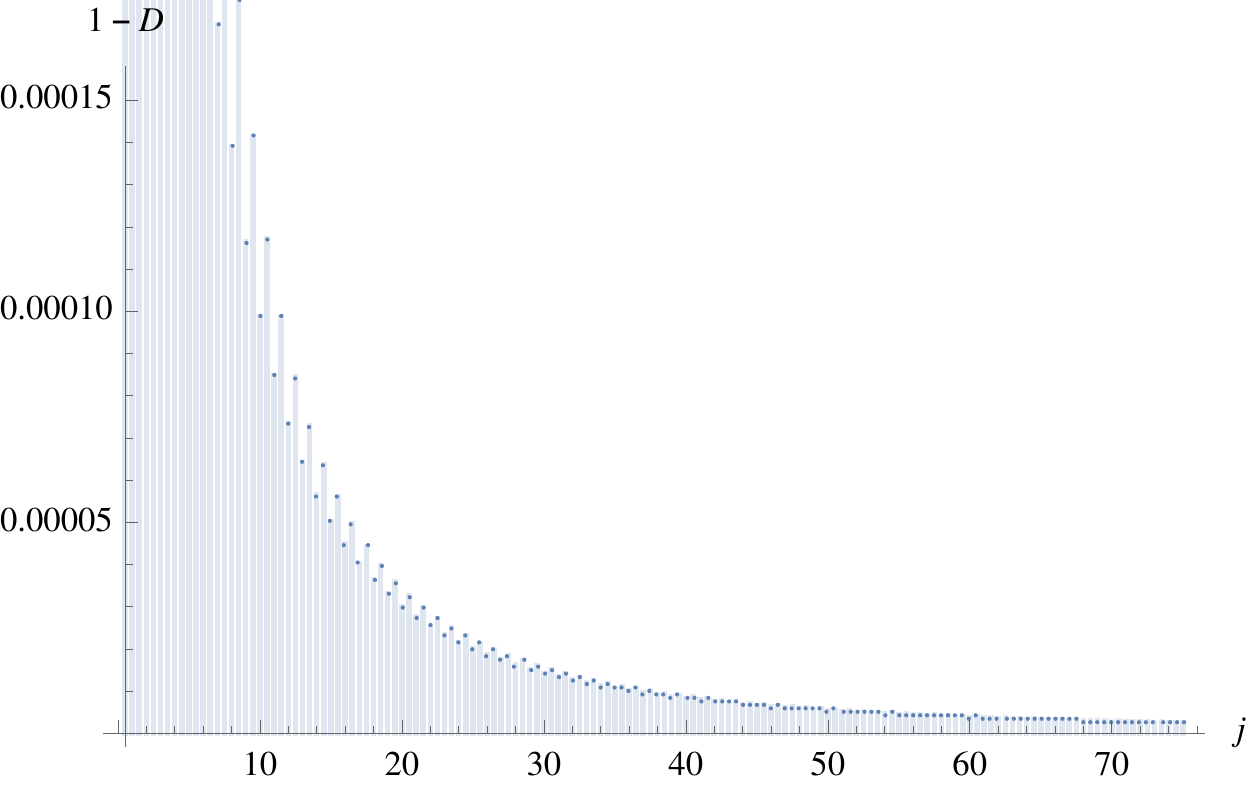}
	\caption{We plot the deviation from 1 of the discriminant $D$ for $k=k(j)$.}
	\label{fig:plot_1-discriminant_k=k(j)}
\end{figure}
Fortunately, we can also solve \eqref{eqn:QIHBCSolSystemEq2} by requiring that $\chi_{s}(j) = 0$. Although this works for any choice of $\alpha$ and $\beta$,we will still demand that $\beta = \alpha^{-1}$ in order to identify equation~\eqref{eqn:QIHBCSolSystemEq1} with the unit determinant condition. We can already see from figure~\ref{fig:plot_1-discriminant_k=k(j)}, where we have plotted
\begin{equation}
1 - D = \left[ \frac{\chi_{s}(j)}{4\, \chi_{c}(j)} \right]^{2} \, ,
\end{equation}
that $\chi_{s}(j)$ will approach $0$ as $j$ grows large. This is confirmed in figure~\ref{fig:chi_s}, where we have plotted $\chi_{s}(j)$ for the same three choices of $k$ as before. The overall tendency of converging to $0$ is again the same in all three cases. However, while the overall convergence is again faster in the case where $k$ depends on $j$, there are individual spins in the plots for fixed $k$, for which $\chi_{s}(j)$ is considerably closer to zero than for any spin less than $70$ in the $j$-dependent case!

\begin{figure}[!htb]
	\includegraphics[width=0.3\textwidth]{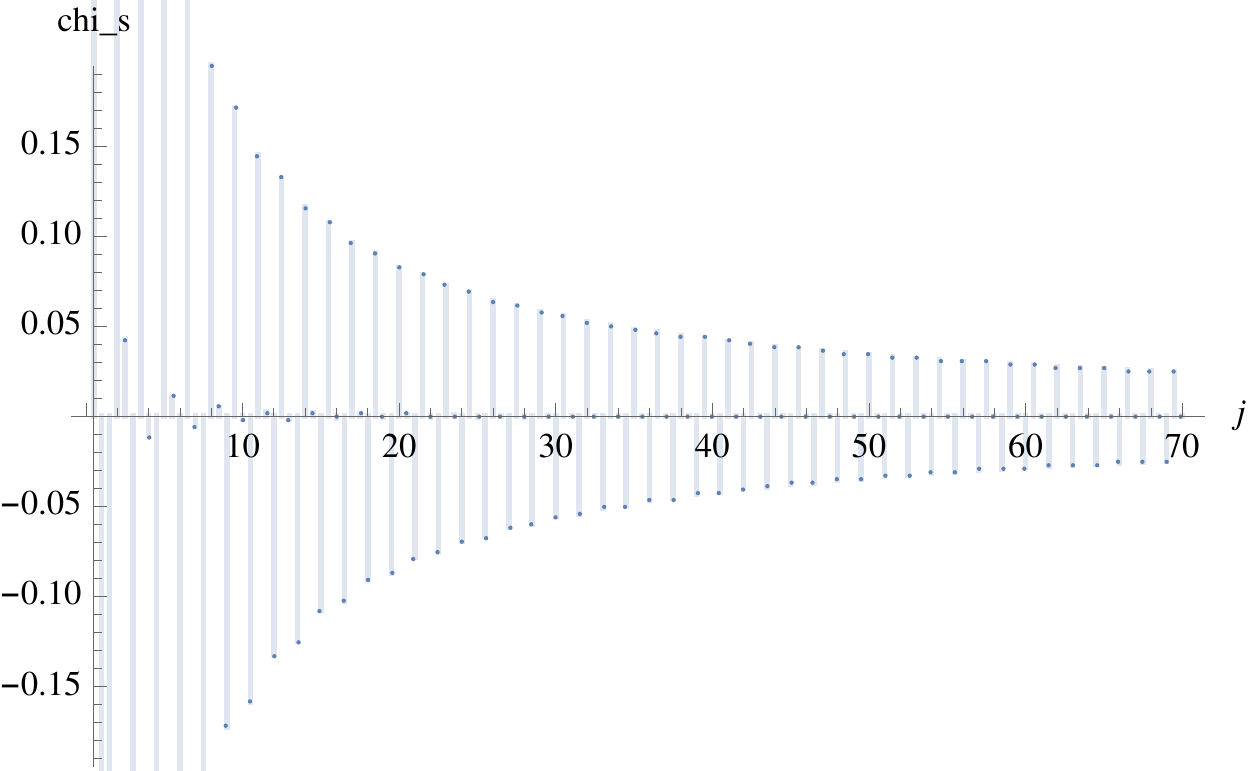}
	\includegraphics[width=0.3\textwidth]{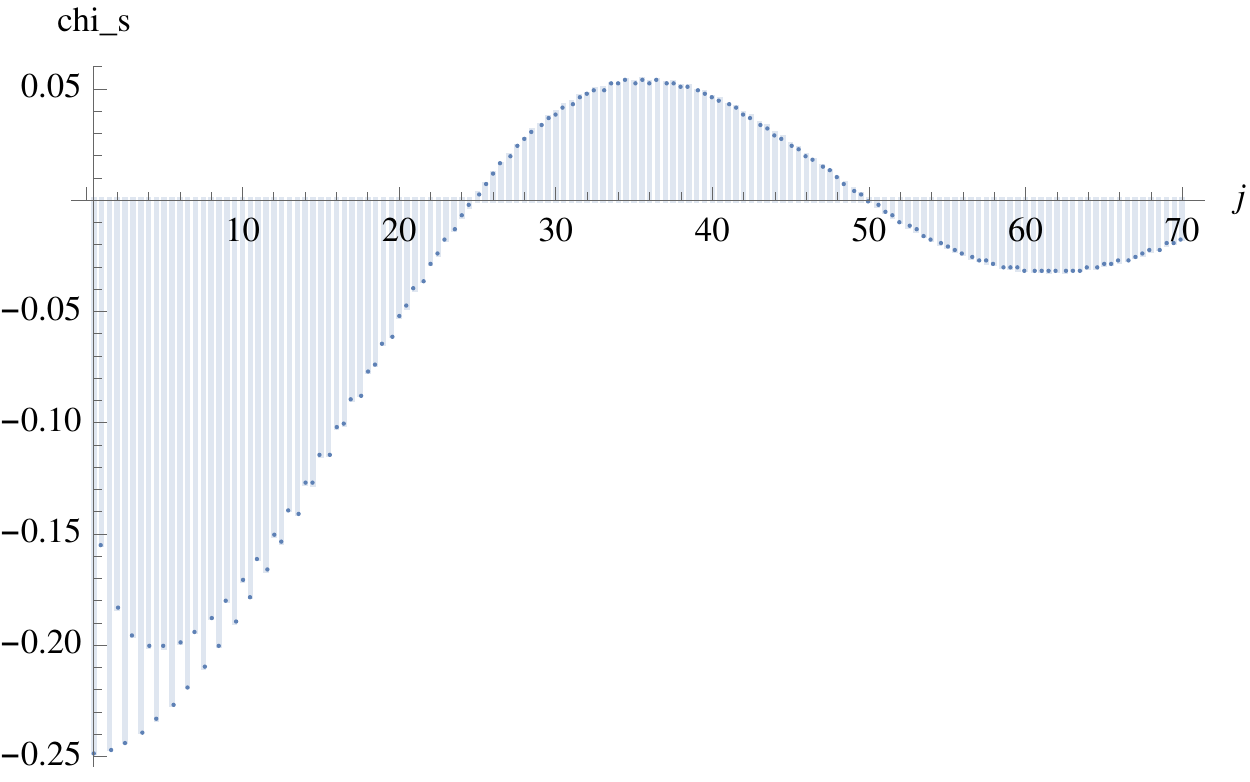}
	\includegraphics[width=0.3\textwidth]{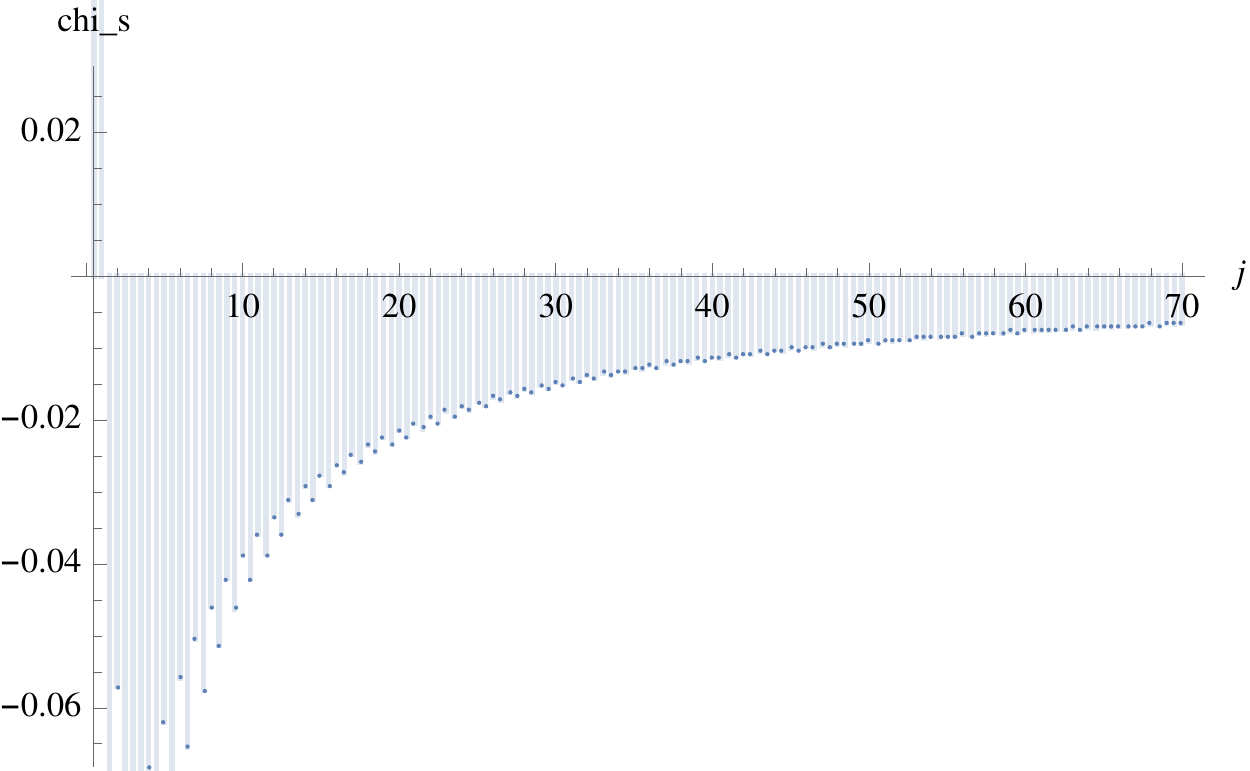}
	\caption{This figure shows plots of $\chi_{s}(j)$. As in figure~\ref{fig:discriminant_plots}, we have chosen $k=3$, $k=101$ and $k=k(j)$, respectively, from left to right.}
	\label{fig:chi_s}
\end{figure}

\section{Conclusion \& Outlook}

In the preceeding sections, we have presented three sets of results:
\begin{enumerate}
\item We have defined the surface-ordered, exponentiated fluxes $\widehat{\mathcal{W}_S}$ on a large class of states in the Hilbert space of LQG. 
\item We have explored many of their properties, such as commutation relations and the spectra of their trace and determinant. Interestingly, the $\widehat{\mathcal{W}_S}$ are in some sense close to classical group elements, but by no means in all aspects. 
\item We have started to analyze what kind of states fulfil the quantum version of the isolated horizon boundary condition. We find that a relevant operator seems to be the determinant of the $\widehat{\mathcal{W}_S}$ on the horizon. But the states we look at are too limited to make any solid statements about quantized IHs. 
\end{enumerate}
One fundamental limitation of our method is that while $\widehat{\mathcal{W}_S}$ determines the holonomy around $S$, it will create new, undetermined holonomies when acting on quantum states. We suspect that this is responsible for the problem that, although a classical surface holonomy is invariant under changes of the homotopy generating the surface, the $\widehat{\mathcal{W}_S}$ appear to depend on the parametrization in the sense that they give the punctures an ordering. This ordering is dependent on the parametization and changing it appears to change the state that results from the action of $\widehat{\mathcal{W}_S}$. This might be partially remedied if the properties of the holonomies created by $\widehat{\mathcal{W}_S}$ could be established through the use of the IHBC. This direction should be studied further. 

Another avenue for future work could be to discard the results \eqref{eq:xis}, \eqref{eq:chis} for the second coefficient in the action of $\widehat{\mathcal{W}_S}$ obtained by the Duflo map, and instead to fix it by demanding that the determinant is equal to 1 for all states, 
\begin{equation}
\xi_s^2(j)= \frac{8}{\Delta_j}\left(1 -\cosh^2\left[\frac{(2j+1)c}{8}\right]\right)\qquad \chi_s^2(j)= \frac{2}{\Delta_j}\left(1 -\cosh^2\left[\frac{(2j+1)c}{4}\right]\right).
\end{equation}
In this setting, one could continue to work with the LQG holonomies on the horizon and perhaps obtain a state described by a measure on the space $\overline{\mathcal{A}}$ of generalized connections. 

A final point that should be studied further is the quantization of the $\widehat{\mathcal{W}_S}$ without setting
\begin{equation}
||E||^{2} = 2 \, ||\Eu||^{2} + 2 \, ||\Ed||^{2}
\end{equation}
(see the discussion in section \ref{ssec:two-edge-punct-calc} for details). This might substantially change the properties of the operators $\widehat{\mathcal{W}_S}$ . 

\begin{acknowledgments}
The authors thank Lee Smolin for helpful comments, and the members of the Institute for Quantum Gravity at the Friedrich-Alexander-University Erlangen-N\"urnberg (FAU) for discussions at various stages of the work. 
\end{acknowledgments}


\appendix

\section{Action of quantum surface holonomy on one-edge puncture state (detailed calculation)}
\label{app:one-edge-calc}
We can rewrite the sum in the last line of the previous equation as

\begin{equation}
\begin{split}
\sum_{m=0}^{2(p+1)} \binom{2p+3}{m} B_{m} &\left( 2^{m} - 2 \right) \left[ 2j + 1 \right]^{2(p+1) - m} \\
= \, 2^{2(p+1)} &\sum_{m=0}^{2(p+1)} \binom{2p+3}{m} B_{m} \, \left[ \frac{2j + 1}{2} \right]^{2(p+1)-m} \\
- \, 2 \, &\sum_{m=0}^{2(p+1)} \binom{2p+3}{m} B_{m} \, \left[ 2j + 1 \right]^{2(p+1)-m}
\end{split}
\label{eqn:Bernoulli_sum_halfint_spin}
\end{equation}

and make use of the relation (known as Faulhaber's formula)

\begin{equation}
\sum_{k=0}^{n} \binom{n+1}{k} B_{k} m^{n-k} = \frac{n+1}{m} \left[ 1^{n} + 2^{n} + \dots + m^{n} \right]
\label{eqn:Bernoulli_relation_1}
\end{equation}

for the Bernoulli numbers (of second kind) $B_{k}$ and positive integers $m,n$ to obtain

\begin{equation}
\begin{split}
\sum_{m=0}^{2(p+1)} &\binom{2p+3}{m} B_{m} \left( 2^{m} - 2 \right) \left[ 2j + 1 \right]^{2(p+1) - m} \\
&= \, 2^{2p+2} \, \frac{2(2p + 3)}{2j+1}  \, \left[ 1^{2(p+1)} + 2^{2(p+1)} + \dots + \left( \frac{2j+1}{2} \right)^{2(p+1)} \right] \\
&~\quad - \, \frac{2(2p + 3)}{2j+1} \, \left[ 1^{2(p+1)} + 2^{2(p+1)} + \dots + \left( 2j + 1 \right)^{2(p+1)} \right] \\
&= \, - \, \frac{2(2p+3)}{2j+1} \, \left[  1^{2(p+1)} + 3^{2(p+1)} + \dots + \left( 2j \right)^{2(p+1)} \right] \, .
\end{split}
\end{equation}

Note that when applying equation \eqref{eqn:Bernoulli_relation_1} to the middle line of equation \eqref{eqn:Bernoulli_sum_halfint_spin} we assumed that $\frac{1}{2} \left( 2j+1 \right)$ is an integer, i.e. that the spin j is a half-integer. Inserting this result back into equation \eqref{eqn:Duflo_on_odd_powers_j-rep} we are left with

\begin{equation}
\begin{split}
Q_{DK} \left[ ||E||^{2k} E_{i} \right] \Biggl\vert_{\Hej}
&= \frac{2}{8^{k}} \, \frac{1}{2j(2j+1)(2j+2)} \, \pi^{(j)} \left[ \widehat{E}_{i} \right] \, \times \\ 
&~\quad \sum_{p=0}^{k} \binom{2k+4}{2p+3} \, \frac{(2p+2)(2p+3)}{(2k+2)(2k+4)} \, \left[ 1^{2(p+1)} + 3^{2(p+1)} + \dots + \left( 2j \right)^{2(p+1)} \right] \\
&= \frac{2}{8^{k}} \, \frac{2k+3}{2k+2} \, \frac{1}{2j(2j+1)(2j+2)} \, \pi^{(j)} \left[ \widehat{E}_{i} \right] \, \times \\ 
&~\quad \sum_{p=0}^{k} \binom{2k+2}{2p+1} \, \left[ 1^{2(p+1)} + 3^{2(p+1)} + \dots + \left( 2j \right)^{2(p+1)} \right] \, .
\end{split}
\label{eqn:Duflo_on_odd_powers_j-rep_2nd_formula}
\end{equation}

Let us focus on the last line to further simplify this expression. We can make use of the relation

\begin{equation}
\sum_{p=0}^{k} \binom{2k+2}{2p+1} n^{2p+1} = \frac{1}{2} \left[ \, \left( n+1 \right)^{2k+2} -  \left( n-1 \right)^{2k+2} \, \right]
\end{equation}

to obtain

\begin{equation}
\begin{split}
\sum_{p=0}^{k} \binom{2k+2}{2p+1} \, &\left[ 1^{2(p+1)} + 3^{2(p+1)} + \dots + \left( 2j \right)^{2(p+1)} \right] = \sum_{p=0}^{k} \binom{2k+2}{2p+1} \, \sum_{l=0}^{\frac{2j-1}{2}} \left( 2 l + 1 \right)^{2p+2}\\
&= \frac{1}{2} \sum_{l=0}^{\frac{2j-1}{2}} \left[ 2l+2 - 1 \right] \left[ \left( 2l+2 \right)^{2k+2} - \left( 2l \right)^{2k+2} \right] \\
&= \frac{1}{2} \sum_{l=0}^{\frac{2j-1}{2}} \left[ \left( 2l+2 \right)^{2k+3} - \left( 2l+2 \right)^{2k+2} - \left( 2l \right)^{2k+3} - \left( 2l \right)^{2k+2} \right] \\
&= \frac{1}{2} \sum_{l=0}^{\frac{2j-1}{2}} \left[ \left( 2l+2 \right)^{2k+3} - \left( 2l \right)^{2k+3} \right] - \frac{1}{2} \sum_{l=0}^{\frac{2j-1}{2}} \left[ \left( 2l+2 \right)^{2k+2} + \left( 2l \right)^{2k+2} \right] \\
&= \frac{1}{2} \left( 2j+1 \right)^{2k+3} - \frac{1}{2} \sum_{l=0}^{\frac{2j-1}{2}} \left[ \left( 2l+2 \right)^{2k+2} + \left( 2l \right)^{2k+2} \right] \\
&= \frac{1}{2} \left( 2j+1 \right)^{2k+3} - \frac{1}{2} \left( 2j+1 \right)^{2k+2} - \sum_{l=1}^{\frac{2j-1}{2}} \left( 2l \right)^{2k+2} \\
&= j \left( 2j+1 \right)^{2k+2} - 2^{2k+2} \sum_{l=1}^{\frac{2j-1}{2}} l^{2k+2} \\
&= 4^{k+1} \left[ j \left( \frac{2j+1}{2} \right)^{2k+2} - \sum_{l=1}^{\frac{2j-1}{2}} l^{2k+2} \right] \, .
\end{split}
\end{equation}

Reinserting this into \eqref{eqn:Duflo_on_odd_powers_j-rep_2nd_formula} we end up with

\begin{equation}
\begin{split}
Q_{DK} &\left[ ||E||^{2k} E_{i} \right] \Biggl\vert_{\Hej} \\
&= \frac{8}{2^{k}} \, \frac{2k+3}{2k+2} \, \frac{1}{2j(2j+1)(2j+2)} \, \left[ j \left( \frac{2j+1}{2} \right)^{2k+2} - \sum_{l=1}^{\frac{2j-1}{2}} l^{2k+2} \right] \, \pi^{(j)} \left[ \widehat{E}_{i} \right] \, .
\end{split}
\end{equation}

Using this result we can now use equation \eqref{eqn:surf_hol_expansion} to calculate

\begin{equation}
\begin{split}
Q_{DK} \left[ W_p \right] \Biggl\vert_{\Hej} &= \sum_{n=0}^{\infty} \frac{1}{(2n)!} \left( \frac{c}{2\sqrt{2}} \right)^{2n} Q_{DK} \left[ ||E||^{2n} \right] \Biggl\vert_{\Hej} \otimes \,\mathbb{1}_{2} \\
&\quad + \sum_{n=0}^{\infty} \frac{1}{(2n+1)!} \left( \frac{c}{2\sqrt{2}} \right)^{2n} c \, \kappa^{il} Q_{DK} \left[ ||E||^{2n} E_{i} \right] \Biggl\vert_{\Hej} \otimes \,\tau_{l} \\
&= \sum_{n=0}^{\infty} \frac{1}{(2n)!} \left( \frac{c}{2\sqrt{2}} \right)^{2n} \frac{1}{8^{n}} \left( 2j+1 \right)^{2n} \id_{\Hej} \otimes \mathbb{1}_{2} \\ 
&\quad + \sum_{n=0}^{\infty} \frac{1}{(2n+1)!} \left( \frac{c}{2\sqrt{2}} \right)^{2n} c \, \frac{8}{2^{n}} \, \frac{2n+3}{2n+2} \, \frac{1}{2j(2j+1)(2j+2)} \, \times \\
&\qquad \left[ j \left( \frac{2j+1}{2} \right)^{2n+2} - \sum_{l=1}^{\frac{2j-1}{2}} l^{2n+2} \right] \, \kappa^{il} \pi^{(j)} \left[ \widehat{E}_{i} \right] \otimes \tau_{l} \\
&= \operatorname{cosh} \left( \frac{(2j+1)c}{8} \right) \, \id_{\Hej} \otimes \mathbb{1}_{2} + \frac{128}{c} \, \frac{\kappa^{il} \pi^{(j)} \left[ \widehat{E}_{i} \right] \otimes \tau_{l}}{2j(2j+1)(2j+2)} \times \\
& \qquad \sum_{n=0}^{\infty} \frac{2n+3}{(2n+2)!} \left( \frac{c}{4} \right)^{2n+2} \, \left[ j \left( \frac{2j+1}{2} \right)^{2n+2} - \sum_{l=1}^{\frac{2j-1}{2}} l^{2n+2} \right] \, .
\end{split}
\end{equation}

Let us look at the term
\begin{equation}
\frac{128}{c} \, \frac{\kappa^{il} \pi^{(j)} \left[ \widehat{E}_{i} \right] \otimes \tau_{l}}{2j(2j+1)(2j+2)} \, \sum_{n=0}^{\infty} \frac{2n+3}{(2n+2)!} \left( \frac{c}{4} \right)^{2n+2} \, \left[ j \left( \frac{2j+1}{2} \right)^{2n+2} - \sum_{l=1}^{\frac{2j-1}{2}} l^{2n+2} \right]
\label{eqn:intermediate1}
\end{equation}
in more detail. The sum over the first term inside the square brackets can be calculated as
\begin{equation}
\begin{split}
\sum_{n=0}^{\infty} \frac{2n+3}{(2n+2)!} \left( \frac{c}{4} \right)^{2n+2} \, j \left( \frac{2j+1}{2} \right)^{2n+2} &= \frac{8j}{2j+1} \, \frac{\operatorname{d}}{\operatorname{d}c} \sum_{n=0}^{\infty} \frac{1}{(2n+2)!} \left[ \frac{(2j+1)c}{8} \right]^{2n+3} \\
&= \frac{8j}{2j+1} \, \frac{\operatorname{d}}{\operatorname{d}c} \left[  \frac{(2j+1)c}{8} \left[ \operatorname{cosh} \left( \frac{(2j+1)c}{8} \right) - 1 \right] \right] \\ 
&= \frac{\operatorname{d}}{\operatorname{d}c} \left[  jc \left[ \operatorname{cosh} \left( \frac{(2j+1)c}{8} \right) - 1 \right] \right] \, .
\end{split}
\end{equation}
The second term can also be simplified via
\begin{equation}
\begin{split}
\sum_{n=0}^{\infty} \frac{2n+3}{(2n+2)!} \left( \frac{c}{4} \right)^{2n+2} \, \sum_{l=1}^{\frac{2j-1}{2}} l^{2n+2} &= \sum_{l=1}^{\frac{2j-1}{2}} \sum_{n=0}^{\infty} \frac{2n+3}{(2n+2)!} \left( \frac{cl}{4} \right)^{2n+2} \\
&=\sum_{l=1}^{\frac{2j-1}{2}} \frac{4}{l} \frac{\operatorname{d}}{\operatorname{d}c} \sum_{n=0}^{\infty} \frac{1}{(2n+2)!} \left( \frac{cl}{4} \right)^{2n+3} \\
&= \sum_{l=1}^{\frac{2j-1}{2}} \frac{4}{l} \frac{\operatorname{d}}{\operatorname{d}c} \frac{cl}{4} \sum_{n=0}^{\infty} \frac{1}{(2n+2)!} \left( \frac{cl}{4} \right)^{2n+2} \\
&= \sum_{l=1}^{\frac{2j-1}{2}} \frac{4}{l} \frac{\operatorname{d}}{\operatorname{d}c} \frac{cl}{4} \left[ \operatorname{cosh} \left( \frac{cl}{4} \right) - 1 \right] \\
&= \frac{\operatorname{d}}{\operatorname{d}c} \, c \sum_{l=1}^{\frac{2j-1}{2}} \left[ \operatorname{cosh} \left( \frac{cl}{4} \right) - 1 \right] \\
&= \frac{\operatorname{d}}{\operatorname{d}c} \, c \left[ \frac{\sinh{\frac{(2j-1)c}{16}}}{\sinh{\frac{c}{8}}} \cosh{\frac{(2j+1)c}{16}} - \left( j - \frac{1}{2} \right) \right] 
\end{split}
\end{equation}
where we used
\begin{equation}
\sum_{m=1}^{n} \operatorname{cosh}(mx) = \frac{\operatorname{sinh}(\frac{nx}{2})}{\operatorname{sinh}(\frac{x}{2})} \operatorname{cosh} \left( \frac{(n+1)x}{2} \right)
\end{equation}
in the last equality. We can thus rewrite the sum in expression \eqref{eqn:intermediate1} as
\begin{equation}
\frac{\operatorname{d}}{\operatorname{d}c} \left[  jc \cosh{\left( \frac{(2j+1)c}{8} \right)} - \frac{c}{2} - c \, \frac{\sinh{\frac{(2j-1)c}{16}}}{\sinh{\frac{c}{8}}} \cosh{\left( \frac{(2j+1)c}{16} \right)} \right] \, .
\end{equation}
Defining
\begin{equation}
Q_{D} \left[ W_p \right] \Biggl\vert_{\Hej} =: \xi_c(j) \, \id_{\Hej} \otimes \mathbb{1}_{2} + i\,\xi_s(j) \, \kappa^{im} \pi^{(j)} \left[ \widehat{E_{i}} \right] \otimes \tau_{m}
\end{equation}
we arrive at
\begin{equation}
\xi_c(j) = \operatorname{cosh} \left( \frac{(2j+1)c}{8} \right)
\end{equation}
and
\begin{equation}
\begin{split}
\xi_s(j) &= \frac{-128i}{2j(2j+1)(2j+2)}\\ 
& \quad \times \frac{1}{c} \, \frac{\operatorname{d}}{\operatorname{d}c} \left[  jc \cosh{\left( \frac{(2j+1)c}{8} \right)} - \frac{c}{2} - c \, \frac{\sinh{\frac{(2j-1)c}{16}}}{\sinh{\frac{c}{8}}} \cosh{\left( \frac{(2j+1)c}{16} \right)} \right]
\end{split}
\end{equation}
for the function $\xi_c(j)$ and $c\xi_s(j)$. In the expression for $\xi_s(j)$ the derivative with respect to $c$ can still be carried out, leading to

\begin{equation}
\begin{split}
\xi_s(j) &= \frac{-8i}{2j(2j+1)(2j+2)} \times \\
&\qquad \left[ 2j(2j+1) \frac{\cosh{\left( \frac{(2j+1)c}{8} \right)}}{\frac{(2j+1)c}{8}} + 2j(2j+1) \sinh{\left( \frac{(2j+1)c}{8} \right)} \right. \\
& \qquad \left. - \frac{1}{\sinh{\left( \frac{c}{8} \right)}} \left( 2j \cosh{\left( \frac{2jc}{8} \right)} + 2j \frac{\sinh{\left( \frac{2jc}{8} \right)}}{\frac{2jc}{8}} - \sinh{\left( \frac{2jc}{8} \right)} \coth{\left( \frac{c}{8} \right)} \right) \right] \, .
\end{split}
\end{equation}

A similar calculation shows that the same results holds for integer values of $j$. The main difference between the two calculations is that one cannot simply apply Faulhaber's formula \eqref{eqn:Bernoulli_relation_1} in the case of integer spins. Instead, one has to use the recently discovered extended version of Faulhaber's formula \citep{Schumacher:2016} in order to simplify the complicated formula we started with.

\section{Action of quantum surface holonomy on two-edge puncture state (detailed calculation)}
\label{app:two_edge_calc}

Here, we will perform the calculation from appendix A again for the two-edge puncture. We will start by inserting  expressions \eqref{eqn:DK-map_gauge-inv_terms} and \eqref{eqn:DK-map_non-gauge-inv_terms} into each line of eq. \eqref{eqn:quantum_surf_hol_expanded_before_Duflo_evaluation} separately in order to keep the calculations legible. Starting with the first line, we have
\begin{equation}
\begin{split}
\sum_{k=0}^{\infty} &\frac{1}{(2k)!} \left( \frac{c}{2} \right)^{2k} \, \sum_{m=0}^{k} \binom{k}{m} \, Q_{DK}[||\Eu||^{2m}] \Biggl\vert_{\Heju} \, Q_{DK}[||\Ed||^{2(k-m)}] \Biggl\vert_{\Hejd} \; \otimes \, \mathbb{1}_{2} \\
&= \sum_{k=0}^{\infty} \frac{1}{(2k)!} \left( \frac{c}{2} \right)^{2k} \, \sum_{m=0}^{k} \binom{k}{m} \, \left[ \frac{(2j^{(u)}+1)^{2}}{8} \right]^{m} \, \left[ \frac{(2j^{(d)}+1)^{2}}{8} \right]^{k-m} \; \id_{\Hej} \otimes \, \mathbb{1}_{2} \\
&= \sum_{k=0}^{\infty} \frac{1}{(2k)!} \left( \frac{c}{2} \right)^{2k} \, \left[ \frac{(2j^{(u)}+1)^{2}}{8} + \frac{(2j^{(d)}+1)^{2}}{8} \right]^{k} \; \id_{\Hejej} \otimes \,\mathbb{1}_{2} \\
&= \cosh \left( \frac{c}{2} \sqrt{\frac{(2j^{(u)}+1)^{2}}{8} + \frac{(2j^{(d)}+1)^{2}}{8}} \right) \; \id_{\Hejej} \otimes \, \mathbb{1}_{2} \, .
\end{split}
\end{equation}
Using the fact that we are considering only gauge-invariant states, we know that we need to have $j^{(u)} = j^{(d)} = j$. The above expression therefore simplifies to
\begin{equation}
\cosh \left( \frac{(2j+1)c}{4} \right) \; \id_{\Heej} \otimes \, \mathbb{1}_{2} \, .
\end{equation}

The second line is considerably more involved and we will split it into two parts corresponding to the two summands in \eqref{eqn:DK-map_non-gauge-inv_terms}. The first part of the second line therefore reads
\begin{equation}
\begin{split}
c \, &\sum_{k=0}^{\infty} \frac{1}{(2k+1)!} \left( \frac{c}{2} \right)^{2k} \kappa^{ij} \, \sum_{m=0}^{k} \binom{k}{m} \, Q_{DK}[||\Eu||^{2m} \Eu_i] \Biggl\vert_{\text{summand } \# 1} \, Q_{DK}[||\Ed||^{2(k-m)}] \; \otimes \tau_{j} \\
&= c \, \sum_{k=0}^{\infty} \frac{1}{(2k+1)!} \left( \frac{c}{2} \right)^{2k} \, \sum_{m=0}^{k} \binom{k}{m} \left[ \frac{(2j^{(d)}+1)^{2}}{8} \right]^{k-m} \\
&\quad \times \frac{2}{8^{m}} \frac{1}{2j^{(u)}(2j^{(u)}+1)(2j^{(u)}+2)} \frac{2m+3}{2m+2} \, j^{(u)} \left( 2j^{(u)}+1 \right)^{2m+2} \, \kappa^{ij} \, \pi^{(j^{(u)})}(\hatEu_{i}) \otimes \,  \tau_{j} \\
&= \frac{c}{2j^{(u)}(2j^{(u)}+1)(2j^{(u)}+2)} \, \sum_{k=0}^{\infty} \frac{1}{(2k+2)!} \left( \frac{c}{2} \right)^{2k} \\
&\quad \times j^{(u)} \frac{\d}{\d j^{(u)}} \sum_{m=0}^{k} \binom{k+1}{m+1} \left[ \frac{(2j^{(d)}+1)^{2}}{8} \right]^{k+1-(m+1)} \frac{1}{8^{m}} \left( 2j^{(u)}+1 \right)^{2m+3} \, \kappa^{ij} \, \pi^{(j^{(u)})}(\hatEu_{i}) \otimes \,  \tau_{j} \\
&= \frac{c}{2j^{(u)}(2j^{(u)}+1)(2j^{(u)}+2)} \, \sum_{k=0}^{\infty} \frac{1}{(2k+2)!} \left( \frac{c}{2} \right)^{2k} \, \kappa^{ij} \, \pi^{(j^{(u)})}(\hatEu_{i}) \otimes \, \tau_{j}\\
&\quad \times j^{(u)} \frac{\d}{\d j^{(u)}} \left\{ 8 \left( 2j^{(u)}+1 \right) \left[ \left( \frac{(2j^{(u)}+1)^{2}}{8} + \frac{(2j^{(d)}+1)^{2}}{8} \right)^{k+1} - \left( \frac{(2j^{(d)}+1)^{2}}{8} \right)^{k+1} \right] \right\} \\
&= \frac{c}{2j^{(u)}(2j^{(u)}+1)(2j^{(u)}+2)} \, \sum_{k=0}^{\infty} \frac{1}{(2k+2)!} \left( \frac{c}{2} \right)^{2k} \, \kappa^{ij} \, \pi^{(j^{(u)})}(\hatEu_{i}) \otimes \,  \tau_{j}\\
&\quad \times j^{(u)} \left\{ 16 \, \left[ \left( \frac{(2j^{(u)}+1)^{2}}{8} + \frac{(2j^{(d)}+1)^{2}}{8} \right)^{k+1} - \left( \frac{(2j^{(d)}+1)^{2}}{8} \right)^{k+1} \right] \right. \\
&\qquad + \left. 8\, \left( 2j^{(u)}+1 \right) \, \left( k+1 \right) \, \left( \frac{(2j^{(u)}+1)^{2}}{8} + \frac{(2j^{(d)}+1)^{2}}{8} \right)^{k} \, \frac{(2j^{(u)}+1)}{2} \right\}
\end{split}
\end{equation}
We have
\begin{equation}
\begin{split}
c \, &\sum_{k=0}^{\infty} \frac{1}{(2k+1)!} \left( \frac{c}{2} \right)^{2k} \kappa^{ij} \, \sum_{m=0}^{k} \binom{k}{m} \, Q_{DK}[||\Eu||^{2m} \Eu_i] \Biggl\vert_{\text{summand } \# 1} \, Q_{DK}[||\Ed||^{2(k-m)}] \; \otimes \tau_{j} \\
&= \frac{2}{\ju + 1} \, \kappa^{mn} \, \pi^{(j^{(u)})}(\hatEu_{m}) \otimes \,  \tau_{n}\\
&\qquad \times \left[ \frac{\cosh \left( \frac{c}{2} \sqrt{\frac{(2\ju + 1)^2}{8} + \frac{(2\jd + 1)^2}{8}} \right) - \cosh \left( \frac{(2\jd + 1)c}{4\sqrt{2}} \right) }{\frac{(2\ju+1)c}{8}} + \frac{2\ju + 1}{2} \frac{\sinh \left( \frac{c}{2} \sqrt{\frac{(2\ju + 1)^2}{8} + \frac{(2\jd + 1)^2}{8}} \right)}{\sqrt{\frac{(2\ju + 1)^2}{8} + \frac{(2\jd + 1)^2}{8}}} \right]
\end{split}
\end{equation}
and, using again that $j^{(u)} = j^{(d)} = j$, we obtain
\begin{equation}
\begin{split}
\frac{8c}{2j(2j+1)(2j+2)} \, &\sum_{k=0}^{\infty} \frac{1}{(2k+2)!} \left( \frac{c}{2} \right)^{2k} \, \kappa^{ij} \, \pi^{(j)}(\hatEu_{i}) \otimes \, \tau_{j}\\
&\quad \times j \left\{ 2 \, \left( \frac{(2j+1)^{2}}{4} \right)^{k+1} - 2 \, \left( \frac{(2j+1)^{2}}{8} \right)^{k+1} + \left( k+1 \right) \frac{(2j+1)^{2}}{2} \left( \frac{(2j+1)^{2}}{4} \right)^{k} \right\} \\
= &\frac{8c}{(2j+1)(2j+2)} \, \sum_{k=0}^{\infty} \frac{1}{(2k+2)!} \left( \frac{c}{2} \right)^{2k} \left[ \frac{(2j+1)}{2} \right]^{2k+2} \, \kappa^{ij} \, \pi^{(j)}(\hatEu_{i}) \otimes \, \tau_{j}\\
&~- \frac{8c}{(2j+1)(2j+2)} \, \sum_{k=0}^{\infty} \frac{1}{(2k+2)!} \left( \frac{c}{2} \right)^{2k} \left[ \frac{(2j+1)}{2\sqrt{2}} \right]^{2k+2} \, \kappa^{ij} \, \pi^{(j)}(\hatEu_{i}) \otimes \, \tau_{j}\\
&~+ \frac{4c}{(2j+1)(2j+2)} \, \sum_{k=0}^{\infty} \frac{1}{(2k+1)!} \left( \frac{c}{2} \right)^{2k} \left[ \frac{(2j+1)}{2} \right]^{2k+2} \, \kappa^{ij} \, \pi^{(j)}(\hatEu_{i}) \otimes \,  \tau_{j}\\
= &\frac{2}{j+1} \left\{ 2 \, \frac{\cosh \left( \frac{(2j+1)c}{4} \right) - 1}{\frac{(2j+1)c}{4}} - \sqrt{2} \, \frac{\cosh \left( \frac{(2j+1)c}{4\sqrt{2}} \right) - 1}{\frac{(2j+1)c}{4\sqrt{2}}} + \sinh \left( \frac{(2j+1)c}{4} \right) \right\} \, \kappa^{ij} \, \pi^{(j)}(\hatEu_{i}) \otimes \,  \tau_{j} \, .
\end{split} 
\end{equation}
for the first contribution. Turning our attention to the second term in \eqref{eqn:DK-map_non-gauge-inv_terms}, we immediately note that it vanishes if either $\ju = 0$ or $\ju = \sfrac{1}{2}$. For higher spins, we get
\begin{equation}
\begin{split}
c \, \sum_{k=0}^{\infty} \frac{1}{(2k+1)!} \left( \frac{c}{2} \right)^{2k} &\kappa^{ij} \, \sum_{m=0}^{k} \binom{k}{m} \, Q_{DK}(||E^{(u)}||^{2m} E_{i}^{(u)}) \Biggl\vert_{\text{summand } \# 2} \, Q_{DK}(||E^{(d)}||^{2(k-m)}) \; \tau_{j} \\
&= - c \, \sum_{k=0}^{\infty} \frac{1}{(2k+1)!} \left( \frac{c}{2} \right)^{2k} \, \sum_{m=0}^{k} \binom{k}{m} \left[ \frac{(2j^{(d)}+1)^{2}}{8} \right]^{k-m} \, \kappa^{ij} \, \pi^{(j^{(u)})}(\widehat{E^{(u)}}_{i}) \otimes \,  \tau_{j} \\
&\quad \times \frac{2}{8^{m}} \frac{1}{2j^{(u)}(2j^{(u)}+1)(2j^{(u)}+2)} \frac{2m+3}{2m+2} \, \sum_{l=1}^{\lfloor j^{(u)} \rfloor} \left( 2l \right)^{2m+2} \, \kappa^{ij} \, \pi^{(j^{(u)})}(\widehat{E^{(u)}}_{i}) \otimes \,  \tau_{j} \\
&= - \, \frac{2c}{2j^{(u)}(2j^{(u)}+1)(2j^{(u)}+2)} \, \sum_{l=1}^{\lfloor j^{(u)} \rfloor} \, \sum_{k=0}^{\infty} \frac{1}{(2k+1)!} \left( \frac{c}{2} \right)^{2k} \, \kappa^{ij} \, \pi^{(j^{(u)})}(\widehat{E^{(u)}}_{i}) \otimes \,  \tau_{j} \\
&\quad \times \sum_{m=0}^{k} \binom{k}{m} \left[ \frac{(2j^{(d)}+1)^{2}}{8} \right]^{k-m} \, \frac{2m+3}{2m+2} \, \frac{1}{8^{m}} \, \left( 2l \right)^{2m+2} \\
&= - \, \frac{c}{2j^{(u)}(2j^{(u)}+1)(2j^{(u)}+2)} \, \sum_{l=1}^{\lfloor j^{(u)} \rfloor} \, \sum_{k=0}^{\infty} \frac{1}{(2k+2)!} \left( \frac{c}{2} \right)^{2k} \, \kappa^{ij} \, \pi^{(j^{(u)})}(\widehat{E^{(u)}}_{i}) \otimes \,  \tau_{j} \\
&\quad \times \sum_{m=0}^{k} \binom{k+1}{m+1} \left[ \frac{(2j^{(d)}+1)^{2}}{8} \right]^{k+1-(m+1)} \, \frac{8}{8^{m+1}} \frac{\d}{\d l} \left( 2l \right)^{2m+3} \\
&= - \, \frac{8c}{2j^{(u)}(2j^{(u)}+1)(2j^{(u)}+2)} \, \sum_{l=1}^{\lfloor j^{(u)} \rfloor} \, \sum_{k=0}^{\infty} \frac{1}{(2k+2)!} \left( \frac{c}{2} \right)^{2k} \, \kappa^{ij} \, \pi^{(j^{(u)})}(\widehat{E^{(u)}}_{i}) \otimes \,  \tau_{j} \\
&\quad \times \frac{\d}{\d l} \left( 2l \right) \sum_{m=1}^{k+1} \binom{k+1}{m} \left[ \frac{(2j^{(d)}+1)^{2}}{8} \right]^{k+1-m} \, \frac{1}{8^{m}} \left( 2l \right)^{2m} \\
&= - \, \frac{8c}{2j^{(u)}(2j^{(u)}+1)(2j^{(u)}+2)} \, \sum_{l=1}^{\lfloor j^{(u)} \rfloor} \, \sum_{k=0}^{\infty} \frac{1}{(2k+2)!} \left( \frac{c}{2} \right)^{2k} \, \kappa^{ij} \, \pi^{(j^{(u)})}(\widehat{E^{(u)}}_{i}) \otimes \,  \tau_{j} \\
&\quad \times \frac{\d}{\d l} \left( 2l \right) \left\{ \left[ \frac{(2j^{(d)}+1)^{2}}{8} + \frac{l^{2}}{2} \right]^{k+1} - \left[ \frac{(2j^{(d)}+1)^{2}}{8} \right]^{k+1} \right\} \, .
\end{split}
\end{equation}

Using once again that $j^{(u)} = j^{(d)} = j$, the second contribution simplifies to
\begin{equation}
\begin{split}
- \frac{16}{2j(2j+1)(2j+2)} \, &\kappa^{ij} \, \pi^{(j)}(\widehat{E^{(u)}}_{i}) \otimes \, \tau_{j} \\
&\quad \times \sum_{l=1}^{\lfloor j \rfloor} \frac{\d}{\d l} \left\{ 2l \, \left[ \frac{\cosh \left( \sqrt{\frac{(2j+1)^{2}}{8} + \frac{l^{2}}{2}} \frac{c}{2} \right) - 1}{\frac{c}{2}} - \frac{\cosh \left( \frac{(2j+1)c}{4\sqrt{2}} \right) - 1}{\frac{c}{2}} \right] \right\}\\
&= - \frac{32}{2j(2j+1)(2j+2)} \, \kappa^{ij} \, \pi^{(j)}(\widehat{E^{(u)}}_{i}) \otimes \, \tau_{j} \\
&\quad \times \sum_{l=1}^{\lfloor j \rfloor} \, \left[ \frac{\cosh \left( \sqrt{\frac{(2j+1)^{2}}{8} + \frac{l^{2}}{2}} \frac{c}{2} \right) - \cosh \left( \frac{(2j+1)c}{4\sqrt{2}} \right)}{\frac{c}{2}} + \frac{\frac{l^{2}}{2} \sinh \left( \sqrt{\frac{(2j+1)^{2}}{8} + \frac{l^{2}}{2}} \frac{c}{2} \right)}{\sqrt{\frac{(2j+1)^{2}}{8} + \frac{l^{2}}{2}}} \right]
\end{split}
\end{equation}

\bibliography{bibliography}

\end{document}